\documentclass[english,runningheads,11pt]{llncs}
\usepackage{graphicx,amssymb,amsmath,amssymb}

\usepackage[in]{fullpage}
\usepackage[top=0.8in, bottom=0.9in, left=0.9in, right=0.9in]{geometry}

\def\calP{\mathcal{P}}

\def\calD{\mathcal{D}}
\def\calV{\mathcal{V}}

\def\calI{\mathcal{I}}
\def\vector{\overrightarrow}
\def\hats{\hat{s}}
\def\hatt{\hat{t}}
\def\bbR{\mathbb{R}}
\def\st{$s$-$t$}
\def\spmed{\calD_{spm}}
\def\st{$s$-$t$}
\def\piR{R_{\pi}}
\newcommand{\spm}{\mbox{$S\!P\!M$}}

\newenvironment{proof}{\noindent {\textbf{Proof:}}\rm}{\hfill $\Box$\rm}
\newtheorem{observation}{Observation}

\begin{document}

\title{On the Geodesic Centers of Polygonal Domains\thanks{An extended-abstract of this paper will appear in the Proceedings of the 24th Annual European Symposium on Algorithms (ESA 2016).}}

\author{
Haitao Wang
}

\institute{
Department of Computer Science\\
Utah State University, Logan, UT 84322, USA\\
\email{haitao.wang@usu.edu}
}

\maketitle

\pagenumbering{arabic}
\setcounter{page}{1}

\vspace*{-0.2in}
\begin{abstract}
In this paper, we study the problem of computing Euclidean geodesic centers of a
polygonal domain $\calP$ with a total of $n$ vertices.
We discover many interesting observations.
We give a necessary condition for a point being a geodesic center.
We show that there is at most one geodesic center
among all points of $\calP$ that have topologically-equivalent shortest path maps.
This implies that the total number of geodesic centers is bounded by
the combinatorial size of the shortest path map equivalence decomposition of
$\calP$, which is known to be $O(n^{10})$.
One key observation is a {\em $\pi$-range property} on shortest path lengths when points are moving.
With these observations, we propose an algorithm that computes all geodesic centers in $O(n^{11}\log n)$ time. Previously, an algorithm of $O(n^{12+\epsilon})$ time was known for this problem, for any $\epsilon>0$.
\end{abstract}


\section{Introduction}
\label{sec:intro}

Let $\calP$ be a polygonal domain with a total of $h$ holes and $n$ vertices,
i.e., $\calP$ is a multiply-connected region whose boundary is a union of $n$
line segments, forming $h+1$ closed polygonal cycles. A simple polygon
is a special case of a polygonal domain with $h=0$. For any two
points $s$ and $t$, a {\em shortest path} or {\em geodesic path} from $s$ to $t$ is a path in $\calP$ whose Euclidean length is minimum among all paths from $s$ to $t$ in
$\calP$; we let $d(s,t)$ denote the Euclidean
length of any shortest path from $s$ to $t$ and we also say that $d(s,t)$
is the {\em shortest path distance or geodesic distance} from $s$ to $t$.

A point $s$ is a {\em geodesic center} of $\calP$ if $s$ minimizes the value
$\max_{t\in \calP}d(s,t)$, i.e., the maximum
geodesic distance from $s$ to all points of $\calP$.
In this paper, we study the problem of computing the geodesic centers of $\calP$.

The problem in simple polygons has been well studied. It is
known that for any point in a simple polygon, its farthest point must
be a vertex of the polygon \cite{ref:SuriCo89}. It has been shown
that the geodesic center in a simple polygon is unique and has at least two
farthest points \cite{ref:PollackCo89}. Due to
these helpful observations, efficient algorithms for finding geodesic
centers in simple polygons have been developed. Asano and
Toussaint \cite{ref:AsanoCo85} first gave an $O(n^4\log n)$ time
algorithm for the problem, and later Pollack, Sharir, and Rote
\cite{ref:PollackCo89} solved the problem in $O(n\log n)$ time.
It had been an open problem whether the problem is solvable in linear time until recently
Ahn et al. \cite{ref:AhnA15} presented a linear-time algorithm for it.

Finding a geodesic center in a polygonal domain $\calP$ is much more
difficult. This is partially due to that a farthest point of a
point in $\calP$ may be in the interior of $\calP$ \cite{ref:BaeTh13}.
Also, it is easy to construct an example where the geodesic center of
$\calP$ is not unique (e.g., see Fig.~\ref{fig:unique}). Bae,
Korman, and Okamoto \cite{ref:BaeCo14CCCG} gave the first-known algorithm
that can compute a geodesic center in $O(n^{12+\epsilon})$ time for
any $\epsilon>0$. They first showed that for any point its farthest points must
be vertices of its shortest path map in $\calP$. Then, they considered
the shortest path map equivalence decomposition (or SPM-equivalence
decomposition) \cite{ref:ChiangTw99}, denoted by $\spmed$; for each cell of
$\spmed$, they computed the upper envelope of $O(n)$ graphs in
three-dimensional space, which takes $O(n^{2+\epsilon})$ time
\cite{ref:HalperinNe94}, to search a
geodesic center in the cell. Since the size of $\spmed$ is $O(n^{10})$
\cite{ref:ChiangTw99}, their algorithm runs in $O(n^{12+\epsilon})$ time.

\begin{figure}[t]
\begin{minipage}[t]{0.28\linewidth}
\begin{center}
\includegraphics[totalheight=1.2in]{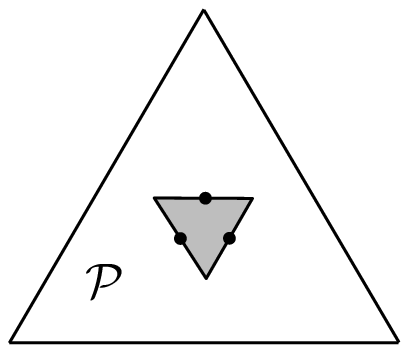}
\caption{\footnotesize The boundary of $\calP$ consists of an outer and an inner
equilateral triangles with their geometric centers co-located. Each of the three thick points
is a geodesic center of $\calP$.
}
\label{fig:unique}
\end{center}
\end{minipage}
\hspace{0.03in}
\begin{minipage}[t]{0.70\linewidth}
\begin{center}
\includegraphics[totalheight=1.7in]{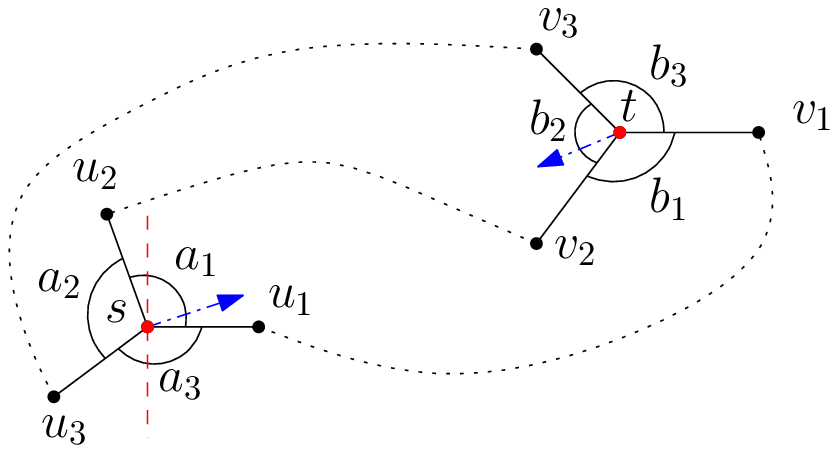}
\caption{\footnotesize Illustrating the $\pi$-range property.
Suppose there are three shortest \st\ paths through
vertices $u_i$ and $v_i$ with $i=1,2,3$, respectively. If $s$ and $t$ move along the blue
arrows simultaneously (possibly with different ``speeds''), then all three
shortest paths strictly decreases (it is difficult to tell whether this is true
from the figure, so these two blue directions here are only for
illustration purpose).  The special case happens when the six
angles $a_i$ and $b_i$ for $i=1,2,3$ satisfy $a_i=b_i$ for $i=1,2,3$.
}
\label{fig:rangeproperty}
\end{center}
\end{minipage}
\vspace*{-0.15in}
\end{figure}


A concept closely related to the geodesic center is the {\em geodesic diameter},
which is the maximum geodesic distance over all pairs of points in $\calP$,
i.e., $\max_{s,t,\in \calP}d(s,t)$. In simple polygons,
due to the property that there always exists a pair of vertices of $\calP$ whose geodesic
distance is equal to the geodesic diameter, efficient algorithms have been given
for computing the geodesic diameters. Chazelle
\cite{ref:ChazelleA82} gave the first algorithm that runs in $O(n^2)$ time.
Later, Suri \cite{ref:SuriCo89} presented an $O(n\log n)$-time algorithm.
Finally, the problem was solved in $O(n)$ time by Hershberger and Suri
\cite{ref:HershbergerMa97}.
Computing the geodesic diameter in a polygonal domain $\calP$ is much more
difficult. It was shown in \cite{ref:BaeTh13}
that the geodesic diameter can be realized by two points
in the interior of $\calP$, in which case there are at least five
distinct shortest paths between the two points. As for the geodesic center problem, this
makes it difficult to discretize the search space. By an exhaustive-search method,
Bae, Korman, and Okamoto
\cite{ref:BaeTh13} gave
the first-known algorithm for computing the diameter of $\calP$ and the
algorithm runs in $O(n^{7.73})$ or $O(n^7(\log n+h))$ time.

Refer to \cite{ref:BaeCo15,ref:DjidjevAn92,ref:KeAn89,ref:MitchellGe00,ref:NilssonCo91,ref:NilssonAn96,ref:SchuiererAn96} for
other variations of geodesic diameter and center problems (e.g., the $L_1$
metric and the link distance case).

\subsection{Our Contributions}

We conduct a ``comprehensive'' study on geodesic centers of $\calP$.
We discover many interesting observations, and some of them
are even surprising.  For example, we show
that even if a geodesic center is in the interior of $\calP$,
it may have only one farthest point, which
is somewhat counter-intuitive.  We give a necessary condition for
a point being a geodesic center. We also show that
there is at most one geodesic center
among all points of $\calP$ that have topologically-equivalent
shortest path maps in $\calP$.  This immediately implies that the interior of
each cell or each edge of the SPM-equivalence
decomposition $\spmed$ can contain at most one geodesic center, and thus, the
total number of geodesic centers of $\calP$ is bounded by the
combinatorial size of $\spmed$, which is known to be $O(n^{10})$
\cite{ref:ChiangTw99}.  Previously, the only known upper bound on the
total number of geodesic centers of $\calP$
is $O(n^{12+\epsilon})$, which is implied by the
algorithm in \cite{ref:BaeCo14CCCG}.

These observations are all more or less due to an interesting
observation, which we call the {\em $\pi$-range property} and is one key
contribution of this paper. Here
we demonstrate an application of the $\pi$-range property (see the paper for the details). Let $s$ and $t$ be two points in the interior of $\calP$ such that
 $t$ is a farthest point of $s$ in $\calP$.
Refer to Fig.~\ref{fig:rangeproperty} for an example. Suppose there are
three shortest paths from $s$ to $t$ as shown in Fig.~\ref{fig:rangeproperty}.
The $\pi$-range property says that unless a special case happens, there exists an
open range of exactly size $\pi$ (e.g., delimited by the right open half-plane
bounded by the vertical line through $s$ in Fig.~\ref{fig:rangeproperty})
such that if $s$ moves along any direction in the
range for an infinitesimal distance, we can always find a direction to move $t$
such that the lengths of {\em all three} shortest paths strictly decrease.
Further, if the special case does not happen,
we can explicitly determine the above range of size $\pi$. In fact, it is the
special case that makes it possible for a geodesic center having only one
farthest point.


With these observations, we propose an exhaustive-search algorithm to
compute a set $S$ of {\em candidate points} such that all geodesic
centers must be in $S$. For example, refer to
Fig.~\ref{fig:candidate} for a schematic diagram,
where a geodesic center $s$ has three farthest points
$t_1,t_2,t_3$ and all these four points are in the interior of $\calP$.
The nine shortest paths
from $s$ to $t_1,t_2,t_3$ provide a system of eight equations, which give eight
(independent) constraints that can
determine the four points $s,t_1,t_2,t_3$ if we consider the coordinates of these points as eight variables. This suggests our exhaustive-search approach to compute
candidate points for such a geodesic center $s$ (similar exhaustive-search
approaches are also used before, e.g., \cite{ref:BaeTh13,ref:ChiangTw99}).
However, if a geodesic center $s$ has only one
farthest point (e.g., Fig.~\ref{fig:rangeproperty}), then we have only three
shortest paths, which give only two constraints. In order to determine $s$ and $t$,
which have four variables, we need two more constraints. It turns out the
$\pi$-range property (i.e., the special case) provides exactly two more
constraints (on the angles as shown in Fig.~\ref{fig:rangeproperty}).
In this way, we can still compute candidate points for such $s$.
Also, if a geodesic center has two farthest points, we will need one more
constraint, which is also provided by the $\pi$-range property (the non-special case).
Note that the previous exhaustive-search
approaches \cite{ref:BaeTh13,ref:ChiangTw99} do not need the $\pi$-range property.

\begin{figure}[t]
\begin{minipage}[t]{\linewidth}
\begin{center}
\includegraphics[totalheight=2.0in]{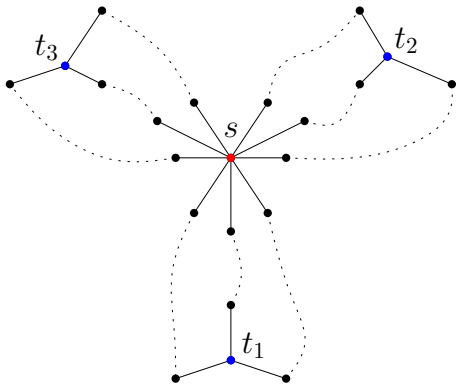}
\caption{\footnotesize Illustrating a geodesic center $s$ with three farthest
points $t_1,t_2, t_3$ such that all these four points are in the interior of
$\calP$. There are three shortest paths from $s$ to each of $t_1,t_2,t_3$.
}
\label{fig:candidate}
\end{center}
\end{minipage}
\vspace*{-0.15in}
\end{figure}

The number of candidate points in $S$ is $O(n^{11})$.
To find all geodesic centers from $S$, a
straightforward solution is to compute the shortest path map for every
point of $S$, which takes $O(n^{12}\log n)$ time in total.
Again, with the help of the $\pi$-range property, we propose a {\em pruning
algorithm} to eliminate most points from $S$ in $O(n^{11}\log n)$ time
such that none of the eliminated points is a geodesic center and the number of
the remaining points of $S$ is only $O(n^{10})$. Consequently, we can find all
geodesic centers in additional $O(n^{11}\log n)$ time by computing
their shortest path maps.

Although we improve the previous $O(n^{12+\epsilon})$ time algorithm in
\cite{ref:BaeCo14CCCG} by a factor of roughly $n^{1+\epsilon}$, the running time is still huge. We feel that our observations (in particular, the $\pi$-range property) may be more interesting than the algorithm
itself. We suspect that these observations
may also find applications in other related problems.
The paper is lengthy and some discussions are quite tedious, which is
mainly due to a considerable number of cases depending on whether a geodesic center and
its farthest points are in the interior, on an edge, or at vertices of
$\calP$, although the essential idea is quite similar for all these
cases.



The rest of the paper is organized as follows. In Section \ref{sec:pre}, we
introduce notation and review some concepts. In Section \ref{sec:obser},
we give our observations.
In particular, we prove the $\pi$-range property in Section \ref{sec:range}.
Our algorithm for computing the candidate points is presented
in Section \ref{sec:candidates}. Finally, we find all
geodesic centers from the candidate points in Section \ref{sec:center}.

\section{Preliminaries}
\label{sec:pre}

Consider any point $s\in \calP$. Let $d_{\max}(s)$ denote the maximum
geodesic distance from $s$ to all points of $\calP$, i.e.,
$d_{\max}(s)=\max_{t\in \calP}d(s,t)$.
A point $t\in \calP$ is a {\em farthest point} of $s$ if $d(s,t)=d_{\max}(s)$.
We use $F(s)$ to denote the set of all farthest points of $s$ in $\calP$.
For any two points $p$ and $q$ in $\calP$, for convenience of
discussions, we say that $p$ is {\em visible} to $q$ if the line
segment $\overline{pq}$ is in $\calP$ and the interior of $\overline{pq}$ does
not contain any vertex of $\calP$.
We use $|pq|$ to denote the (Euclidean) length of any line segment $\overline{pq}$.
Note that two points $s$ and $t$ in $\calP$ may have more than one shortest path
between them, and if not specified, we use $\pi(s,t)$ to denote any such
shortest path.

For simplicity of discussion, we make a general position assumption that any two
vertices of $\calP$ have only one shortest path and no three vertices of $\calP$
are on the same line.


Denote by $\calI$ the set of all interior points of $\calP$, $\calV$ the set
of all vertices of $\calP$, and $E$ the set of all relatively interior points on
the edges of $\calP$ (i.e., $E$ is the boundary of $\calP$ minus $\calV$).

\paragraph{Shortest path maps.}
Given a point $s\in \calP$, a {\em shortest path map} of $s$
\cite{ref:ChiangTw99}, denoted by
$\spm(s)$, is a decomposition of $\calP$ into regions (or cells) such that in each
cell $\sigma$, the combinatorial structures of shortest paths from $s$ to
all points $t$ in $\sigma$ are the same, and more specifically, the sequence of obstacle
vertices along $\pi(s,t)$ is fixed for all $t$ in $\sigma$.
Further, the {\em root} of $\sigma$, denoted by $r(\sigma)$,
 is the last vertex of $\calV\cup \{s\}$ in
the path $\pi(s,t)$  for any point $t\in \sigma$
(hence $\pi(s,t)=\pi(s,r(\sigma))\cup \overline{r(\sigma)t}$;
note that $r(\sigma)$ is $s$ if $s$ is visible to $t$).
As in \cite{ref:ChiangTw99},
we classify each edge of $\sigma$ into three types: a portion of
an edge of $\calP$, {\em an extension segment}, which is a line segment extended
from $r(\sigma)$ along the opposite direction from  $r(\sigma)$ to the
vertex of $\pi(s,t)$ preceding $r(\sigma)$, and {\em a bisector
curve/edge} that is
a hyperbolic arc. For each point $t$ in
a bisector edge of $\spm(s)$, $t$ is on the common boundary of two cells and
there are two shortest paths from $s$ to $t$ through the roots of the
two cells, respectively (and neither path contains both roots). The
{\em vertices} of $\spm(s)$
include $\calV\cup \{s\}$ and all intersections of edges of $\spm(s)$.
If a vertex $t$ of $\spm(s)$ is an intersection of two or more
bisector edges, then there
are more than two shortest paths from $s$ to $t$.
The map
$\spm(s)$ has $O(n)$ vertices, edges, and cells, and can be computed
in $O(n\log n)$ time \cite{ref:HershbergerAn99}.
It was shown \cite{ref:BaeCo14CCCG}
that any farthest point of $s$ in $\calP$ must be a vertex of $\spm(s)$.

For differentiation, we will refer to the vertices of $\calV$ as {\em polygon
vertices} and refer to the edges of $E$ as {\em polygon edges}.

The SPM-equivalence decomposition $\spmed$ of $\calP$ \cite{ref:ChiangTw99}
is a subdivision of
$\calP$ into regions such that for all points $s$ in the interior of the same
region
or edge of $\spmed$, the shortest path maps $\spm(s)$ are topologically
equivalent. 
Chiang and
Mitchell \cite{ref:ChiangTw99} showed that the combinatorial complexity of
$\spmed$ is bounded by $O(n^{10})$ and $\spmed$
can be computed in $O(n^{11})$ time.

\paragraph{Directions and ranges.}
In this paper, we will have intensive discussions on moving points
along certain directions.
For any direction $r$, we represent $r$ by the angle $\alpha(r)\in [0,2\pi)$
counterclockwise from the positive direction of the $x$-axis.
For convenience, whenever we are talking about
an angle $\alpha$, unless otherwise specified,
depending on the context we may refer to any angle $\alpha+2\pi\cdot
k$ for $k\in \mathbb{Z}$.
For any two angles $\alpha_1$ and $\alpha_2$ with
$\alpha_1\leq \alpha_2 < \alpha_1+2\pi$, the interval
$[\alpha_1,\alpha_2]$ represents a {\em direction range} that includes
all directions whose angles are in $[\alpha_1,\alpha_2]$, and
$\alpha_2-\alpha_1$ is called the {\em size} of the range.
Note that the range can be open (e.g., $(\alpha_1,\alpha_2)$) and the size of any
direction range is no more than $2\pi$.

Consider a half-plane $h$ whose bounding line is through a point $s$
in the plane. We say $h$
{\em delimits} a range of size $\pi$ of directions for $s$ that consists of all
directions along which $s$ will move towards inside $h$. If $h$ is an open
half-plane, then the range is open as well.

A direction $r$ for $s\in \calP$ is called a {\em free direction} of $s$ if we move $s$
along $r$ for an infinitesimal distance then $s$ is still in $\calP$.
We use $R_f(s)$ to denote the range of all free directions of $s$.
Clearly, if $s\in \calI$, $R_f(s)$ contains all directions;
if $s\in E$, $R_f(s)$ is a (closed) range of size $\pi$; if $s\in
V$, $R_f(s)$ is delimited by the two
incident polygon edges of $s$.


\section{Observations}
\label{sec:obser}

Consider any point $s\in \calP$ and let $t$ be any farthest point of
$s$. Recall that $t$ is a vertex of $\spm(s)$ \cite{ref:BaeCo14CCCG}.
Suppose we move $s$ infinitesimally along a free direction $r$ to a new point
$s'$. Since $|ss'|$ is infinitesimal, we can assume that
$s$ and $s'$ are in the same cell $\sigma$ of $\spmed$.
Further, if $s$ is in the interior of $\sigma$, then $s'$ is also in the interior
of $\sigma$.

Regardless of whether $s$ is in the interior of $\sigma$ or not,
there is a vertex $t'\in \spm(s')$ {\em corresponding to}
the vertex $t$ of $\spm(s)$ in the following sense \cite{ref:ChiangTw99}:
If the line segment $\overline{s't'}$ is a shortest path from $s'$ to $t'$, then
$\overline{st}$ is a shortest path from $s$ to $t$; otherwise,
if $s',u_1,u_2,\ldots,u_k,t'$ is the sequence of the vertices of
$\calV\cup\{s',t'\}$ in a shortest path from $s'$ to $t'$, then
$s,u_1,u_2,\ldots,u_k,t$ is also the sequence of the vertices of
$\calV\cup\{s,t\}$ in a shortest path from $s$ to $t$.

In the case that $s$ is on the boundary of $\sigma$ while $s'$ is in the
interior of $\sigma$,
there might be more than one such vertex $t'\in \spmed$ corresponding to
$t$ (refer to \cite{ref:ChiangTw99} for the details)
and we use $M_t(s')$ to denote the set of all such vertices $t'$. We
should point out that although a vertex in $\spm(s)$ may correspond to more than one
vertex in $\spm(s')$, any vertex in $\spm(s')$ can correspond to one and only
one vertex in $\spm(s)$ (because $s'$ is always in the interior of $\sigma$).

We introduce the following definition which is crucial to the paper.

\begin{definition}
A free direction $r$ is called an {\em admissible direction} of $s$ with
respect to $t$ if as we move $s$ infinitesimally along $r$ to
a new point $s'$, $d(s',t')<d(s,t)$ holds for each $t'\in M_t(s')$.
\end{definition}


For any $t\in F(s)$, let $R(s,t)$ denote the set of all admissible directions of
$s$ with respect to $t$; let $R(s)=\bigcap_{t\in F(s)}R(s,t)$. The following Lemma
\ref{lem:10}, which gives a necessary condition for a point being a geodesic center of
$\calP$, explains why we consider admissible directions.

Before proving Lemma \ref{lem:10}, we introduce some notation and
Observation \ref{obser:10}.

Consider any two points $s$ and $t$ in $\calP$.
Suppose the vertices of $\calV\cup\{s,t\}$ along a shortest \st\ path $\pi(s,t)$ are
$s=u_0,u_1,\ldots,u_k=t$. According to our definition on the
``visibility'', $s$ is visible to $t$ if and only if $k=1$.
If $s$ is not visible to $t$, then $k\neq 1$ and we call $u_1$
an {\em s-pivot} and $u_{k-1}$ a {\em t-pivot} of $\pi(s,t)$.
It is possible that there are multiple shortest paths between $s$ and $t$, and
thus there might be multiple $s$-pivots and $t$-pivots for
$(s,t)$. We use $U_s(t)$ and $U_t(s)$ to denote the sets of all
$s$-pivots and $t$-pivots for $(s,t)$, respectively.
Note that according to our above definition, for any
$u\in U_s(t)$, the line segment $\overline{su}$ does not contain any
polygon vertex in its interior.

We have the following observation. Similar results have been given in
\cite{ref:BaeTh13}.

\begin{observation}\label{obser:10}
Suppose $t$ is a farthest point of a point $s$.
\begin{enumerate}
\item
If $t$ is in $\calI$, then $|U_t(s)|\geq 3$ and $t$ must
be in the interior of the convex hull of the vertices of $U_t(s)$.
\item
If $t$ is in $E$, say, $t\in e$ for a polygon edge $e$ of $E$,
then $|U_t(s)|\geq 2$ and $U_t(s)$ has at least one vertex in the
open half-plane bounded by the supporting line of $e$ and containing the interior
of $\calP$ in the small neighborhood of $e$.  Further, $U_t(s)$ has at least one
vertex in each of the two open half-planes bounded by the line
through $t$ and perpendicular to $e$.
\end{enumerate}
\end{observation}
\begin{proof}
The proof is similar to those in \cite{ref:BaeTh13}. The main
idea is that if the conditions are
not satisfied, then we can always find a point $t'$ further to $s$ than $t$,
incurring contradiction.
%
\end{proof}

\begin{lemma}\label{lem:10}
If $s$ is a geodesic center of $\calP$, then $R(s)=\emptyset$.
\end{lemma}
\begin{proof}
Assume to the contrary that $R(s)\neq \emptyset$. Let $r$ be any direction in $R(s)$. Then, $r$ is in $R(s,t)$ for each $t\in F(s)$. In other words, $r$ is an
admissible direction of $s$ with respect to each $t\in F(s)$.
Suppose we move $s$ infinitesimally along $r$ to a new point $s'$. In the
following, we show that $d_{\max}(s')<d_{\max}(s)$, which contradicts with that
$s$ is a geodesic center.

Consider any $t'\in F(s')$, i.e., $t'$ is a farthest point of $s'$.
To prove $d_{\max}(s')<d_{\max}(s)$, it is sufficient to show that $d(s',t')<d_{\max}(s)$.

If $t'$ is in $M_t(s')$ for any $t\in F(s)$, then
since $r$ is an admissible direction of $s$ with respect to $t$, it holds that
$d(s',t')<d(s,t)=d_{\max}(s)$. In the following, we assume
$t'$ is not in $M_t(s')$ for any $t\in F(s)$.

Let $V(s)$ be the set of all vertices of $\spm(s)$ that are not in $F(s)$.
Let $\delta = d_{\max}(s) - \max_{t\in V(s)}d(s,t)$. Note that since $|V(s)|$ is
finite, the value $\delta$ is well-defined. Also note that the value $\delta$ is
fixed and does not depend on $s'$. Since $s$ moves infinitesimally to $s'$, we can
assume $|ss'|<\delta$.

Since $t'$ is a farthest point of $s'$, $t'$ must be a vertex of $\spm(s')$. Let
$t^*$ be the vertex of $\spm(s)$ corresponding to $t'$. Note that although a
vertex of $\spm(s)$ may correspond to multiple vertices of $\spm(s')$, a vertex
of $\spm(s')$ corresponds to one and only one vertex of $\spm(s)$. Hence, the
vertex $t^*$ is unique, and further,
$U_{t'}(s')\subseteq U_{t^*}(s)$. Since $t'$ is not in $M_t(s')$ for any $t\in F(s)$, we
know that $t^*$ is not in $F(s)$ but in $V(s)$. Therefore, $d(s,t^*)\leq \max_{t\in
V(s)}d(s,t)=d_{\max}(s)-\delta$.

Since $\overline{ss'}\in \calP$, it holds that $d(s',t')\leq |s's|+d(s,t')$.
Therefore, if we can prove $d(s,t')\leq d(s,t^*)$, since
$|ss'|<\delta$ and $d(s,t^*)\leq d_{\max}(s)-\delta$, we can obtain $d(s',t')\leq
|s's|+d(s,t')<\delta+d_{\max}(s)-\delta=d_{\max}(s)$.
In the sequel, we prove $d(s,t')\leq d(s,t^*)$.

If $t'\in \calV$, then $t^*$ must be $t'$ \cite{ref:ChiangTw99}.
Hence, $d(s,t')= d(s,t^*)$ and thus  $d(s,t')\leq d(s,t^*)$ trivially
follows. In the following, we assume $t'\not\in \calV$. Thus, $t'$ is either in $E$ or
$\calI$.

Recall that $U_{t'}(s')\subseteq U_{t^*}(s)$. To prove
$d(s,t')\leq d(s,t^*)$, it is sufficient to find a vertex $u\in U_{t'}(s')$ such
that $|{ut'}|\leq |{ut^*}|$
because $d(s,t')=d(s,u)+|ut'|$ and $d(s,t^*)=d(s,u)+|ut^*|$.
To this end, we will make use of Observation \ref{obser:10}.

\begin{enumerate}
\item
If $t'\in E$, let $e$ be the polygon edge of $E$ that contains $t'$.
Since $t'$ is a farthest point of $s'$,
by Observation \ref{obser:10}, there must be a vertex of $U_{t'}(s')$ on either
open half-plane bounded by $l_{t'}$, where $l_{t'}$ is the line through $t'$ and
perpendicular to $e$.  Since $t'\in E$,
$t^*$ is either on $e$ or an endpoint of $e$ \cite{ref:ChiangTw99}. In either
case, one open
half-plane bounded by $l_{t'}$ contains $t^*$ and the other does not (e.g., see
Fig.~\ref{fig:necessary}).
Let $u$ be the
vertex of $U_{t'}(s')$ in the open half-plane bounded by $l_{t'}$ that does not
contain $t^*$. Clearly, $|{ut'}|\leq |{ut^*}|$ holds.

\begin{figure}[t]
\begin{minipage}[t]{0.49\linewidth}
\begin{center}
\includegraphics[totalheight=0.8in]{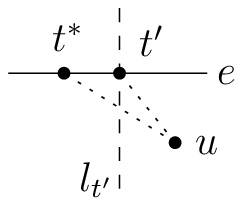}
\caption{\footnotesize Illustrating the proof of Lemma \ref{lem:10} for the case
$t'\in E$:
$|ut'|<|ut^*|$.
}
\label{fig:necessary}
\end{center}
\end{minipage}
\hspace*{0.02in}
\begin{minipage}[t]{0.49\linewidth}
\begin{center}
\includegraphics[totalheight=0.8in]{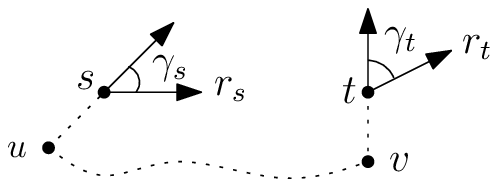}
\caption{\footnotesize Illustrating the definitions of the angles $\gamma_s$ and
$\gamma_t$.
}
\label{fig:angles}
\end{center}
\end{minipage}
\vspace*{-0.15in}
\end{figure}

\item
If $t'\in \calI$, then since $t'$ is a farthest point of $s'$,
by Observation \ref{obser:10}, $t'$ must be in the interior of the convex hull
of all vertices of $U_{t'}(s')$. 
Then, regardless of whatever position of $t^*$ is,
there must be a vertex $u\in U_{t'}(s')$ on the convex hull such that $|{ut'}|\leq |{ut^*}|$ holds.
\end{enumerate}

The lemma is thus proved.
\end{proof}

As explained in Section \ref{sec:intro}, we will compute candidate points for
geodesic centers. As a necessary condition, Lemma \ref{lem:10} will be
helpful for computing those candidate points.


Consider any point $s\in \calP$. Let $t$ be a farthest point of $s$.
We need to find a way to determine the admissible range $R(s,t)$. To this end, we will
give a sufficient condition for a direction being in $R(s,t)$.
We first assume that $s$ is not visible to $t$, and as will be seen later,
the other case is trivial.

Let $u$ and $v$ respectively be the $s$-pivot and the $t$-pivot of
$(s,t)$ in a shortest \st\ path $\pi(s,t)$. Clearly,
$d(s,t)=|{su}|+d(u,v)+|{vt}|$. We define
$d_{u,v}(s,t)=|{su}|+d(u,v)+|{vt}|$ as a function of $s\in \bbR^2$ and $t\in \bbR^2$.
Suppose we move $s$ along a free direction $r_s$ with
the unit speed and move $t$ along a free direction $r_t$ with a speed $\tau\geq 0$.
Let $\gamma_s$ denote the smaller angle between the following two rays
originated from $s$ (e.g., see Fig.~\ref{fig:angles}): one with
direction $r_s$ and one with direction from $u$ to
$s$. Similarly, let $\gamma_t$ denote the smaller angle between the following two rays
originated from $t$: one with direction $r_t$ and one with direction from $v$ to
$t$. In fact, as discussed in \cite{ref:BaeTh13},
if we consider $d(s,t)$ as a four-variate function,
the triple $(r_s,r_t,\tau)$ corresponds to a vector $\rho$ in
$\bbR^4$, and
the directional derivative of $d_{u,v}(s,t)$ at $(s,t)\in \bbR^4$ along
$\rho$, denoted by $d_{u,v}'(s,t)$, and the second directional
derivative of $d_{u,v}(s,t)$ at $(s,t)$ along $\rho$, denoted by $d_{u,v}''(s,t)$, are
\begin{equation}\label{equ:10}
d_{u,v}'(s,t)=\cos\gamma_s+\tau\cos\gamma_t,\ \ \ d_{u,v}''(s,t)=
\frac{\sin^2\gamma_s}{|{su}|}+\tau\cdot
\frac{\sin^2\gamma_t}{|{tv}|}.
\end{equation}
Since $\tau\geq 0$, $d_{u,v}''(s,t)\geq 0$ always holds.
Further, if $\tau\neq 0$, then
$d_{u,v}''(s,t)=0$ if and only if $\sin^2\gamma_s = \sin^2\gamma_t = 0$,
i.e., each of $\gamma_s$ and $\gamma_t$ is either $0$ or $\pi$.
In the following, in order to make the discussions more intuitive, we choose to
use the parameters $r_s$, $r_t$, and $\tau$ instead of the vectors of $\bbR^4$.

For each vertex $u\in U_s(t)$, there must be a vertex $v\in U_t(s)$ such that
the concatenation of $\overline{su}$, $\pi(u,v)$, and $\overline{vt}$ is a
shortest path from $s$ to $t$, and we call such a vertex $v$ a {\em coupled
$t$-pivot} of $u$ (if $u$ has more than one such vertex, then all of them
are coupled $t$-pivots of $u$). Similarly, for each vertex $v\in
U_t(s)$, we also define its coupled $s$-pivots in $U_s(t)$.

The following lemma provides a sufficient condition for a direction being an
admissible direction for $s$ with respect to $t$.

\begin{lemma}\label{lem:20}
Suppose $t$ is a farthest point of $s$ and $s$ is not visible to $t$.
\begin{enumerate}
\item
For $t\in \calI$,
a free direction $r_s$ is in $R(s,t)$ 
if there is a free direction $r_t$ for $t$ with a speed $\tau\geq 0$ such that when we
move $s$ along $r_s$ with the unit speed and move $t$ along $r_t$ with speed
$\tau$, each vertex $v\in U_t(s)$ has a coupled $s$-pivot $u$ with either
$d_{u,v}'(s,t)<0$, or $d_{u,v}'(s,t)=0$ and $d_{u,v}''(s,t)=0$.
\item
For $t\in E$, a free direction $r_s$ is in $R(s,t)$
if there is a free direction $r_t$ for $t$ that is parallel to the polygon edge of
$E$ containing $t$ with a speed $\tau\geq 0$ such that when we
move $s$ along $r_s$ with the unit speed and move $t$ along $r_t$ with speed
$\tau$, each vertex $v\in U_t(s)$ has a coupled $s$-pivot $u$ with either
$d_{u,v}'(s,t)<0$, or $d_{u,v}'(s,t)=0$ and $d_{u,v}''(s,t)=0$.
\item
For $t\in \calV$, a free direction $r_s$ is in $R(s,t)$
if we move $s$ along $r_s$ with the unit speed,
each vertex $v\in U_t(s)$ has a coupled $s$-pivot $u$ with either
$d_{u,v}'(s,t)<0$, or $d_{u,v}'(s,t)=0$ and $d_{u,v}''(s,t)=0$.
\end{enumerate}
\end{lemma}
\begin{proof}
Suppose we move $s$ infinitesimally along $r_s$ to $s'$. The point
$t$, as a vertex of $\spm(s)$, corresponds to a set $M_t(s')$ of vertices
in $\spm(s')$. To prove that $r_s$ is an admissible direction for $s$ with
respect to $t$, we need to show that $d(s',t')<d(s,t)$ for any $t'\in M_t(s')$.
In the following, we discuss the three cases depending on whether $s$ is in
$\calI$, $E$, or $\calV$.

We remark that the proof would be much simpler if we only considered the
``non-degenerate'' case
whether $s$ is in the interior of a cell of the SPM-equivalence decomposition
$\spmed$ (because in that case $M_t(s')$ has only one vertex).

\subsubsection*{The case $t\in \calI$.}
We first prove the case $t\in \calI$.
We begin with proving the following claim.

\noindent{\em Claim:  Suppose we move $s$ to $s'$ infinitesimally along a free
direction $r$; if there is a point $t^*$ in the interior of the convex hull of
the vertices of $U_t(s)$ such that $d(s',v)+|{vt^*}|\leq d(s,t)$ holds
for each $v\in U_t(s)$ and $d(s',v)+|{vt^*}|< d(s,t)$ holds
for at least one vertex $v\in U_t(s)$, then $r$ is in $R(s,t)$.
}

We prove the claim as follows.
Consider any $t'\in M_t(s')$. To prove $r\in R(s,t)$, it is sufficient is to show that
$d(s',t')<d(s,t)$.

Since $s$ moves to $s'$ infinitesimally, the distance between $t$ and $t'$ is
also infinitesimal. Let $H$ be the convex hull of  the vertices of $U_t(s)$.
By Observation \ref{obser:10}, $t$ is in the interior of $H$. Hence, $t'$ is
also in the interior of $H$. Further, since $t'\in M_t(s')$ and $U_{t'}(s')\subseteq U_t(s)$,
it holds that $d(s',t')=\min_{v\in U_t(s)}(d(s',v)+|{vt'}|)$.

If $t'=t^*$, because there exists a vertex $v\in U_t(s)$ with
$d(s',v)+|{vt^*}|< d(s,t)$, we can obtain $d(s',t')=\min_{v\in
U_t(s)}(d(s',v)+|{vt'}|)<d(s,t)$, which proves the claim. Below we
assume $t'\neq t^*$.

We triangulate $H$ by adding a line segment from $t^*$ to each vertex
of $H$. Let $\triangle v_iv_jt^*$ be a triangle that contains $t'$, where $v_i$
and $v_j$ are two adjacent vertices of $H$. Since
$t'\neq t^*$, it is easy to see that at least one of
$|{t'v_i}|<|{t^*v_i}|$ and
$|{t'v_j}|<|{t^*v_j}|$ must hold.
Consequently, we can derive the following
\begin{equation*}
\begin{split}
d(s',t') & =\min_{v\in U_t(s)}(d(s',v)+|{vt'}|)
\leq \min\{d(s',v_i)+|{v_it'}|, d(s',v_j)+|{v_jt'}|\}\\
	 &< \max\{d(s',v_i)+|{v_it^*}|, d(s',v_j)+|{v_jt^*}|\}\leq d(s,t).
\end{split}
\end{equation*}
The last inequality is due to the condition in the claim that
$d(s',v)+|{vt^*}|\leq d(s,t)$ holds for each $v\in U_t(s)$.
The claim is thus proved.

Now we are back to prove the lemma for the case $t\in \calI$. Suppose we move
$s$ infinitesimally along $r_s$ with the unit speed to $s'$,
and move $t$ simultaneously along $r_t$ with speed $\tau$; let $t^*$ be the
point where $t$ is located when $s$ arrives at $s'$.
Since $|ss'|$ is infinitesimal, $|tt^*|$ is also infinitesimal.
In the following, we will show that $t^*$ satisfies the condition in the above
claim, which will lead to the lemma.

For each $v_i\in U_t(s)$, let $u_i$ be the coupled $s$-pivot of $v_i$
such that either $d'_{u_i,v_i}(s,t)<0$, or $d'_{u_i,v_i}(s,t)=0$ and
$d''_{u_i,v_i}(s,t)=0$. Note that $d_{u_i,v_i}(s',t^*)$ is the length of a path from $s'$ to
$t^*$ that is the concatenation of $\overline{s'u_i}$, $\pi(u_i,v_i)$, and
$\overline{v_it^*}$. Similarly, $d(s',v_i)+|v_it^*|$ is also the length of a path from $s'$ to
$t^*$ that is the concatenation of $\pi(s',v_i)$, and
$\overline{v_it^*}$.  Hence, the following holds
\begin{equation}\label{equ:20}
d(s',v_i)+|v_it^*|\leq d_{u_i,v_i}(s',t^*).
\end{equation}


On the one hand, for any $v_i\in U_t(s)$, since either $d'_{u_i,v_i}(s,t)<0$ or
$d'_{u_i,v_i}(s,t)=0$, we obtain $d_{u_i,v_i}(s,t)\geq d_{u_i,v_i}(s',t^*)$.
Hence, with Inequality \eqref{equ:20},
we obtain $d(s',v_i)+|v_it^*|\leq d_{u_i,v_i}(s,t)$ for any $v_i\in
U_t(s)$.

\begin{figure}[t]
\begin{minipage}[t]{\linewidth}
\begin{center}
\includegraphics[totalheight=0.8in]{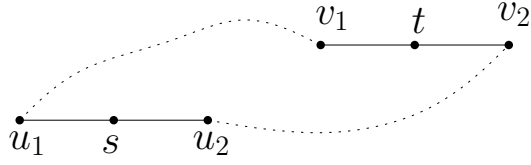}
\caption{\footnotesize Suppose $s$ and $t$ have two shortest paths: $\pi_{u_1,v_1}(s,t)$ and $\pi_{u_2,v_2}(s,t)$, with $s\in \overline{u_1u_2}$ and $t\in \overline{v_1v_2}$. If $s$ moves towards $u_1$ and $t$ moves towards $v_2$ (or $s$ moves towards $u_2$ and $t$ moves towards $v_1$) for the same speed, then $d'_{u_i,v_i}(s,t)=0$ and $d''_{u_i,v_i}(s,t)=0$ hold for both $i=1,2$.
}
\label{fig:specialpivots}
\end{center}
\end{minipage}
\vspace*{-0.15in}
\end{figure}

On the other hand, since $t\in \calI$ and $t$ is a vertex of $\spm(s)$, by
Observation \ref{obser:10}, $U_t(s)$
has at least three vertices. A vertex $v$ of $U_t(s)$ is called a {\em special
$t$-pivot} if it has a coupled $s$-pivot $u$ such that $d'_{u,v}(s,t)=0$
and $d''_{u,v}(s,t)=0$ (as shown in
\cite{ref:BaeTh13} this case happens only if $r_s$ is towards $u$
and $r_t$ is leaving $v$, or $r_s$ is leaving $u$ and $r_t$ is towards $v$; e.g., see Fig~\ref{fig:specialpivots}).
It was shown in \cite{ref:BaeTh13} that $U_t(s)$ has at most two special
$t$-pivots (e.g., see Fig~\ref{fig:specialpivots}). Hence, there is at least one vertex $v_i\in U_t(s)$ such that
$d'_{u_i,v_i}(s,t)<0$. This implies that $d(s,t)>d_{u_i,v_i}(s',t^*)$. With
Inequality \eqref{equ:20}, we have $d(s',v_i)+|v_it^*| < d(s,t)$.

The above proves that $t^*$ satisfies the condition in the claim. Hence, the
lemma follows for the case $t\in \calI$.

\subsubsection*{The case $t\in \calV$.}
We proceed on the third case $t\in \calV$. Recall the definitions of $s'$, $t'$, and
$M_t(s')$ in the beginning of the proof of the lemma. Our goal is to prove
$d(s,t)>d(s',t')$.
Note that since $t$ is a polygon vertex, $t'$ is either $t$ or
on one of the two polygon edges of $E$
incident to $t$ \cite{ref:ChiangTw99}.

Consider any vertex $v\in U_t(s)$. It has a coupled $s$-pivot $u$ such that
$d'_{u,v}(s,t)<0$, or $d_{u,v}'(s,t)=0$ and $d_{u,v}''(s,t)=0$. In fact, since
$t$ does not move (i.e., $\tau=0$), it is not possible that both $d_{u,v}'(s,t)=0$ and
$d_{u,v}''(s,t)=0$ hold. Indeed, since $\tau=0$, according to Equation \eqref{equ:10},
$d'_{u,v}(s,t)=\cos \gamma_s$ and $d''_{u,v}(s,t)=\frac{\sin^2 \gamma_s}{|su|}$.
Both $d_{u,v}'(s,t)=0$ and $d_{u,v}''(s,t)=0$ hold if and only if
$\cos\gamma_s=\sin\gamma_s=0$, which is not possible for any angle
$\gamma_s$.  Hence, we obtain $d'_{u,v}(s,t)<0$.
This implies that $d(s,t)=d_{u,v}(s,t)> |s'u|+d(u,v)+|vt|$. Further, since
$d(s,t)=|su|+d(u,v)+|vt|$, it holds that $|su|>|s'u|$.

Let $v$ be any vertex of $U_{t'}(s')$. By the definitions of $s'$ and $t'$,
$U_{t'}(s')\subseteq U_t(s)$.  Hence, $v$ is in $U_t(s)$.
The above shows that $v$ has a coupled $s$-pivot $u$ with
$d'_{u,v}(s,t)<0$, and $|su|>|s'u|$.
Since $v\in U_{t'}(s')$,
$\overline{s'u}\cup \pi(u,v)\cup \overline{vt'}$ is a
shortest path from $s'$ to $t'$.  Hence, $d(s',t')=|s'u|+d(u,v)+|vt'|$.

If $t'=t$, then since  $|su|>|s'u|$, it follows that
$d(s',t')=|s'u|+d(u,v)+|vt'|<|su|+d(u,v)+|vt|=d(s,t)$, which proves
the lemma.

If $t'\neq t$, as discussed above, $t'$ is on one of the two polygon edges
incident to $t$, which implies that $U_t(s)$ has more than one vertex
\cite{ref:ChiangTw99}.

We claim that there must exist a vertex $v\in U_t(s)$ such that $|t'v|\leq
|tv|$. Indeed, suppose to the contrary that $|t'v|> |tv|$ for every $v\in
U_t(s)$. Then, since $t'$ is in $M_t(s')$ and $|ss'|$ is infinitesimal, $|tt'|$
is also infinitesimal. Note that $\pi(s,t)$ cannot be a line segment since
otherwise $U_t(s)$ would have only one vertex, incurring contradiction.
Also, due to our general position
assumption that no three polygon vertices are on the same line, since $t$ is polygon
vertex, the last three vertices of any shortest path from $s$ to $t$ are not
in the same line. Since $|tt'|$ is infinitesimal, it holds that
$d(s,t')=\min_{v\in U_t(s)}(d(s,v)+|vt'|)$ (similar results were also proved in
\cite{ref:BaeTh13}). Since $|t'v|> |tv|$ for every $v\in U_t(s)$,
$d(s,t')>\min_{v\in U_t(s)}(d(s,v)+|vt|)=d(s,t)$, which contradicts with that
$t$ is a farthest point of $s$.

In light of the above claim, we assume $|t'v_i|\leq |tv_i|$ for a
vertex $v_i\in U_t(s)$.
Let $u_i$ be the coupled $s$-pivot of $v_i$
such that either $d'_{u_i,v_i}(s,t)<0$ or $d'_{u_i,v_i}(s,t)=0$.
We have shown above that $|su_i|>|s'u_i|$
Note that since $t'\in M_t(s')$, it also holds that $d(s',t')=\min_{v\in U_t(s)}(d(s',v)+|vt'|)$.
Based on the above discussion, we can derive the following,
\begin{equation*}
\begin{split}
d(s',t') & =\min_{v\in U_t(s)}(d(s',v)+|{vt'}|)\leq d(s',v_i)+|v_it'|  \\
	 &\leq |s'u_i| + d(u_i,v_i) + |v_it'| < |su_i| + d(u_i,v_i) + |v_it|=d(s,t).
\end{split}
\end{equation*}

This proves the lemma for the case $t\in \calV$.

\subsubsection*{The case $t\in E$.}
Let $e$ be the polygon edge that contains $t$.
Suppose we move $s$ infinitesimally along $r_s$ with the unit speed to
$s'$ and move $t$ simultaneously along $r_t$ with speed $\tau$ for
the same time to a point $t^*$ (i.e., $t^*$ is the location of $t$ when $s$
arrives at $s'$). Since $|ss'|$ is infinitesimal,
$|tt^*|$ is also infinitesimal. Since $r_t$ is parallel to $e$ and
$t$ is in the interior of $e$, $t^*$ is also in the interior of $e$.
Recall that $t$ corresponds to a set $M_t(s')$ of vertices in $\spm(s')$. Our
goal is to prove $d(s',t')<d(s,t)$ for any $t'\in M_t(s')$. Consider any $t'\in
M_t(s')$. Below, we prove that $d(s',t')<d(s,t)$.

As shown in \cite{ref:ChiangTw99}, $t'$ may be on $e$ or not. In either case,
$t'$ is infinitesimally close to $t$ as $|ss'|$ is infinitesimal.

Consider any $v_i\in U_t(s)$. Let $u_i$ be the coupled $s$-pivot of $v$
such that either $d'_{u_i,v_i}(s,t)<0$, or $d'_{u_i,v_i}(s,t)=0$ and
$d''_{u_i,v_i}(s,t)=0$. Hence, $d_{u_i,v_i}(s,t)\geq d_{u_i,v_i}(s',t^*)$.
By Observation \ref{obser:10}, at least one vertex of $U_t(s)$ must be in the
open half-plane bounded by the supporting line of $e$ and containing the interior
of $\calP$ in the small neighborhood of $e$; let $v_j$ be such a vertex. Since
$t$ moves along the direction $r_t$, which is parallel to $e$, according to
Equation \eqref{equ:10}, it is not possible that $d''_{u_j,v_j}(s,t)=0$. This
implies that $d'_{u_j,v_j}(s,t)<0$. Consequently, we obtain
$d_{u_j,v_j}(s,t)>d_{u_j,v_j}(s',t^*)$.

Next we prove the following {\em claim}: For any point $q$ on $e$ that is
infinitesimally close to $t$, it holds that $\min_{v_i\in
U_t(s)}(|s'u_i|+d(u_i,v_i)+|v_iq|)<d(s,t)$.

Indeed, if $q=t^*$, then we have $\min_{v_i\in
U_t(s)}(|s'u_i|+d(u_i,v_i)+|v_iq|)\leq
|s'u_j|+d(u_j,v_j)+|v_jq|=d_{u_j,v_j}(s',t^*) < d_{u_j,v_j}(s,t)=d(s,t)$.

\begin{figure}[t]
\begin{minipage}[t]{0.49\linewidth}
\begin{center}
\includegraphics[totalheight=1.0in]{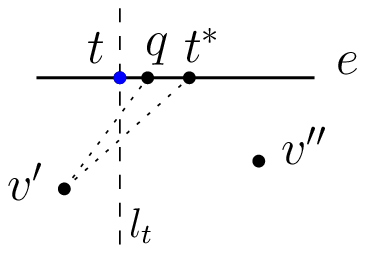}
\caption{\footnotesize Since $v'$ is to the left of both $q$ and $t^*$ and $q$
is to the left $t^*$, it holds that $|v'q|<|v't^*|$.
}
\label{fig:caseE10}
\end{center}
\end{minipage}
\hspace*{0.02in}
\begin{minipage}[t]{0.49\linewidth}
\begin{center}
\includegraphics[totalheight=1.0in]{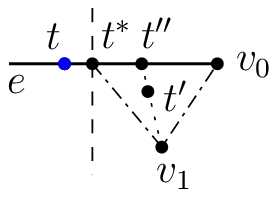}
\caption{\footnotesize Illustrating the proof: $|v_1t^*|>|v_1t''|>|v_1t'|$.
}
\label{fig:caseE20}
\end{center}
\end{minipage}
\vspace*{-0.15in}
\end{figure}

Next we assume $q\neq t^*$. By Observation \ref{obser:10}, $U_t(s)$
has at least one vertex in each of the two open half-planes bounded by $l_t$,
where $l_t$ is the line through $s$ and perpendicular to $e$. Without loss of
generality, we assume $e$ is horizontal (e.g., see Fig.~\ref{fig:caseE10}).
Hence, there is a vertex $v'\in U_t(s)$
strictly to the left of $l_t$ and a vertex $v''\in U_t(s)$ strictly to
the right of $l_t$. Since both $q$ and $t^*$ are infinitesimally close
to $t$, $v'$ (resp., $v''$) is strictly to the left (resp., right) of both $q$
and $t^*$. Without loss of generality, we assume $q$ is to the left of $t^*$.
Then, we have $|qv'|<|t^*v'|$. Let $v'$ be $v_k$ for some index $k$.
Consequently,
$\min_{v_i\in U_t(s)}(|s'u_i|+d(u_i,v_i)+|v_iq|)\leq |s'u_k|+d(u_k,v_k)+|v_kq|<
|s'u_k|+d(u_k,v_k)+|v_kt^*|=d_{u_k,v_k}(s',t^*)\leq d_{u_k,v_k}(s,t)=d(s,t)$.

Therefore, the above claim is proved.

Now we are back to our original problem for proving $d(s',t')<d(s,t)$. Depending
on whether $t'$ is on $e$ or not, there are two cases.

If $t'$ is on $e$, then since $t'$ is infinitesimally close to $t$, then by the
above claim, it holds that $\min_{v_i\in U_t(s)}(|s'u_i|+d(u_i,v_i)+|v_it'|)<d(s,t)$.
Note that since $t'$ is in $M_t(s')$, $d(s',t')=\min_{v_i\in
U_t(s)}(|s'u_i|+d(u_i,v_i)+|v_it'|)$. Hence, we obtain $d(s',t')<d(s,t)$.

If $t'$ is not on $e$, the proof is somewhat similar in spirit to the case $t\in \calI$.

We begin with proving an {\em observation} that there must be a vertex $v\in U_t(s)$ such that
$|vt'|<|vt^*|$. Without loss of generality, we assume $e$ is horizontal. All
vertices of $U_t(s)$ are in one of the closed half-planes bounded by $l_e$,
where $l_e$ is the horizontal line containing $e$.
Without loss of generality, we assume all vertices of $U_t(s)$ are below or on
the line $l_e$, i.e., they are in the closed half-plane bounded by $l_e$ from
above (let $h$ denote the half-plane).
Since $t'$ is not on $e$ and $t'$ is infinitesimally close to $t$,
$t'$ is strictly below $l_e$. Next, we do a ``triangulation'' around the point
$t^*$. Imagine that we rotate a rightwards ray originated from $t^*$ clockwise to sweep the
half-plane $h$, and let $v_1,v_2,\ldots v_m$ be the vertices of
$U_t(s)$ hit by our sweeping ray in order. Note that $v_1$ and $v_m$ may be on $e$. If
$v_1$ is not on $e$, let $v_0$ be the right endpoint of $e$. If
$v_m$ is not on $e$, let $v_{m+1}$ be the left endpoint of $e$.
Since $t'$ is infinitesimally close to $t$ and $t'$ is in
$h$, $t'$ must be in one of the triangles $\triangle t^*v_iv_{i+1}$ for $0\leq
i\leq m$.

Suppose $t'$ is in $\triangle
t^*v_iv_{i+1}$ for some $1\leq i\leq m-1$. Then, $v_i$ and $v_{i+1}$ are both from
$U_t(s)$. Since $t'$ is in $\triangle
t^*v_iv_{i+1}$, one of $|t'v_i|<|t^*v_i|$ and
$|t'v_j|<|t^*v_j|$  must hold, and this proves the above observation.

If $t'$ is in $\triangle t^*v_iv_{i+1}$ for $i=0$ or $m$, then one of $v_i$  and
$v_{i+1}$ is not in $U_t(s)$. So we cannot use the same argument as above. In
the following, we only prove the case for $i=0$, and the other case is similar.
When $t'$ is in $\triangle t^*v_0v_1$, $v_0$ is the right endpoint of
$e$ and $v_1$ is not
on $e$ (e.g., see Fig.~\ref{fig:caseE20}). In the following, we show that $|v_1t'|<|v_1t^*|$, which will prove the
observation.

By Observation \ref{obser:10}, $U_t(s)$ has at least one vertex strictly
to right of the vertical line through $t$. Since $|tt^*|$ is infinitesimal,
$U_t(s)$ has at least one vertex strictly to the right of
the vertical line $l_{t^*}$ through $t^*$ as
well. By the definition of $v_1$, $v_1$ must be strictly to the right of $l_{t^*}$.
Hence, the slope of the line through $v_1$ and $t^*$ is strictly negative.
Recall that $t'$ is in $\triangle t^*v_0v_1$. If $t'$ is on $\overline{v_1t^*}$,
then since $t'\neq t^*$ (due to $t'\not\in e$), we have $|v_1t^*|>|v_1t'|$.
Otherwise, we extend $\overline{v_1t'}$ until it hits $\overline{v_0t^*}$ at a
point $t''$, which is strictly to the right of $t^*$.
Since both $t^*$  and $t'$ are infinitesimally close to $t$,
$|t^*t'|$ is also infinitesimal and thus $|t^*t''|$ is infinitesimal as well.
Since the slope of $\overline{v_1t^*}$ is strictly negative and $t''$ is
strictly to the right of $t^*$, we obtain $|v_1t^*|>|v_1t''|>|v_1t'|$.

The above proves the observation.

In light of the above observation, we assume $|t'v_i|< |t^*v_i|$ for a
vertex $v_i\in U_t(s)$.
Let $u_i$ be the coupled $s$-pivot of $v_i$
such that either $d'_{u_i,v_i}(s,t)<0$, or $d'_{u_i,v_i}(s,t)=0$.
Note that $d_{u_i,v_i}(s,t)\geq d_{u_i,v_i}(s',t^*)$.
Since $t'$ is in $M_t(s')$, $d(s',t')=\min_{v\in U_t(s)}(d(s',v)+|vt'|)$.
Based on the above discussion, we derive the following
\begin{equation*}
\begin{split}
d(s',t') & =\min_{v\in U_t(s)}(d(s',v)+|\overline{vt'}|)\leq d(s',v_i)+|v_it'|
< d(s',v_i)+|v_it^*| \\
	 &\leq |s'u_i| + d(u_i,v_i) + |v_it^*| = d_{u_i,v_i}(s',t^*) \leq d_{u_i,v_i}(s,t)=d(s,t).
\end{split}
\end{equation*}

This proves that $d(s',t')<d(s,t)$. This finishes the proof for
the case $t\in E$.

The lemma is thus proved for all three cases.
\end{proof}

Lemma \ref{lem:20} is on the case where $s$ is not visible to $t$. If $s$ is
visible to $t$, the result is trivial, as shown in Observation
\ref{obser:20}.

\begin{observation}\label{obser:20}
Suppose $t$ is a farthest point of $s$ and $s$ is visible to $t$. Then
$t$ must be a polygon vertex of $\calV$. Further, a free direction $r_s$ of
$s$ is in $R(s,t)$ if and only if $r_s$ is towards the interior of $h_s(t)$,
where $h_s(t)$ is the open half-plane containing $t$ and bounded by the line
through $s$ and perpendicular to $\overline{st}$ (e.g., see
Fig.~\ref{fig:halfplane}).
\end{observation}
\begin{proof}
Since $s$ is visible to $t$, $\overline{st}$ is the only shortest path from $s$
to $t$. As $t$ is a vertex of $\spm(s)$, $t$ cannot be in $\calI$
or $E$, since otherwise there would be more than one shortest \st\
path.  Thus, $t\in \calV$.

\begin{figure}[t]
\begin{minipage}[t]{\linewidth}
\begin{center}
\includegraphics[totalheight=0.8in]{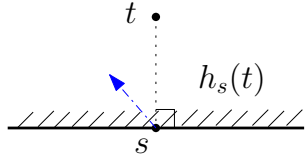}
\caption{\footnotesize Illustrating Observation \ref{obser:20}: The blue
direction is towards the interior of $h_s(t)$ and thus is in $R(s,t)$.
}
\label{fig:halfplane}
\end{center}
\end{minipage}
\vspace*{-0.15in}
\end{figure}

Consider any free direction $r_s$ of $s$. Suppose we move $s$
infinitesimally along $r_s$ to $s'$. According to our definition of
``visibility'', $\overline{st}$ does not contain any polygon vertex in
its interior. Since $|ss'|$ is infinitesimal, $s'$ is also visible to
$t$. Therefore, the point $t$, as a vertex of $\spm(s)$, corresponds to a vertex
of $\spm(s')$ that is $t$ itself. Hence, $r_s$ is in $R(s,t)$ if and
only $|s't|<|st|$. Clearly, $|s't|<|st|$ if and only if
$r_s$ is towards the interior of the open half-plane $h_s(t)$.
\end{proof}

By Observation \ref{obser:20}, if $s$ is visible to $t$, then the
range
$R(s,t)$ is the intersection of the free direction range $R_f(s)$ and an open range of size $\pi$ delimited by the open half-plane $h_s(t)$.

The next lemma is proved by using Lemmas \ref{lem:10} and \ref{lem:20}
as well as Observation \ref{obser:20}.

\begin{lemma}\label{lem:30}
Among all points of $\calP$ that have topologically equivalent shortest path
maps in $\calP$, there is at most one geodesic center. This implies that
each cell or edge of $\spmed$ contains at most one geodesic center in its
interior, which further implies that the number of geodesic centers of $\calP$
is $O(|\spmed|)$, where $|\spmed|$ is the combinatorial complexity of
$\spmed$.
\end{lemma}
\begin{proof}
Let $Q$ be any set of points of $\calP$ that have topologically equivalent
shortest path maps in $\calP$. We show that there is at most one geodesic center in $Q$,
which will prove the lemma. Note that any two
points of $Q$ must be visible to each other since otherwise their shortest path maps
would not be
topologically equivalent. Let $s$ be a geodesic center in $Q$. Let $s'$ be any other
point in $Q$. In the following, we prove that $d_{\max}(s')>d_{\max}(s)$, which
implies that $s'$ cannot be a geodesic center, and thus the lemma will be proved.

Consider the direction $r_s$ of moving $s$ towards $s'$. Since $s$ is
visible to $s'$, $r_s$ is a free direction.
Since $s$ is a geodesic center, by Lemma \ref{lem:10}, $R(s)$ is empty. Thus,
$s$ has a farthest point $t$ such that $r_s\not\in R(s,t)$. Because $t$ is
a farthest point of $s$, $d(s,t)=d_{\max}(s)$.

If $s$ is visible to $t$, by Observation \ref{obser:20}, $r_s$ is not
towards the interior of the open half-plane $h_s(t)$. Hence,
$|s't|>|st|$, and thus $d_{\max}(s')\geq d(s',t)\geq
|s't|>|st|=d(s,t)=d_{\max}(s)$.

In the following, we assume that $s$ is not visible to $t$.
Depending on whether $t$ is in $\calV$, $E$, or
$\calI$, there are three cases.

\paragraph{The case $t\in \calV$.}
Suppose $t\in \calV$.  Due to $r_s\not\in R(s,t)$,
by Lemma \ref{lem:20}(3), if we move $s$ along $r_s$
with unit speed, there exists a vertex $v\in U_t(s)$ such that either $d'_{u,v}(s,t)>0$,
or $d'_{u,v}(s,t)=0$ and $d''_{u,v}(s,v)\neq 0$, for any coupled $s$-pivot
$u$ of $v$. Let $u$ be any coupled $s$-pivot of $v$.

Note that $\overline{su}\cup \pi(u,v)\cup \overline{vt}$ is
a shortest path from $s$ to $t$, whose length is $d_{u,v}(s,t)$.
Since $\spm(s)$ and $\spm(s')$ are
topologically equivalent and $t$ is a polygon vertex, $\overline{s'u}\cup
\pi(u,v)\cup \overline{vt}$ is a shortest path from $s'$ to $t$, whose
length is $d_{u,v}(s',t)$. In the following, we show that
$d_{u,v}(s',t)>d_{u,v}(s,t)$.

Indeed, as $s$ moves towards $s'$ along $r_s$ with unit speed, recall
that the second derivative
$d''_{u,v}(s,t)\geq 0$ always holds. Hence, if $d'_{u,v}(s,t)>0$,
then during the movement of $s$, it always holds that $d'_{u,v}(s,t)> 0$. This
implies that $d_{u,v}(s',t)>d_{u,v}(s,t)$. Similarly, if
$d'_{u,v}(s,t)=0$ and $d''_{u,v}(s,t)\neq 0$, then since $d''_{u,v}(s,t)\geq 0$,
we obtain $d''_{u,v}(s,t)> 0$. Consequently, as $s$ moves towards $s'$
along $r_s$ with unit speed, except for the starting moment,
it holds that  $d'_{u,v}(s,t)> 0$ and $d''_{u,v}(s,t)\geq 0$. Thus,
$d_{u,v}(s',t)>d_{u,v}(s,t)$.

The above proves that $d(s',t)=d_{u,v}(s',t)>d_{u,v}(s,t)=d(s,t)$.  Therefore,
we obtain $d_{\max}(s')>d_{\max}(s)$ since $d_{\max}(s')\geq d(s',t)>
d(s,t)=d_{\max}(s)$.

\paragraph{The case $t\in E$.}
If $t\in E$, let $e$ denote the polygon edge of $E$ that contains $t$.
Due to $r_s\not\in R(s,t)$, by Lemma \ref{lem:20}(2), if we move $s$ along $r_s$
with unit speed and move $t$ along $e$ with any speed $\tau\geq 0$,
there exists a vertex $v\in U_t(s)$ such that either $d'_{u,v}(s,t)>0$,
or $d'_{u,v}(s,t)=0$ and $d''_{u,v}(s,v)\neq 0$, for any coupled $s$-pivot
$u$ of $v$. Let $u$ be any coupled $s$-pivot of $v$.

Note that $\overline{su}\cup \pi(u,v)\cup \overline{vt}$ is
a shortest path from $s$ to $t$, whose length is $d_{u,v}(s,t)$.
Since $\spm(s)$ and $\spm(s')$ are topologically equivalent and $t\in e$,
the point $t$, as a vertex of $\spm(s)$,
corresponds to one and only one vertex $t'$ in $\spm(s')$ that is also on $e$.
Then, $\overline{s'u}\cup d(u,v)\cup \overline{vt'}$ is a shortest path from
$s'$ to $t'$, whose length is $d_{u,v}(s',t')$.
In the following, we show that
$d_{u,v}(s',t')>d_{u,v}(s,t)$, which will lead to
$d_{\max}(s')>d_{\max}(s)$ as $d_{\max}(s')\geq d(s',t')=d_{u,v}(s',t')$ and
$d_{\max}(s)=d(s,t)=d_{u,v}(s,t)$.

Suppose we move $s$ along $r_s$ towards $s'$ with unit speed and
move $t$ along $e$ towards $t'$ with speed $\tau=|tt'|/|ss'|$.
Hence, when $s$ arrives at $s'$, $t$ arrives at $t'$ simultaneously.
By Lemma \ref{lem:20}(2), either $d'_{u,v}(s,t)>0$,
or $d'_{u,v}(s,t)=0$ and $d''_{u,v}(s,v)\neq 0$. In either case, by the same analysis
as in the above case $t\in \calV$, we can show that $d_{u,v}(s',t')>d_{u,v}(s,t)$
and we omit the details.

\paragraph{The case $t\in \calI$.}
The proof for this case is very similar.
Due to $r_s\not\in R(s,t)$, by Lemma \ref{lem:20}(1), if we move $s$ along $r_s$
with unit speed and move $t$ along any free direction
with any speed $\tau\geq 0$,
there exists a vertex $v\in U_t(s)$ such that either $d'_{u,v}(s,t)>0$,
or $d'_{u,v}(s,t)=0$ and $d''_{u,v}(s,v)\neq 0$, for any coupled $s$-pivot
$u$ of $v$.

Since $\spm(s)$ and $\spm(s')$ are topologically equivalent, $t$ corresponds to
one and only one vertex
$t'$ in $\spm(s')$ that is also in $\calI$.
The rest of the argument is exactly the same as that for the case $t\in E$. We
can prove that $d_{\max}(s')>d_{\max}(s)$ and we omit the details.
%

This completes the proof for the lemma.
\end{proof}

{\em Remark:} By extending the proof of Lemma \ref{lem:30},
it is possible to obtain a slightly stronger result: Every cell (including its
boundary) of $\spmed$ contains at most one
geodesic center.

The following corollary can be proved by the same techniques as Lemma
\ref{lem:30}, and it implies that if $t$ is a farthest point of $s$, then slightly moving $s$ along a free direction that is not in $R(s,t)$ can never obtain a geodesic center.

\begin{corollary}\label{corr:10}
Suppose $t$ is a farthest point of $s$. If we move $s$ infinitesimally
along a free direction that is not in $R(s,t)$, then $d_{\max}(s)$
will become strictly larger.
\end{corollary}
\begin{proof}
Let $r_s$ be any free direction that is not in $R(s,t)$.
Suppose we move $s$ infinitesimally along $r_s$ to $s'$.
By using exactly the same argument as in the proof of Lemma \ref{lem:30},
we can show that $\spm(t)$ has a vertex $t'$ corresponding to $t$ in $\spm(s)$ such that $d(s,t)<d(s',t')$. Therefore, we obtain $d_{max}(s)=d(s,t)<d(s',t')\leq d_{\max}(s')$, which proves the corollary.
%
%
\end{proof}

So far we have shown that the total number of geodesic centers is bounded by the
combinatorial size of $\spmed$. This result, although it is interesting in its own
right, is not quite helpful for computing the geodesic centers. In order to
compute candidate points for geodesic centers, we need to find a way to
determine the range $R(s,t)$. It turns out that it is sufficient to determine
$R(s,t)$ when $t$ is in a non-degenerate position with respect to $s$ in the
following sense: Suppose $t$ is a farthest point of $s$; we say that $t$ is
{\em non-degenerate} with respect to $s$ if there are
exactly three, two, and one shortest \st\ paths for $t$ in $\calI$, $E$, and $\calV$,
respectively (by Observation \ref{obser:10}, this implies that $|U_t(s)|$ is
$3$, $2$, and $1$, respectively for the three cases).


Lemma \ref{lem:20} gives a sufficient condition for a direction in $R(s,t)$.
The following lemma gives both a sufficient and a necessary
condition for a direction in $R(s,t)$ when $t$ is non-degenerate, and the lemma will be used to explicitly compute the range $R(s,t)$ in Section \ref{sec:range}.
Note that Observation \ref{obser:20} already gives a way to determine $R(s,t)$
when $s$ is visible to $t$.

\begin{lemma}\label{lem:40}
Suppose $t$ is a non-degenerate farthest point of $s$ and $s$ is not visible to
$t$. Then,
a free direction $r_s$ is in $R(s,t)$ if and only if
\begin{enumerate}
\item
for $t\in \calI$,
there is a free direction $r_t$ for $t$ with a speed $\tau\geq 0$ such that when we
move $s$ along $r_s$ with unit speed and move $t$ along $r_t$ with speed
$\tau$, each vertex $v\in U_t(s)$ has a coupled $s$-pivot $u$ with
$d_{u,v}'(s,t)<0$.
\item
for $t\in E$,
there is a free direction $r_t$ for $t$ that is parallel to the polygon edge
containing $t$ with a speed $\tau\geq 0$ such that when we
move $s$ along $r_s$ with unit speed and move $t$ along $r_t$ with speed
$\tau$, each vertex $v\in U_t(s)$ has a coupled $s$-pivot $u$ with
$d_{u,v}'(s,t)<0$.
\item
for $t\in \calV$,
when we move $s$ along $r_s$ with unit speed,
each vertex $v\in U_t(s)$ has a coupled $s$-pivot $u$ with
$d_{u,v}'(s,t)<0$.
\end{enumerate}
\end{lemma}
\begin{proof}
First of all, in any of these three cases, if the condition in the lemma
statement holds, by Lemma \ref{lem:20}, $r_s$ is in $R(s,t)$.
In the following, we prove the other direction of the lemma.

Let $r_s$ be in $R(s,t)$.
Suppose we move $s$ along $r_s$ infinitesimally to $s'$.
Since $t$ is non-degenerate with respect to $s$,
the point $t$, as a vertex of $\spm(s)$,
corresponds to one and only one vertex $t'$ in $\spm(s')$ (i.e.,
$M_t(s')=\{t'\}$) \cite{ref:ChiangTw99}.
Due to $r_s\in R(s,t)$, $d(s,t)>d(s',t')$.
In the following, we prove the three cases: $t\in \calI$, $t\in E$, and $t\in
\calV$.

\paragraph{The case $t\in \calI$.}
Suppose we move $s$ towards $s'$ with unit speed and move $t$ towards $t'$ with
speed $|tt'|/|ss'|$. Then, when $s$ arrives at $s'$, $t$ arrives at $t'$
simultaneously.

Consider any vertex $v\in U_t(s)$. To prove the lemma for this case, our goal is to show that
there exists a coupled $s$-pivot $u$ of $v$ with $d'_{u,v}(s,t)<0$.

Let $u$ be any coupled $s$-pivot of $v$. Hence, $d(s,t)=d_{u,v}(s,t)$.
Since $M_t(s')=\{t'\}$,
$\overline{s'u}\cup \pi(u,v)\cup \overline{vt'}$ is a shortest path from $s'$ to
$t'$, and thus $d(s',t')=d_{u,v}(s',t')$.
We claim that $d'_{u,v}(s,t)<0$.
Suppose to the contrary that $d'_{u,v}(s,t)\geq 0$. If
$d'_{u,v}(s,t)> 0$, then we would obtain $d_{u,v}(s,t)<d_{u,v}(s',t')$, which
contradicts with $d(s,t)>d(s',t')$. Similarly, if $d'_{u,v}(s,t)=0$, since
$d''_{u,v}(s,t)\geq 0$, we would obtain $d_{u,v}(s,t)\leq d_{u,v}(s',t')$, which
contradicts with $d(s,t)>d(s',t')$. Hence, $d'_{u,v}(s,t)<0$ is proved.

This proves the lemma for the case $t\in \calI$.

\paragraph{The case $t\in E$.}
Let $e$ be the polygon edge containing $t$. Then, $t'$ is also on $e$. Suppose we
move $s$ towards $s'$ with unit speed and move $t$ on $e$ towards $t'$ with
speed $|tt'|/|ss'|$. Hence, when $s$ arrives at $s'$, $t$ arrives at $t'$
simultaneously.

Consider any vertex $v\in U_t(s)$. Let $u$ be any coupled $s$-pivot of
$v$. As in the above case, $d_{u,v}(s,t)=d(s,t)$ and $d_{u,v}(s',t')=d(s',t')$.
Using the same analysis as in the above case we can also prove that
$d'_{u,v}(s,t)<0$, which leads to the lemma.

\paragraph{The case $t\in \calV$.}
In this case, since $t$ is a polygon vertex, $t'=t$.
Suppose we move $s$ towards $s'$ with unit speed.
Consider any vertex $v\in U_t(s)$. Let $u$ be any coupled $s$-pivot of
$v$. Again, $d_{u,v}(s,t)=d(s,t)$ and $d_{u,v}(s',t')=d(s',t')$.
Using the same analysis as in the first case we can also prove that
$d'_{u,v}(s,t)<0$, which leads to the lemma.
\end{proof}

Lemma \ref{lem:40} will be used to determine the range $R(s,t)$ for a
non-degenerate farthest point of $s$. The details are deferred in
Section \ref{sec:range}, where we will show that $R(s,t)$ is the
intersection of the free direction range $R_f(s)$ and an open range of
size $\pi$ (i.e., the previously mentioned $\pi$-range property).

In addition, we present Lemma \ref{lem:50}, which will be useful for computing the candidate points in Section \ref{sec:candidates}. If a farthest point $t$ of $s$ is not non-degenerate, then we say that $t$ is {\em degenerate} (note that $s$ cannot be visible to $t$ in the degenerate case).  Lemma \ref{lem:50} provides a sufficient condition for a direction in $R(s,t)$ particularly for a degenerate farthest point $t$ of $s$.

\begin{lemma}\label{lem:50}
Suppose $t$ is a degenerate farthest point of $s$.
\begin{enumerate}
\item
For $t\in \calI$, a free direction $r_s$ is in $R(s,t)$ if the following conditions are satisfied: (1)
there exist three vertices $v_1,v_2,v_3\in U_t(s)$ such that
$t$ is in the interior of the triangle $\triangle v_1v_2v_3$ (i.e., $\{v_1,v_2,v_3\}$ satisfies the same condition as $U_t(s)$ in Observation \ref{obser:10}(1));
(2) there exists a free direction $r_t$ for $t$ with a speed $\tau\geq 0$ such that when we
move $s$ along $r_s$ with unit speed and move $t$ along $r_t$ with speed
$\tau$, each vertex $v\in \{v_1,v_2,v_3\}$ has a coupled $s$-pivot $u$ with
$d_{u,v}'(s,t)<0$.
\item
For $t\in E$, suppose $e$ is the polygon edge containing $t$.
A free direction $r_s$ is in $R(s,t)$
if the following conditions are satisfied: (1)
there exist two vertices $v_1,v_2\in U_t(s)$ such that
$\{v_1,v_2\}$ has one vertex in each of the two open half-planes bounded by the line through $t$ and perpendicular to $e$,
and $\{v_1,v_2\}$ has at least one vertex in the
open half-plane bounded by the supporting line of $e$ and containing the interior
of $\calP$ in the small neighborhood of $e$ (i.e., $\{v_1,v_2\}$ satisfies the same condition as $U_t(s)$ in Observation \ref{obser:10}(2));
(2) there is a free direction $r_t$ for $t$ parallel to $e$
with a speed $\tau\geq 0$ such that when we
move $s$ along $r_s$ with unit speed and move $t$ along $r_t$ with speed
$\tau$, each vertex $v\in \{v_1,v_2\}$ has a coupled $s$-pivot $u$ with
$d_{u,v}'(s,t)<0$.
\end{enumerate}
\end{lemma}
\begin{proof}
The proof uses similar techniques as in the proof of Lemma \ref{lem:20}. Indeed,
the proof of Lemma \ref{lem:20} mainly relies on Observation \ref{obser:10}.
Here, for the case $t\in \calI$, $\{v_1,v_2,v_3\}$ satisfies Observation
\ref{obser:10}(1); for  the case $t\in E$, $\{v_1,v_2\}$ satisfies Observation
\ref{obser:10}(2). Therefore, similar techniques as in the proof of Lemma
\ref{lem:20} can be used here. We briefly discuss it below.

\paragraph{The case $t\in \calI$.}
We first consider the case $t\in \calI$. Let $v_1$, $v_2$, and $v_3$ be the polygon vertices
specified in the lemma statement. For each $i$ with $1\leq i\leq 3$, let $u_i$
be the coupled $s$-pivot of $v_i$ with $d'_{u_i,v_i}(s,t)<0$. Suppose we move
$s$ along $r_s$ infinitesimally with unit speed to $s'$
and simultaneously move $t$ along $r_t$ with speed $\tau$ to a point $t^*$
(i.e., $t^*$ is the location of $t$ when $s$
arrives at $s'$). Since $|ss'|$ is
infinitesimal, $|tt^*|$ is also infinitesimal. Since $t$ is in the interior of
$\triangle v_1v_2v_3$, $t^*$ is also in the interior of the triangle. Consider
any $t'$ in $M_t(s')$. To prove that $r_s$ is in $R(s,t)$, our goal
is to show that $d(s',t')<d(s,t)$.

Since $|ss'|$ is infinitesimal, $|tt'|$ is also infinitesimal and $t'$ is also
in the interior of $\triangle v_1v_2v_3$. For each $1\leq i\leq 3$, since
$d'_{u_i,v_i}(s,t)<0$, it holds that $d_{u_i,v_i}(s,t)>d_{u_i,v_i}(s',t^*)$.
The three segments $\overline{t^*v_i}$ for $i=1,2,3$ partition $\triangle
v_1v_2v_3$ into three smaller triangles and $t'$ must be one of them. Without
loss of generality, assume $t'$ is in $\triangle v_1v_2t^*$. Hence, at least one
of $|t'v_1|<|t^*v_1|$ and $|t'v_2|<|t^*v_2|$ holds. Without loss of generality,
we assume the former one holds. Therefore, we can derive
$d(s',t')\leq
|s'u_1|+d(u_1,v_1)+|v_1t'|<|s'u_1|+d(u_1,v_1)+|v_1t^*|=d_{u_1,v_1}(s',t^*)<d_{u_1,v_1}(s,t)=d(s,t)$.

\paragraph{The case $t\in E$.}
Let $v_1$ and $v_2$ be the two polygon vertices
specified in the lemma statement. For each $i$ with $1\leq i\leq 2$, let $u_i$
be the coupled $s$-pivot of $v_i$ with $d'_{u_i,v_i}(s,t)<0$. Suppose we move
$s$ along $r_s$ infinitesimally with unit speed to $s'$
and simultaneously move $t$ along $e$ with speed $\tau$ to a point $t^*$. Since $|ss'|$ is
infinitesimal, $|tt^*|$ is also infinitesimal. Consider
any $t'$ in $M_t(s')$. To prove that $r_s$ is in $R(s,t)$, our goal
is to show that $d(s',t')<d(s,t)$.

For each $1\leq i\leq 2$, since
$d'_{u_i,v_i}(s,t)<0$, it holds that $d_{u_i,v_i}(s,t)>d_{u_i,v_i}(s',t^*)$.
Since $|tt^*|$ is infinitesimal, $\{v_1,v_2\}$ has one vertex in each of the open half-planes bounded by the line through $t^*$ and perpendicular to $e$.
Also, because at least one vertex of $v_1$ and $v_2$ is in the
open half-plane bounded by the supporting line of $e$ and containing the interior
of $\calP$ in the small neighborhood of $e$, regardless
of whether $t'$ is on $e$ or not, we can use
the similar approach as in the proof of Lemma \ref{lem:20} (for the case $t\in
E$) to show that either $|t'v_1|<|t^*v_1|$ or $|t'v_2|<|t^*v_2|$ holds.
Consequently, by the similar argument as the above case for $t\in \calI$, we can
obtain $d(s',t')<d(s,t)$.
\end{proof}

\section{Determining the Admissible Direction Range $R(s,t)$ and the
$\pi$-Range Property}
\label{sec:range}
In this section, we determine the admissible direction range $R(s,t)$ for any
point $s$ and any of its non-degenerate farthest point $t$. In particular, we
will prove the $\pi$-range property mentioned in Section \ref{sec:intro}.

Depending on whether $t$ is in $\calV$, $E$, and $\calI$, there are three cases. Recall that $R_f(s)$ is the range of
all free directions of $s$.
In each case, we will show that $R(s,t)$ is the intersection of $R_f(s)$ and an
open range $\piR(s,t)$ of size $\pi$.  We call  $\piR(s,t)$ the {\em
$\pi$-range}.
As will be seen later, the $\pi$-range $\piR(s,t)$  can
be explicitly determined based on the positions of
$s$, $t$, and the vertices of $U_s(t)$ and $U_t(s)$.

In fact, for each case, we will give more general results
that are on shortest path distance functions. These
more general results will also be useful for computing the candidate points
later in Section \ref{sec:candidates}.

\subsection{The Case $t\in\calV$}
\label{sec:caseV}
We first discuss the case $t\in \calV$. The result is relatively straightforward
in this case. If $s$ is visible to $t$, the $\pi$-range $\piR(s,t)$ is defined to be the
open range of directions delimited by the open half-plane $h_s(t)$
as defined in Observation \ref{obser:20}; by
Observation \ref{obser:20}, $R(s,t)=R_f(s)\cap \piR(s,t)$.

In the following, we assume $s$ is not visible to $t$.
We first present a more general result on a shortest path function.
Let $s$ and $t$ be any two points in
$\calP$ such that $t$ is in $\calV$ and $s$ is not visible to $t$.
Let $\pi(s,t)$ be any shortest \st\ path in $\calP$.
Let $u$ and $v$ be the $s$-pivot and $t$-pivot in $\pi(s,t)$, respectively.
Thus, $d_{u,v}(s,t)=|su|+d(u,v)+|vt|$. Now we consider $d_{u,v}(s,t)$ as a function of
$s$ and $t$ in the entire plane $\bbR^2$ (not only in $\calP$; namely, when we
move $s$ and $t$, they are allowed to move outside $\calP$, but the function
$d_{u,v}(s,t)$ is always defined as $|su|+d(u,v)+|vt|$, where $d(u,v)$ is a
fixed value).


The {\em $\pi$-range} $\piR(s,t)$ is defined with respect to $t$ and the path $\pi(s,t)$ as follows: a direction $r_s$ for $s$ is in $\piR(s,t)$ if $d'_{u,v}(s,t)<0$ when
we move $s$ along $r_s$ with unit speed.  The following lemma is quite
straightforward.

\begin{lemma}\label{lem:60}
The $\pi$-range $\piR(s,t)$ is exactly the open range of size $\pi$ delimited by
$h_s(u)$, where $h_s(u)$ is the open half-plane containing $u$ and
bounded by the line through $s$ and perpendicular to $\overline{su}$.
\end{lemma}
\begin{proof}
The analysis is similar to Observation \ref{obser:20}.
Suppose we move $s$ along a direction $r_s$ with unit speed. Then, $d'_{u,v}(s,t)<0$
if and only if $r_s$ is towards the interior of $h_s(u)$.
\end{proof}

Now we are back to our original problem to determine $R(s,t)$ for a
non-degenerate farthest point $t$ of $s$ with $t\in \calV$. Since $t$ is non-degenerate and $t$ is in $\calV$, there is only one
shortest path $\pi(s,t)$ from $s$ to $t$. We define $\piR(s,t)$ as above.
Based on Observation \ref{obser:20} and Lemmas \ref{lem:40}(3), we have Lemma \ref{lem:70}, and thus  $R(s,t)$ can be determined by Observation \ref{obser:20} and Lemma \ref{lem:60}.


\begin{lemma}\label{lem:70}
$R(s,t)=R_f(s)\cap \piR(s,t)$.
\end{lemma}
\begin{proof}
If $s$ is visible to $t$,
we have already shown that the lemma is true. Below we assume $s$ is
not visible to $t$.

Let $u$ and $v$ be the $s$-pivot and $t$-pivot in $\pi(s,t)$, respectively. Note
that $v$ is the only vertex in $U_t(s)$ and $u$ is the only vertex in $U_s(t)$.

\begin{enumerate}
\item
Consider any direction $r_s\in R(s,t)$, i.e., $r_s$ is an admissible direction
for $s$ with respect to $t$.
According to Lemma \ref{lem:40}(3), when we move $s$ along $r_s$ with unit speed,
it holds that $d'_{u,v}(s,t)<0$. This implies that $r_s$ is in $\piR(s,t)$.
Since $r_s$ is in $R(s,t)$, $r_s$ is in $R_f(s)$.
Therefore, $r_s$  is in $R_f(s)\cap \piR(s,t)$.
\item
Consider any direction $r_s$ in $R_f(s)\cap \piR(s,t)$.
First of all, $r_s$ is a free direction. Since $r_s$ is
in $\piR(s,t)$, when $s$ moves along $r_s$ with unit speed, $d'_{u,v}(s,t)<0$.
Since $v$ is the only vertex in $U_t(s)$, by Lemma
\ref{lem:40}(3), $r_s$ is in $R(s,t)$.
\end{enumerate}
The lemma is thus proved.
\end{proof}

\subsection{The Case $t\in E$}
\label{sec:rangeE}

The analysis for this case is substantially more complicated than the previous
case, although the next case for $t\in \calI$ is even more challenging. One may
consider the analysis for this case as a ``warm-up'' for the most
general case $t\in \calI$.

As in the previous case, we first present a more general result that
is on two
shortest path distance functions. Let $s$ and $t$ be any two points in $\calP$ such that
$t$ is in $E$ and there are two shortest \st\ paths
$\pi_1(s,t)$ and $\pi_2(s,t)$ (this implies that $s$ is not visible
to $t$). Let $e$ be the polygon edge containing
$t$ and let $l(e)$ denote the line containing $e$.  For each $i=1,2$, let
$\pi_i(s,t)=\pi_{u_i,v_i}(s,t)$, i.e., $u_i$ and $v_i$ are the
$s$-pivot and $t$-pivot of $\pi_i(s,t)$, respectively. We further
require the set $\{v_1,v_2\}$ satisfy the same condition as $U_t(s)$ in Observation
\ref{obser:10}(2), i.e., $\{v_1,v_2\}$ has at least one vertex in the
open half-plane bounded by $l(e)$ and containing the interior of
$\calP$ in the small neighborhood of $e$, and it has at least one
vertex in each of the two open half-planes bounded by the line through
$t$ and perpendicular to $e$.
We say that the two shortest paths $\pi_1(s,t)$ and $\pi_2(s,t)$ are {\em
canonical} with respect to $s$ and $t$ if $\{v_1,v_2\}$ satisfies the above condition. In the following, we assume $\pi_1(s,t)$ and $\pi_2(s,t)$ are canonical.
Note that the condition implies that $v_1\neq v_2$. However, $u_1=u_2$ is possible.

For each $i=1,2$, we consider $d_{u_i,v_i}(s,t)=|su_i|+d(u_i,v_i)+|v_it|$ as a
function of $s\in \bbR^2$ (instead of $s\in \calP$ only) and $t\in e$.

In this case, the {\em $\pi$-range} $\piR(s,t)$ of $s$ is defined with respect to $t$
and the two paths $\pi_1(s,t)$ and $\pi_2(s,t)$
as follows: a direction $r_s$ for $s$ is in $\piR(s,t)$ if there exists a
direction $r_t$ parallel to $e$ for $t$ with a speed $\tau\geq 0$ such that when
we move $s$ along $r_s$ with unit speed and move $t$ along $r_t$ with
speed $\tau\geq 0$, $d'_{u_i,v_i}(s,t)<0$ holds for $i=1,2$.

In Section \ref{sec:caseV}, we showed that the $\pi$-range for the case $s\in \calV$
is an open range of size $\pi$. Here we will show a similar
result in Lemma \ref{lem:80} unless a special case happens.
Although the analysis in Section \ref{sec:caseV} is quite straightforward, the result here for two functions $d_{u_i,v_i}(s,t)$ with $i=1,2$ is somewhat surprising and the proof is
substantially more difficult (one could imagine that a similar result
for three functions, as shown in Section \ref{sec:three}, is even more
surprising). Before presenting Lemma \ref{lem:80}, we introduce some notation.

For any two points $p$ and $q$ in the plane, define $\vector{pq}$ as the
direction from $p$ to $q$.

Recall that the angle of any direction $r$ is defined to be the
 angle in $[0,2\pi)$ counterclockwise from the positive direction of the $x$-axis.
Let $\alpha_1$ denote the angle of the direction $\vector{su_1}$, and
let $\alpha_2$ denote the angle of the direction $\vector{su_2}$ (e.g., see
Fig.~\ref{fig:twoangles}).
Note that by our way of defining pivot
vertices, $\alpha_1=\alpha_2$ if and only if $u_1=u_2$.

\begin{figure}[t]
\begin{minipage}[t]{\linewidth}
\begin{center}
\includegraphics[totalheight=1.0in]{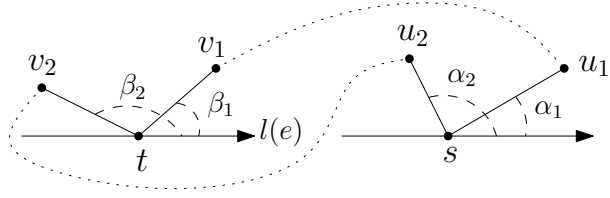}
\caption{\footnotesize Illustrating the definitions of the angles.
}
\label{fig:twoangles}
\end{center}
\end{minipage}
\vspace*{-0.15in}
\end{figure}

Note that $v_1$ and $v_2$ are in a closed half-plane bounded by the line $l(e)$.
We assign a direction to
$l(e)$ such that each of $v_1$ and $v_2$ are to the left or on $l(e)$.
Define $\beta_i$ as the smallest angle to rotate $l(e)$ counterclockwise such
that the direction of $l(e)$ becomes the same as the direction $\vector{tv_i}$,
for each $i=1,2$ (e.g., see Fig.~\ref{fig:twoangles}). Hence, both $\beta_1$ and
$\beta_2$ are in $[0,\pi]$.
Without of loss of generality, we assume $\beta_1\leq \beta_2$ (otherwise the analysis is symmetric). Since $\{v_1,v_2\}$ contains at least one
vertex in each of the open half-planes bounded by the line through $t$ and
perpendicular to $e$, we have $\beta_1\in [0,\pi/2)$ and $\beta_2\in
(\pi/2,\pi]$. Further, since at least one of $v_1$ and $v_2$ is not on $l(e)$, it is not possible that both $\beta=0$ and $\beta=\pi$ hold.


Let $\alpha = \alpha_2-\alpha_1$. We refer to the case where $\beta_1+\beta_2=\pi$ and $\alpha=\pm\pi$ (i.e., $\alpha$ is $\pi$ or $-\pi$) as the {\em special case}. In the special case, $s$ is on $\overline{u_1u_2}$ and the vertical line through $t$ and perpendicular to $l(e)$ bisects the angle $\angle v_1tv_2$.


\begin{lemma}\label{lem:80}
The $\pi$-range $\piR(s,t)$ is determined as follows (e.g., see Fig.~\ref{fig:edgecase}).
\begin{equation*}
\piR(s,t) = \begin{cases}
(\alpha_1-\arctan(\frac{\lambda}{\sin(\alpha)}), \alpha_1-\arctan(\frac{\lambda}{\sin(\alpha)})+\pi)
& \text{if $\sin(\alpha)>0$},\\
(\alpha_1-\arctan(\frac{\lambda}{\sin(\alpha)})-\pi, \alpha_1-\arctan(\frac{\lambda}{\sin(\alpha)})) & \text{if $\sin(\alpha)<0$},\\
(\alpha_1-\pi/2,\alpha_1+\pi/2) & \text{if $\sin(\alpha)=0$ and $\lambda>0$},\\
(\alpha_1-3\pi/2,\alpha_1-\pi/2) & \text{if $\sin(\alpha)=0$ and $\lambda<0$}, \\
\emptyset & \text{if $\sin(\alpha)=0$ and $\lambda=0$},\\
\end{cases}
\end{equation*}
where $\lambda=\cos\alpha-\frac{\cos\beta_2}{\cos\beta_1}$. Further, $\alpha=\pm\pi$ and $\beta_1+\beta_2=\pi$ (i.e., the special case) if and only if $\sin(\alpha)=0$ and $\lambda=0$.
\end{lemma}

\begin{figure}[t]
\begin{minipage}[t]{\linewidth}
\begin{center}
\includegraphics[totalheight=1.6in]{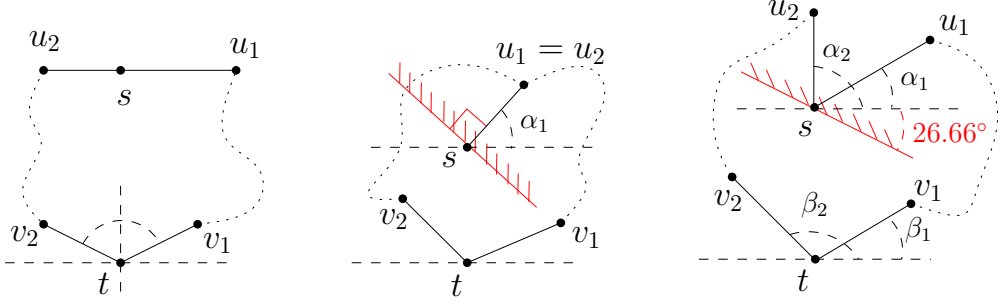}
\caption{\footnotesize Illustrating several concrete examples for Lemma~\ref{lem:80}. Left: the special case (the vertical line through $t$ bisects $\angle v_1tv_2$); in this case, $\piR(s,t)=\emptyset$. Middle: the case where $\alpha_1=\alpha_2$, and thus $\alpha=0$, $\sin(\alpha)=0$, and $u_1=u_2$; in this case, $\piR(s,t)=(\alpha_1-\pi/2,\alpha_1+\pi/2)$, which is delimited by the open half-plane (marked with red color in the figure) bounded by the line through $s$ and perpendicular to $\overline{su_1}$. Right: the most general case where $\alpha_1=30^{\circ}$, $\alpha_2=90^{\circ}$, $\beta_1=30^{\circ}$, $\beta_2=135^{\circ}$; by calculation, $\lambda\approx1.3165$, $\arctan(\frac{\gamma}{\sin(\alpha)})\approx56.66^{\circ}$, and thus, $\piR(s,t)\approx(\alpha_1-56.66^{\circ},\alpha_1+56.66^{\circ})=(-26.66^{\circ},153.34^{\circ})$; the open half-plane that delimits $\piR(s,t)$ is marked with red color in the figure. }
\label{fig:edgecase}
\end{center}
\end{minipage}
\vspace*{-0.15in}
\end{figure}

\paragraph{Remark:}
Unlike Lemma \ref{lem:60} whose geometric intuition is very straightforward, it is not clear to us how to interpret Lemma~\ref{lem:80} intuitively. 

We defer the proof of Lemma \ref{lem:80} to Section \ref{sec:lem80proof}.
According to Lemma \ref{lem:80}, if the special case happens, $\piR(s,t)$ is
empty; otherwise, it is an open range of size exactly $\pi$.
Since $\alpha=0$ if and only if $u_1=u_2$, the case $u_1=u_2$ is also covered by the lemma.


Now we are back to our original problem to determine the range
$R(s,t)$ for a non-degenerate farthest point $t\in E$ of $s$.
By Observation \ref{obser:20}, $s$ is not visible to $t$. Further, $s$
and $t$ have exactly two shortest paths $\pi_1(s,t)$ and $\pi_2(s,t)$.
Clearly, by Observation \ref{obser:10}(2), the two paths are canonical.
Therefore, the $\pi$-range $\piR(s,t)$ of $s$ with respect to $t$ and the two shortest paths $\pi_1(s,t)$ and $\pi_2(s,t)$ can be
determined by Lemma \ref{lem:80}.

By using Lemma \ref{lem:40}(2), we have the following lemma.

\begin{lemma}\label{lem:90}
$R(s,t) = \piR(s,t) \cap R_f(s)$.
\end{lemma}
\begin{proof}
For each $i=1,2$, let $u_i$ and $v_i$ be the $s$-pivot and $t$-pivot
of the shortest path $\pi_i(s,t)$, respectively. Since $s$ and $t$ have only
two shortest paths $\pi_1(s,t)$ and $\pi_2(s,t)$,
$U_t(s)=\{v_1,v_2\}$ and $U_s(t)=\{u_1,u_2\}$. Further, for each
$i=1,2$, $v_i$ has only one $s$-pivot, which is $u_i$.

\begin{enumerate}
\item
Consider any direction $r_s\in R(s,t)$. Clearly, $r_s\in R_f(s)$.
By Lemma \ref{lem:40}(2), there exists a direction $r_t$ parallel to $e$ for
$t$ with a speed $\tau\geq 0$ such that if we move $s$ along $r_s$ for unit speed
and move $t$ along $l_t$ with speed $\tau$,
each vertex $v\in U_t(s)$ has a coupled $s$-pivot $u$ with $d'_{u,v}(s,t)<0$.
Since $U_t(s)=\{v_1,v_2\}$ and for each $i=1,2$, $v_i$ has only one $s$-pivot $u_i$, it holds that $d'_{u_i,v_i}(s,t)<0$. Hence, $r_s$
is in $\piR(s,t)$.


\item
Consider any $r_s$ in $R_f(s,t)\cap \piR(s,t)$. First of all, $r_s$ is a free
direction. Since $r_s$ is in $\piR(s,t)$, there exists a direction $r_t$ parallel to $e$ for
$t$ with a speed $\tau\geq 0$ such that if we move $s$ along $r_s$ for unit speed
and move $t$ along $r_t$ with speed $\tau$, $d'_{u_i,v_i}(s,t)<0$ for each
$i=1,2$. Since $U_t(s)=\{v_1,v_2\}$, according to Lemma \ref{lem:40}(2), $r_s$
is in $R(s,t)$.

\end{enumerate}
The lemma thus follows.
\end{proof}

Suppose $t$ is the only farthest point of $s$ and $t$ is non-degenerate.
According to Lemma \ref{lem:80}, if the special case happens,
$\piR(s,t)=\emptyset$, and thus $R(s,t)=\emptyset$ by Lemma \ref{lem:90}.
By Corollary \ref{corr:10},
whenever we move $s$ along any free direction
infinitesimally, the value $d_{\max}(s)$ will be strictly increasing. Therefore,
it is possible that the point $s$, which is in $\calI$ and
has only one farthest point, is a
geodesic center. It is not difficult to
construct such an example by following the left figure of
Fig.~\ref{fig:edgecase}; e.g., see Fig.~\ref{fig:interiorcenter}.
Hence, we have the following corollary.

\begin{corollary}\label{corr:20}
It is possible that a geodesic center is in $\calI$ and has only one farthest point.
\end{corollary}

\begin{figure}[t]
\begin{minipage}[t]{\linewidth}
\begin{center}
\includegraphics[totalheight=1.5in]{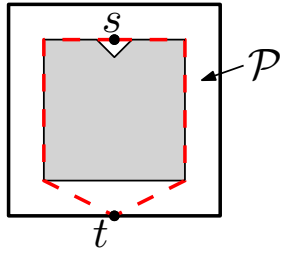}
\caption{Illustrating an example in which a geodesic center $s$ is in $\calI$ and has only one farthest point $t$. The polygonal domain $\calP$ is between two (very close) concentric squares plus an additional (very small) triangle so that $s$ is in $\calI$. The point $s$ is at the middle of the top edge of the inner square, and $t$ is at the middle of the bottom edge of the outer square. One can verify that $s$ is a geodesic center and $t$ is the only farthest point of $s$. The two shortest paths from $s$ to $t$ are shown with red dashed segments. Note that the middle point of every edge of the inner square is a geodesic center.}
\label{fig:interiorcenter}
\end{center}
\end{minipage}
\vspace*{-0.15in}
\end{figure}

\subsection{Proof of Lemma \ref{lem:80}}
\label{sec:lem80proof}


Consider any direction $r_s$ for moving $s$. Let $\theta_s$ denote the angle of
$r_s$. Recall that the moving direction $r_t$ for $t$ is parallel to $e$ and
we have assigned a direction to the line $l(e)$.
Let $\theta_t$ denote the smallest angle to rotate $l(e)$ such that $l(e)$
becomes the same direction as $r_t$ (thus the definition of $\theta_t$ is
``consistent'' with the definitions of $\beta_1$ and $\beta_2$).
Since $r_t$ is parallel to $e$, $\theta_t$
is either $0$ or $\pi$.

According to Equation \eqref{equ:10}, we can obtain the derivatives of the two
functions $d_{u_1,v_1}(s,t)$ and $d_{u_2,v_2}(s,t)$ as follows
\begin{equation}
\label{equ:40}
\begin{cases}
d'_{u_1,v_1}(s,t)= -\cos(\alpha_1-\theta_s) - \tau\cdot \cos(\beta_1-\theta_t),\\
d'_{u_2,v_2}(s,t)= -\cos(\alpha_2-\theta_s) - \tau\cdot \cos(\beta_2-\theta_t).
\end{cases}
\end{equation}

Therefore, for each $i=1,2$, $d'_{u_i,v_i}(s,t)$ is a function of $\theta_s$,
$\theta_t$, and $\tau$. In order to simplify our proof for Lemma \ref{lem:80},
we first give the following lemma.

\begin{lemma}\label{lem:100}
A direction $r_s$ is in $\piR(s,t)$ if and only if there exist 
$\theta_t\in \{0,\pi\}$ and $\tau\geq 0$ such that $d'_{u_1,v_1}(s,t)=0$
and $d'_{u_2,v_2}(s,t)<0$ (the same also holds if we exchange the indices $1$ and $2$).
\end{lemma}
\begin{proof}
Given any direction $r_s$ for $s$, the angle $\theta_s$ is fixed, and thus
for each $i=1,2$, $d'_{u_i,v_i}(s,t)$ is a function of $\theta_t$ and $\tau$.
In the following, we will use $d_i'(\theta_t,\tau)$ to represent
$d'_{u_i,v_i}(s,t)$ for each $i=1,2$.

We first prove one direction of the lemma.
Assume $d_1'(\theta_t,\tau)=0$ and $d_2'(\theta_t,\tau)<0$ for some $\theta_t\in
\{0,\pi\}$ and $\tau\geq 0$. Our goal is to show that $r_s$ is in $\piR(s,t)$. To
this end, it is sufficient to find another pair $(\theta_t',\tau')$ with
$\theta_t'\in \{0,\pi\}$ and $\tau'\geq 0$ such that $d_1'(\theta'_t,\tau')<0$ and
$d_2'(\theta'_t,\tau')<0$.

Recall that neither $\beta_1$ nor $\beta_2$ can be $\pi/2$. Since $\theta_t$ is
either $0$ or $\pi$, $\beta_1-\theta_t\neq \pm\pi/2$, and thus
$\cos(\beta_1-\theta_t)\neq 0$.
Since $d'_1(\theta_t,\tau)=0$ and $d_2'(\theta_t,\tau)<0$, if $\tau>0$, regardless of whether
$\cos(\beta_1-\theta_t)$ is positive or negative, we can always change $\tau$
infinitesimally to a new value $\tau'>0$ such that
$d_1'(\theta_t,\tau')<0$ and $d_2'(\theta_t,\tau')<0$.

If $\tau=0$, then by Equation \eqref{equ:40}, $-\cos(\alpha_1-\theta_s)=0$
and  $-\cos(\alpha_2-\theta_s)<0$.
We let $\tau'$ be an infinitesimally small positive value and let
$\theta_t'$ be $0$. Since $\beta_1\in [0,\pi/2)$, we have
$\cos(\beta_1-\theta_t')>0$, and thus $d_1'(\theta'_t,\tau')<0$ due to
$-\cos(\alpha_1-\theta_s)=0$. Further, since $\tau'$ is an infinitesimal small
positive value and $-\cos(\alpha_2-\theta_s)<0$, we have $d_2'(\theta'_t,\tau')<0$.

This proves that $r_s$ must be in $\piR(s,t)$.

We proceed to prove the other direction of the lemma. Assume $r_s$ is in $\piR(s,t)$.
By the definition of $\piR(s,t)$, there exist $\theta_t\in \{0,\pi\}$ and $\tau\geq 0$ such that
$d_1'(\theta_t,\tau)<0$ and $d_2'(\theta_t,\tau)<0$.
Our goal is to find another pair $(\theta'_t,\tau')$ with
$\theta'_t\in \{0,\pi\}$ and $\tau'\geq 0$ such that
$d_1'(\theta'_t,\tau')=0$ and $d_2'(\theta'_t,\tau')<0$.

Recall that $\cos(\beta_1-\theta_t)\neq 0$.
Depending on whether $\cos(\beta_1-\theta_t)$ is positive or negative, there
are two cases.

If $\cos(\beta_1-\theta_t)<0$, then $\cos(\beta_2-\theta_t)>0$ because
$\beta_1\in[0,\pi/2)$, $\beta_2\in (\pi/2,\pi]$, and $\theta_t\in
\{0,\pi\}$. Since
$\tau\geq 0$, if we
increase the value $\tau$, $d_1'(\theta_t,\tau)$ will strictly increase and
$d_1'(\theta_t,\tau)$ will strictly decrease. Hence, if we keep increasing
$\tau$, there must be a moment when
$d_1'(\theta_t,\tau)=0$ and $d_1'(\theta_t,\tau)<0$. We are done with the proof.

If $\cos(\beta_1-\theta_t)>0$, then $\cos(\beta_2-\theta_t)<0$. Depending on
whether $\tau=0$, there are two subcases.

\begin{itemize}
\item

If $\tau=0$,
let $\theta_t'=(\theta_t+\pi \mod 2\pi)$ (i.e., we reverse the moving
direction of $t$). Consequently, we obtain $\cos(\beta_1-\theta'_t)<0$ and
$\cos(\beta_2-\theta'_t)>0$. Then, we can use the same approach as
above, i.e., keep
increasing $\tau$ until $d_1'(\theta'_t,\tau)=0$ and
$d_1'(\theta'_t,\tau)<0$.

\item
If  $\tau>0$, then if we decrease $\tau$, $d_1'(\theta_t,\tau)$ increases
and $d_2'(\theta_t,\tau)$ decreases. We keep decreasing the value $\tau$ until
one of the two events happens: $\tau=0$ and $d_1'(\theta_t,\tau)=0$.

Whichever event happens first, it always holds that
$d_2'(\theta_t,\tau)<0$. If $d_1'(\theta_t,\tau)=0$ happens first (or
both events happen simultaneously), we
are done with the proof. Otherwise, we obtain
$\tau=0$ and both $d_1'(\theta_t,\tau)$ and $d_2'(\theta_t,\tau)$ are
negative. Then, we can use the same approach as the above case for
$\tau=0$ to prove.
\end{itemize}

This completes the proof for the lemma.
\end{proof}

To simplify the notation, let $w_1=-d'_{u_1,v_1}(s,t)$ and
$w_2=-d'_{u_2,v_2}(s,t)$. Once $r_s$ is fixed (and thus $\theta_s$ is
fixed), both $w_1$ and $w_2$
are implicitly considered as functions of $\theta_t\in\{0,\pi\}$ and $\tau\geq 0$.
By Lemma \ref{lem:100}, $\piR(s,t)$ consists of all directions $r_s$
for $s$ such that there exist $\theta_t\in\{0,\pi\}$ and $\tau\geq 0$
with $w_1=0$ and $w_2>0$.

Let $x=\alpha_1-\theta_s$. Recall that $\alpha=\alpha_2-\alpha_1$.
Then, $\alpha_2-\theta_s=x+\alpha$. Thus, we have
\begin{equation*}
\begin{cases}
w_1=\cos(x) +\tau\cdot \cos(\beta_1-\theta_t),\\
w_2=\cos(x+\alpha) +\tau\cdot \cos(\beta_2-\theta_t).
\end{cases}
\end{equation*}

First of all, let $w_1=0$.
As shown in the proof of Lemma \ref{lem:100}, since $\beta_1\in [0, \pi/2)$ and  $\theta_t\in \{0,\pi\}$, it holds that $\cos(\beta_1-\theta_t)\neq 0$. Consequently,
we have $$\tau=-\frac{\cos(x)}{\cos(\beta_1-\theta_t)}.$$
Since $\tau\geq 0$ and $\beta_1\in [0,\pi/2)$, we obtain the following:
if $\cos(x)>0$, then $\theta_t=\pi$; if
$\cos(x)<0$, then $\theta_t=0$; if $\cos(x)=0$, then $\theta_t$ can be either
$0$ or $\pi$.

By substituting $\tau$ with $-\frac{\cos(x)}{\cos(\beta_1-\theta_t)}$ in $w_2$ and using the angle sum identities of trigonometric functions, we have
\begin{equation*}
\begin{split}
w_2 & = \cos(x+\alpha) + \tau\cdot \cos(\beta_2-\theta_t) \\
    & = \cos(x)\cdot \cos(\alpha) - \sin(x)\cdot \sin(\alpha) -
    \frac{\cos(x)}{\cos(\beta_1-\theta_t)}\cdot \cos(\beta_2-\theta_t)\\
    & = \cos(x)\cdot \cos(\alpha) - \sin(x)\cdot \sin(\alpha) -
    \cos(x)\cdot \frac{\cos(\beta_2)}{\cos(\beta_1)}.
\end{split}
\end{equation*}
The last equality in the above equation is because regardless of whether
$\theta_t$ is $0$ or $\pi$, it always holds that
$\frac{\cos(\beta_2-\theta_t)}{\cos(\beta_1-\theta_t)}=\frac{\cos(\beta_2)}{\cos(\beta_1)}$.
Recall that $\lambda=\cos(\alpha)-\frac{\cos(\beta_2)}{\cos(\beta_1)}$ in the
statement of Lemma \ref{lem:80}. Hence, $w_2=\cos(x)\cdot \lambda - \sin(x)\cdot
\sin(\alpha)$. Therefore, $w_2>0$ is equivalent to $\cos(x)\cdot
\lambda>\sin(x)\cdot \sin(\alpha)$.

The above discussion shows the following observation.

\begin{observation}
$w_1=0$ and $w_2>0$ if and only if the following holds:
\begin{equation}\label{equ:50}
\cos(x)\cdot\lambda>\sin(x)\cdot \sin(\alpha).
\end{equation}
\end{observation}

Depending on whether $\sin(\alpha)$ is positive, negative, or zero, there are three cases.

\begin{enumerate}
\item
If $\sin(\alpha)>0$, depending on whether $\cos(x)$ is positive, negative, or
zero, there are further three subcases.

\begin{enumerate}
\item
If $\cos(x)>0$, then $x\in (-\pi/2,\pi/2)$ and
Inequality \eqref{equ:50} is equivalent to $\tan(x)
< \frac{\lambda}{\sin(\alpha)}$.
Thus, $w_1=0$ and $w_2>0$ if and only if $x\in (-\pi/2,
\arctan(\frac{\lambda}{\sin(\alpha)}))$.
\item
If $\cos(x)<0$, then $x\in (-3\pi/2,-\pi/2)$ and
Inequality \eqref{equ:50} is equivalent to $\tan(x)
> \frac{\lambda}{\sin(x)}$.
Thus, $w_1=0$ and $w_2>0$ if and only if $x\in
(\arctan(\frac{\lambda}{\sin(\alpha)})-\pi,-\pi/2)$.

\item
If $\cos(x)=0$, then $x=\pm\frac{\pi}{2}$ and
Inequality \eqref{equ:50} is equivalent to
$\sin(x) < 0$, which further implies that $x$ must be $-\pi/2$.
Thus, $w_1=0$ and $w_2>0$ if and only
if $x=-\pi/2$.
\end{enumerate}

Combining the above discussions for the case $\sin(\alpha)>0$,
$w_1=0$ and $w_2>0$ if and only if $x\in
(\arctan(\frac{\lambda}{\sin(\alpha)})-\pi,\arctan(\frac{\lambda}{\sin(\alpha)}))$.

\item
If $\sin(\alpha)<0$, the analysis is symmetric.
Depending on whether $\cos(x)$ is positive, negative, or
zero, there are further three subcases.

\begin{enumerate}
\item
If $\cos(x)>0$, then $x\in (-\pi/2,\pi/2)$ and
Inequality \eqref{equ:50} is equivalent to $\tan(x)
> \frac{\lambda}{\sin(\alpha)}$.
Thus, $w_1=0$ and $w_2>0$ if and only if $x\in
(\arctan(\frac{\lambda}{\sin(\alpha)}), \pi/2)$.

\item
If $\cos(x)<0$, then $x\in (\pi/2,3\pi/2)$ and
Inequality \eqref{equ:50} is equivalent to $\tan(x)
< \frac{\lambda}{\sin(x)}$.
Thus, $w_1=0$ and $w_2>0$ if and only if $x\in
(\pi/2,\arctan(\frac{\lambda}{\sin(\alpha)})+\pi)$.

\item
If $\cos(x)=0$, then $x=\pm\frac{\pi}{2}$ and
Inequality \eqref{equ:50} is equivalent to
$\sin(x) > 0$, which further implies that $x$ must be $\pi/2$.
Thus, $w_1=0$ and $w_2>0$ if and only
if $x=\pi/2$.
\end{enumerate}

Combining the above discussions for the case $\sin(\alpha)>0$,
$w_1=0$ and $w_2>0$ if and only if $x\in
(\arctan(\frac{\lambda}{\sin(\alpha)}),\arctan(\frac{\lambda}{\sin(\alpha)})+\pi)$.

\item
If $\sin(\alpha)=0$, then
Inequality \eqref{equ:50} is equivalent to
$\cos(x) \cdot \lambda > 0$. Depending on whether $\lambda$ is positive, negative,
or zero, there are further three subcases.

\begin{enumerate}
\item
If $\lambda>0$, then Inequality \eqref{equ:50} is equivalent to
$\cos(x) > 0$, implying that $x\in (-\pi/2,\pi/2)$.

\item
If $\lambda<0$, then Inequality \eqref{equ:50} is equivalent to
$\cos(x) < 0$, implying that $x\in (\pi/2,3\pi/2)$.

\item
If $\lambda=0$, then Inequality \eqref{equ:50} is equivalent to
$0 < 0$, which is not possible for any $x$.

\end{enumerate}

\end{enumerate}

As a summary, $w_1=0$ and $w_2>0$ if and only if: for $\sin(\alpha)>0$, $x\in (\arctan(\frac{\lambda}{\sin(\alpha)})-\pi,\arctan(\frac{\lambda}{\sin(\alpha)}))$;
for $\sin(\alpha)>0$, $x\in (\arctan(\frac{\lambda}{\sin(\alpha)}),\arctan(\frac{\lambda}{\sin(\alpha)})+\pi)$;
for $\sin(\alpha)=0$ and $\lambda > 0$, $x\in (-\pi/2,\pi/2)$;
for $\sin(\alpha)=0$ and $\lambda < 0$, $x\in (\pi/2,3\pi/2)$;
for $\sin(\alpha)=0$ and $\lambda = 0$, $x\in \emptyset$.

Recall that $x=\alpha_1-\theta_s$. By Lemma \ref{lem:100}, we obtain the $\pi$-range $\piR(s,t)$ as follows.
\begin{equation*}
\piR(s,t) = \begin{cases}
(\alpha_1-\arctan(\frac{\lambda}{\sin(\alpha)}), \alpha_1-\arctan(\frac{\lambda}{\sin(\alpha)})+\pi)
& \text{if $\sin(\alpha)>0$},\\
(\alpha_1-\arctan(\frac{\lambda}{\sin(\alpha)})-\pi, \alpha_1-\arctan(\frac{\lambda}{\sin(\alpha)})) & \text{if $\sin(\alpha)<0$},\\
(\alpha_1-\pi/2,\alpha_1+\pi/2) & \text{if $\sin(\alpha)=0$ and $\lambda>0$},\\
(\alpha_1-3\pi/2,\alpha_1-\pi/2) & \text{if $\sin(\alpha)=0$ and $\lambda<0$}, \\
\emptyset & \text{if $\sin(\alpha)=0$ and $\lambda=0$}.\\
\end{cases}
\end{equation*}

We complete the proof of Lemma \ref{lem:80} by showing the following claim:
$\alpha=\pm\pi$ and $\beta_1+\beta_2=\pi$ if and only if $\sin(\alpha)=0$ and $\lambda =
0$. We prove the claim below.

Assume $\alpha=\pm\pi$ and $\beta_1+\beta_2=\pi$. Then,
$\sin(\alpha)=0$, and
$\lambda=\cos(\alpha)-\frac{\cos(\beta_2)}{\cos(\beta_1)}=-1-\frac{\cos(\beta_2)}{\cos(\beta_1)}$.
Since $\beta_1+\beta_2=\pi$, $\cos(\beta_1)=-\cos(\beta_2)$. Hence,
$\lambda=0$.

On the other hand, assume $\sin(\alpha)=0$ and $\lambda=0$. Then,
$\alpha$ is $0$ or $\pm \pi$. We claim that $\alpha$ cannot be zero. Indeed, suppose to the contrary that $\alpha =0$. Since $\beta_1\in [0,\pi/2)$ and $\beta_2\in (\pi/2,\pi]$, it always holds that $\frac{\cos(\beta_2)}{\cos(\beta_1)}<0$. Hence, $\lambda=\cos(\alpha)-\frac{\cos(\beta_2)}{\cos(\beta_1)}=1-\frac{\cos(\beta_2)}{\cos(\beta_1)}$, which is always positive, contradicting with $\lambda =0$.

The above proves that $\alpha = \pm\pi$. Then, we have $\lambda=-1-\frac{\cos(\beta_2)}{\cos(\beta_1)}$. Due to $\lambda = 0$, it holds that $\cos(\beta_2)=-\cos(\beta_1)$, which implies that $\beta_1+\beta_2=\pi$ since $\beta_1\in [0,\pi/2)$ and $\beta_2\in (\pi/2,\pi]$.

\subsection{The Case $t\in \calI$}
\label{sec:three}

The analysis for this case is substantially more difficult than the
case $t\in E$. As before, we first present a more general result that
is on
three shortest path distance functions.

Let $s$ and $t$ be any two points in $\calP$ such that
$t$ is in $\calI$ and there are three shortest \st\ paths $\pi_1(s,t)$,
$\pi_2(s,t)$, and $\pi_3(s,t)$ (this implies that $s$ is not visible
to $t$). For each $i=1,2,3$, let
$\pi_i(s,t)=\pi_{u_i,v_i}(s,t)$, i.e., $u_i$ and $v_i$ are the
$s$-pivot and $t$-pivot of $\pi_i(s,t)$, respectively.
We say that the three paths are {\em canonical} with respect to $s$ and $t$ if
they have the following two properties.
\begin{enumerate}
\item
$t$ is in the interior of the triangle $\triangle
v_1v_2v_3$.
\item Suppose we reorder the indices such that $v_1$, $v_2$, and $v_3$ are
{\em clockwise} around $t$, then $u_1$, $u_2$, and $u_3$ are {\em
counterclockwise} around $s$ (e.g., see Fig.~\ref{fig:rangeproperty}).
\end{enumerate}

The above first property implies that $v_1$, $v_2$, and $v_3$ are
distinct, but this may not be true for $u_1,u_2,u_3$.
In the following, we assume that the three shortest paths $\pi_i(s,t)$ with $1\leq
i\leq 3$ are canonical, and we reorder the indices such that
$v_1$, $v_2$, and $v_3$ are clockwise around $t$ and $u_1$, $u_2$, and $u_3$ are
counterclockwise around $s$.
For each $i=1,2,3$,  we consider
$d_{u_i,v_i}(s,t)=|su_i|+d(u_i,v_i)+|v_it|$ as a function of $s\in
\bbR^2$ and $t\in \bbR^2$.

In this case, the {\em $\pi$-range} $\piR(s,t)$ of $s$ is defined with respect to $t$
and the three paths $\pi_i(s,t)$ for $i=1,2,3$ as follows:
a direction $r_s$ for $s$ is in $\piR(s,t)$ if there exists a
direction $r_t$ for $t$
with a speed $\tau\geq 0$ such that when we move $s$ along $r_s$ with
unit speed and move $t$ along $r_t$ with speed $\tau$,
$d'_{u_i,v_i}<0$ holds for $i=1,2,3$.

As Lemma \ref{lem:80} in the previous cases,
we will have a similar lemma (Lemma \ref{lem:110}), which says that unless a
special case happens the range $\piR(s,t)$ is
an open range of size exactly $\pi$. The proof is much more
challenging.  Before presenting Lemma \ref{lem:110},
we introduce some notation.


Recall the definitions of the angles of directions.
For each $i=1,2,3$, let $\beta_i$ denote the angle of the direction
$\vector{tv_i}$ (i.e., the angle of $\vector{tv_i}$ counterclockwise from the positive
$x$-axis). Further, we define three angles $b_i$ for $i=1,2,3$
as follows (e.g., see Fig.~\ref{fig:rangeproperty}).
Define $b_1$ as the smallest angle we need to rotate the direction
$\vector{tv_1}$ {\em clockwise} to $\vector{tv_2}$;
define $b_2$ as the smallest angle we need to rotate the direction
$\vector{tv_2}$ clockwise to $\vector{tv_3}$;
define $b_3$ as the smallest angle we need to rotate the direction
$\vector{tv_3}$ clockwise to $\vector{tv_1}$.

For any two angles $\alpha'$ and $\alpha''$, we use $\alpha'\equiv\alpha''$ to
denote $\alpha'=\alpha''\mod 2\pi$.

It is easy to see that $b_1\equiv\beta_1-\beta_2$, $b_2\equiv\beta_2-\beta_3$, and
$b_3\equiv\beta_3-\beta_1$.
Note that since $t$ is in the interior of $\triangle v_1v_2v_3$, it
holds that $b_i\in (0,\pi)$ for $i=1,2,3$. Note that $b_1+b_2+b_3=2\pi$.

For each $i=1,2,3$, let $\alpha_i$ denote the angle of the direction
$\vector{su_i}$. According to our definition of pivot vertices, $u_i=u_j$ if and only if
$\alpha_i=\alpha_j$ for any two $i,j\in \{1,2,3\}$.
We define three angles $a_i$ for $i=1,2,3$ as follows (e.g., see
Fig.~\ref{fig:rangeproperty}).
Define $a_1$ as the smallest angle we need to rotate the direction
$\vector{su_1}$ {\em counterclockwise} to $\vector{su_2}$;
define $a_2$ as the smallest angle we need to rotate the direction
$\vector{su_2}$ clockwise to $\vector{su_3}$;
define $a_3$ as the smallest angle we need to rotate the direction
$\vector{su_3}$ clockwise to $\vector{su_1}$.
Hence, $a_1\equiv\alpha_2-\alpha_1$, $a_2\equiv\alpha_3-\alpha_2$, and
$a_3\equiv\alpha_1-\alpha_3$.

We refer to the case where $a_i=b_i$ for each $i=1,2,3$ as the {\em special
case}.

\begin{lemma}\label{lem:110}
The $\pi$-range $\piR(s,t)$ is determined as follows (e.g., see Fig.~\ref{fig:innercase}).
\begin{equation*}
\piR(s,t) = \begin{cases}
(\alpha_1-\arctan(\frac{\delta_1-\delta_2}{\delta}),
\alpha_1-\arctan(\frac{\delta_1-\delta_2}{\delta})+\pi)
& \text{if $\delta>0$},\\
(\alpha_1-\arctan(\frac{\delta_1-\delta_2}{\delta})-\pi,
\alpha_1-\arctan(\frac{\delta_1-\delta_2}{\delta})) & \text{if $\delta<0$},\\
(\alpha_1-\pi/2,\alpha_1+\pi/2) & \text{if $\delta=0$ and $\delta_1>\delta_2$}, \\
(\alpha_1-3\pi/2,\alpha_1-\pi/2) & \text{if $\delta=0$ and $\delta_1<\delta_2$}, \\
\emptyset & \text{if $\delta=0$ and $\delta_1=\delta_2$}, \\
\end{cases}
\end{equation*}
where $\delta=\frac{\sin(\alpha_3-\alpha_1)}{\sin(\beta_3-\beta_1)}-
\frac{\sin(\alpha_2-\alpha_1)}{\sin(\beta_2-\beta_1)}$,
$\delta_1=\frac{\cos(\beta_2-\beta_1)-\cos(\alpha_2-\alpha_1)}{\sin(\beta_2-\beta_1)}$,
and
$\delta_2=\frac{\cos(\beta_3-\beta_1)-\cos(\alpha_3-\alpha_1)}{\sin(\beta_3-\beta_1)}$.
Further, $a_i=b_i$ for each $i=1,2,3$ (i.e., the special case) if and only if $\delta  =0$ and $\delta_1=\delta_2$.
\end{lemma}

\begin{figure}[t]
\begin{minipage}[t]{\linewidth}
\begin{center}
\includegraphics[totalheight=1.5in]{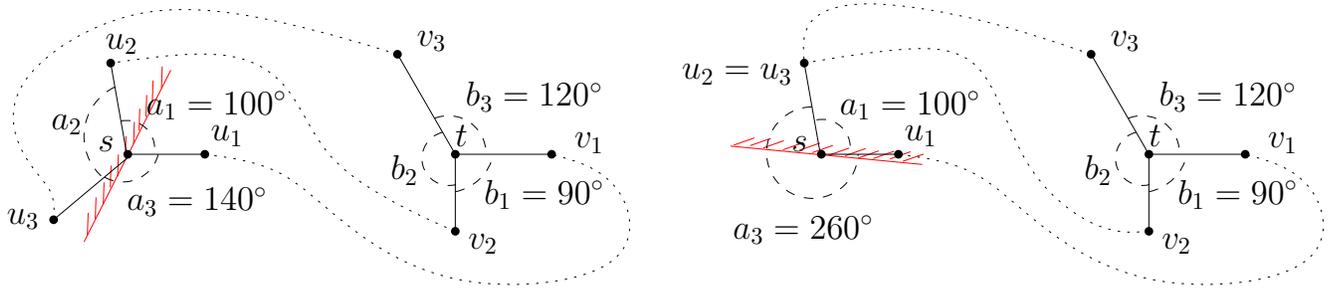}
\caption{\footnotesize Illustrating two concrete examples for Lemma~\ref{lem:110}. Left: The sizes of the angles of $a_i$ and $b_i$ for $1\leq i\leq 3$ are already shown in the figure with $\alpha_1=0$. By calculation, $\delta\approx 0.2426$, $\delta_1\approx-0.1736$, $\delta_2\approx 0.3072$, $\arctan(\frac{\delta_1-\delta_2}{\delta})\approx-63.23^{\circ}$, and thus $\piR(s,t)\approx(\alpha_1+63.23^{\circ},\alpha_1+63.23^{\circ}+180^{\circ})=(63.23^{\circ},243.23^{\circ})$. Right: a case where $\alpha_2=\alpha_3$ with $\alpha_1=0$. Thus, $a_2=0$ and $u_2=u_3$. The sizes of other angles are already shown in the figure. By calculation, $\delta\approx 2.1220$, $\delta_1\approx-0.1736$, $\delta_2\approx -0.3768$, $\arctan(\frac{\delta_1-\delta_2}{\delta})\approx 5.47^{\circ}$, and thus $\piR(s,t)\approx(\alpha_1-5.47^{\circ},\alpha_1-5.47^{\circ}+180^{\circ})=(-5.47^{\circ},174.53^{\circ})$.
The open half-planes that delimit $\piR(s,t)$ in both examples are marked with red color.
}
\label{fig:innercase}
\end{center}
\end{minipage}
\vspace*{-0.15in}
\end{figure}

We defer the proof of Lemma \ref{lem:110} to Section \ref{sec:lem110proof}. According to Lemma \ref{lem:110}, if the special case happens, then $\piR(s,t)$ is empty; otherwise, it is an open range of size exactly $\pi$.

Now we are back to our original problem to determine the range
$R(s,t)$ for a non-degenerate farthest point $t\in \calI$ of $s$. Since there
are exactly three shortest \st\ paths $\pi_{i}(s,t)$ for $i=1,2,3$, the three
paths must be canonical. To see this, by
Observation \ref{obser:10}, $t$ is in the interior of $\triangle
v_1v_2v_3$. Further, it is easy to see that no two of the three paths cross
each other 
since otherwise there would be more than three shortest
\st\ paths, this implies that the second property of the canonical paths holds.
Let $\piR(s,t)$ be the $\pi$-range of $s$ with respect
to $t$ and the above three shortest paths. We have the following result.

\begin{lemma}\label{lem:120}
$R(s,t)=\piR(s,t)\cap R_f(s)$.
\end{lemma}
\begin{proof}
For each $i=1,2,3$, let $u_i$ and $v_i$ be the $s$-pivot and $t$-pivot
of the shortest path $\pi_i(s,t)$, respectively. Since $s$ and $t$
have only three shortest paths, $U_t(s)=\{v_1,v_2,v_3\}$ and
$U_s(t)=\{u_1,u_2,u_3\}$. Further, for each $i=1,2,3$, $v_i$ have only
one $s$-pivot, which is $u_i$.

\begin{enumerate}
\item
Consider any $r_s\in R(s,t)$. Clearly, $r_s\in R_f(s)$. By Lemma \ref{lem:40}(1), there is a
direction $r_t$ for $t$ with a speed $\tau\geq 0$ such that when $s$
moves along $r_s$ with unit speed and $t$ moves along $r_t$ with speed
$\tau$, each $v\in U_t(s)$ has a coupled $s$-pivot $u$ with
$d'_{u,v}(s,t)<0$. Since $U_t(s)=\{v_1,v_2,v_3\}$ and each $v_i$ has
only one $s$-pivot $u_i$ for each $1\leq i\leq 3$, we have $d'_{u_i,v_i}(s,t)<0$ for
$i=1,2,3$. Hence, $r_s$ is in $\piR(s,t)$.

\item
Consider any $r_s$ in $\piR(s,t)\cap R_f(s)$. First of all,
$r_s$ is a free direction. Since $r_s$ is in $\piR(s,t)$, there exists
a direction $r_r$ for $t$ with a speed $\tau\geq 0$ such that if $s$
moves along $r_s$ for unit speed and $t$ moves along $r_t$ with speed
$\tau$, then $d'_{u_i,v_i}(s,t)<0$ for $i=1,2,3$. Since
$U_t(s)=\{v_1,v_2,v_3\}$, according to Lemma \ref{lem:40}(1), $r_s$ is
in $R(s,t)$.
\end{enumerate}

The lemma thus follows.
\end{proof}


%

\subsection{Proof of Lemma \ref{lem:110}}
\label{sec:lem110proof}

Consider any direction $r_s$ for moving $s$ and $r_t$ for moving $t$.
Let $\theta_s$ denote the angle of the direction $r_s$. Let
$\theta_t$ denote the angle of $r_t$.
According to our analysis for Equation \eqref{equ:10}, we can obtain the derivatives of the three
functions $d'_{u_i,v_i}(s,t)$ for $i=1,2,3$, as follows.
\begin{equation}
\label{equ:60}
\begin{cases}
d'_{u_1,v_1}(s,t)= -\cos(\alpha_1-\theta_s) - \tau\cdot \cos(\beta_1-\theta_t),\\
d'_{u_2,v_2}(s,t)= -\cos(\alpha_2-\theta_s) - \tau\cdot \cos(\beta_2-\theta_t),\\
d'_{u_3,v_3}(s,t)= -\cos(\alpha_3-\theta_s) - \tau\cdot \cos(\beta_3-\theta_t).
\end{cases}
\end{equation}

Therefore, for each $i=1,2,3$, $d'_{u_i,v_i}(s,t)$ is a function of $\theta_s$,
$\theta_t$, and $\tau$.
In order to simplify our proof for Lemma \ref{lem:110}, we first give the following lemma.

\begin{lemma}\label{lem:130}
A direction $r_s$ is in $\piR(s,t)$ if and only if there exist
$\theta_t\in [0,2\pi)$ and $\tau\geq 0$ such that $d'_{u_1,v_1}(s,t)=0$ and
$d'_{u_i,v_i}(s,t)<0$ for $i=2,3$ (the same result holds if we switch
the index $1$ with $2$ or $3$).
\end{lemma}
\begin{proof}
In order to make the notation consistent with the rest of Section \ref{sec:lem110proof}, instead of proving the statement of the lemma, we prove the following statement
(which is essentially the same as the lemma): $r_s$ is in $\piR(s,t)$
if and only if there exist $\theta_t$ and $\tau\geq 0$ such that $d'_{u_2,v_2}(s,t)=0$ and $d'_{u_i,v_i}(s,t)<0$ for $i=1,3$.

Let $x=\alpha_1-\theta_s$ and $y=\beta_1-\theta_t$. Hence, we have the following:
\begin{equation*}
\label{eq:40}
\begin{cases}
d'_{u_1,v_1}(s,t)= -\cos(x) - \tau\cdot \cos(y),\\
d'_{u_2,v_2}(s,t)= -\cos(\alpha_2-\alpha_1+x) - \tau\cdot \cos(\beta_2-\beta_1+y),\\
d'_{u_3,v_3}(s,t)= -\cos(\alpha_3-\alpha_1+x) - \tau\cdot \cos(\beta_3-\beta_1+y).
\end{cases}
\end{equation*}

Consider any fixed direction $r_s$. The angle $\theta_s$ is also fixed, and thus $x$ is fixed. Hence, for each $i=1,2,3$, $d'_{u_i,v_i}(s,t)$ is a function of $\theta_t$ and $\tau$, and thus also a function of $y$ and $\tau$.
In the following proof, we will use $d_i'(y,\tau)$ to represent
$d'_{u_i,v_i}(s,t)$ for each $i=1,2,3$.

We first prove one direction of the lemma.
Suppose $r_s$ is in $\piR(s,t)$. Then there exist $y$ and $\tau\geq 0$ such
that $d_i'(y,\tau)<0$ for each $i=1,2,3$. Our goal is to prove that
there exit $y'$ and $\tau'\geq 0$ such that $d_2'(y',\tau')=0$ and
$d_i'(y',\tau')<0$ for $i=1,3$.

The idea is to change $y$ and $\tau$ simultaneously in such a way that
$d_2'(y,\tau)$ is constant, but $d_1'(y,\tau)$ strictly increases
while $d_3'(y,\tau)$ strictly decreases, and we keep making the change until $d_1'(y,\tau)$ becomes
zero. To this end, we will find a way that when we change  $\tau$ and
$y$ simultaneously, $\tau\cdot \cos(y)$ is constant, $-\tau\cdot
\cos(\beta_2-\beta_1+y)$ increases, and $-\tau\cdot
\cos(\beta_3-\beta_1+y)$ decreases, as follows.

Since $\cos(\beta_2-\beta_1+y)=\cos(\beta_2-\beta_1)\cdot \cos(y)-\sin(\beta_2-\beta_1)\cdot \sin(y)$, we have
$$d_2'(y,\tau) = -\cos(\alpha_2-\alpha_1+x) - \tau \cdot \cos(y)\cdot \cos(\beta_2-\beta_1) + \tau\cdot \sin(y)\cdot \sin(\beta_2-\beta_1).$$

Similarly, we derive
$$d_3'(y,\tau) = -\cos(\alpha_3-\alpha_1+x) - \tau \cdot \cos(y)\cdot \cos(\beta_3-\beta_1) + \tau\cdot \sin(y)\cdot \sin(\beta_3-\beta_1).$$

Consider a point $q$ in a polar coordinate system with coordinate $(\tau,y)$, i.e.,
$\tau=|oq|$ and the polar angle of $q$ is $y$, where $o$ is the origin
(e.g., see Fig.~\ref{fig:polar}). Let $l_q$ be the vertical line through $q$.
Note that if $\tau=0$, then $q=o$ and $l_q$ is the vertical line through $o$.

\begin{figure}[t]
\begin{minipage}[t]{\linewidth}
\begin{center}
\includegraphics[totalheight=1.2in]{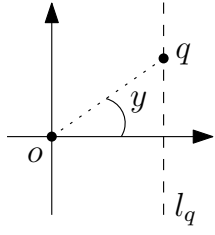}
\caption{\footnotesize Illustrating the polar coordinate of $q$ with $\tau=|oq|$.
}
\label{fig:polar}
\end{center}
\end{minipage}
\vspace*{-0.15in}
\end{figure}

For any point $p$ on $l_q$, let $\tau_p=|op|$ and let $y_p$ be the polar angle
of $p$. Suppose we move $p$  on $l_q$ from top to bottom, it
is easy to see that $\tau_p\cos(y_p)$ is constant and $\tau_p\sin(y_p)$ is
strictly decreasing. Based on this observation, we will find $y'$ and
$\tau'\geq 0$ such that $d_2'(y',\tau')=0$ and
$d_i'(y',\tau')<0$ for $i=1,3$, as follows.

Recall that
$\beta_1-\beta_2\equiv b_1$ and $b_1\in (0,\pi)$. Thus, $\sin(\beta_2-\beta_1)<0$.
Similarly, since $\beta_3-\beta_1\equiv b_3$ and $b_3\in (0,\pi)$,
$\sin(\beta_3-\beta_1)>0$. Hence, if we change the
values of $y$ and $\tau$ simultaneously such that $q$ moves along $l_q$
downwards, then $d_1'(y,\tau)$ does not change, but
$d_2'(y,\tau)$ strictly increases while
$d_3'(y,\tau)$ strictly decreases. We keep making the above change until at some moment
$d_2'(y,\tau)$ becomes zero, at which moment we have
$d_2'(y,\tau)=0$ and  $d_2'(y,\tau)<0$ for $i=1,3$.

The above proves one direction of the lemma. Next, we prove the other direction.

Suppose there exist $y$ and $\tau\geq 0$
such that $d_2'(y,\tau)=0$ and $d'_i(y,\tau)<0$ for $i=1,3$.
Our goal is to show that there exist $y'$ and $\tau'\geq 0$
such that $d'_i(y',\tau')<0$ for $i=1,2,3$.

\begin{enumerate}
\item
If $\tau>0$, depending on whether $\cos(\beta_2-\beta_1+y)=0$, there are two subcases.

\begin{enumerate}
\item
If $\cos(\beta_2-\beta_1+y)=0$, then it is always possible to change $y$ infinitesimally
such that $\cos(\beta_2-\beta_1+y)>0$.  Due to $\tau>0$,
$d_2'(y,\tau)<0$. Also, since the change of $y$ is
infinitesimal and both $\cos(y)$ and  $\cos(\beta_3-\beta_1+y)$ are
continuous functions, we still have $d'_i(y,\tau)<0$ for $i=2,3$.
We are done with the proof.
\item
If $\cos(\beta_2-\beta_1+y)\neq 0$, depending on whether $\cos(\beta_2-\beta_1+y)$
is positive or negative,
we can infinitesimally increase or decrease $\tau$, such that $\tau$ is still
positive and $d_i'(y,\tau)<0$ holds for each $i=1,2,3$.
We are done with the proof.
\end{enumerate}

\item

If $\tau=0$, then again depending on whether $\cos(\beta_2-\beta_1+y)$
is positive, negative, or zero, there are three subcases.

\begin{enumerate}
\item
If $\cos(\beta_2-\beta_1+y)> 0$,
then we can increase $\tau$ infinitesimally such that $\tau>0$
and $d_i'(y,\tau)<0$ for $i=1,2,3$.

\item
If $\cos(\beta_2-\beta_1+y)< 0$, then we first change $y$ to $y+\pi$
(i.e., the moving direction of $t$ is reversed). After this, since
$\tau=0$, we still have
$d'_2(y,\tau)=0$ and $d'_i(y,\tau)<0$ for $i=1,3$. However,
the difference is that now $\cos(\beta_2-\beta_1+y)>0$ for the new $y$.
Then, we can use the same analysis as the first subcase.

\item
If $\cos(\beta_2-\beta_1+y) = 0$, then we first slightly change $y$ such that
$\cos(\beta_2-\beta_1+y)>0$. Since $\tau=0$, we still have
$d'_2(y,\tau)=0$ and $d'_i(y,\tau)<0$ for $i=1,3$. However,
the difference is that now $\cos(\beta_2-\beta_1+y)>0$ for the new
$y$. Then, we can use the same analysis as the first
subcase.
\end{enumerate}
\end{enumerate}

This completes the proof of the lemma.
\end{proof}


To simplify the notation, let $w_i=-d'_{u_i,v_i}(s,t)$ for each $i=1,2,3$.
Let $x=\alpha_1-\theta_s$ and $y=\beta_1-\theta_t$. Then, by Equation \eqref{equ:60},
we have the following.
\begin{equation}
\label{equ:70}
\begin{cases}
w_1 = \cos(x) + \tau\cdot \cos(y),\\
w_2 = \cos(\alpha_2-\alpha_1+x) + \tau\cdot \cos(\beta_2-\beta_1+y),\\
w_3 = \cos(\alpha_3-\alpha_1+x) + \tau\cdot \cos(\beta_3-\beta_1+y).
\end{cases}
\end{equation}

Once $r_s$ is fixed, both $\theta_s$ and $x$ are fixed. Also,
given any $\theta_t$, we can determine $y$, and vice versa. In the following,
for each $i=1,2,3$, we consider $w_i$ implicitly as a function of
$y\in [0,2\pi)$ and $\tau\geq 0$. By Lemma \ref{lem:130},
$\piR(s,t)$ consists of all
directions $r_s$ for $s$ such that there exist
$y$ and $\tau\geq 0$ with $w_1=0$ and $w_i>0$ for $i=2,3$.

Let $w_1=0$. Then $\cos(x)+\tau\cdot\cos(y)=0$. Depending on whether $\cos(x)$
is positive, negative, or zero, there are three cases.

\subsubsection{The case $\cos(x)>0$.}

In this case, since $\tau\cdot \cos(y) = -\cos(x)$, we have
$\tau\cdot \cos(y)<0$ and $\tau\neq 0$. Since
$\tau\geq 0$, we obtain $\tau>0$ and $\cos(y)<0$, implying
that $y\in (\pi/2,3\pi/2)$. Further,
we have  $\tau=-\cos(x)/\cos(y)$. By replacing $\tau$
with $-\cos(x)/\cos(y)$ in Equation \eqref{equ:70} for $w_2$, we obtain
\begin{equation*}
\begin{split}
w_2 &= \cos(\alpha_2 - \alpha_1)\cdot \cos(x) - \sin(\alpha_2-\alpha_1)\cdot
\sin(x)\\
&\ \ \ \ - \frac{\cos(x)}{\cos(y)}\cdot \big[\cos(\beta_2-\beta_1)\cdot \cos(y)
-\sin(\beta_2-\beta_1)\cdot \sin(y) \big]\\
& = \cos(\alpha_2 - \alpha_1)\cdot \cos(x) - \sin(\alpha_2-\alpha_1)\cdot
\sin(x) \\
 &\ \ \ \ - \cos(\beta_2-\beta_1)\cdot \cos(x)   + \sin(\beta_2-\beta_1)\cdot  \tan(y)
 \cdot \cos(x).
\end{split}
\end{equation*}
Since $\cos(x)>0$, if we divide the right side of the above formula by
$\cos(x)$, we obtain that $w_2>0$ if and only if
$$\cos(\alpha_2 - \alpha_1) - \sin(\alpha_2-\alpha_1)\cdot
\tan(x) - \cos(\beta_2-\beta_1) + \sin(\beta_2-\beta_1)\cdot
\tan(y)>0.$$

Recall that $\beta_1-\beta_2\equiv b_1$ and $b_1\in (0,\pi)$,
and thus $\sin(\beta_2-\beta_1)<0$.
Hence, the above inequality is equivalent to
\begin{equation}
\label{equ:80}
\tan(y) < \frac{  \cos(\beta_2-\beta_1)
-\cos(\alpha_2 - \alpha_1) + \sin(\alpha_2-\alpha_1)\cdot
\tan(x)} {\sin(\beta_2-\beta_1)}.
\end{equation}

Therefore, we obtained that $w_2>0$ if and only if Inequality
\eqref{equ:80} holds.

Similarly, for $w_3$, we can obtain that $w_3>0$ if and only if the following
inequality holds:
\begin{equation}
\label{equ:90}
\tan(y) > \frac{  \cos(\beta_3-\beta_1)
-\cos(\alpha_3 - \alpha_1) + \sin(\alpha_3-\alpha_1)\cdot
\tan(x)} {\sin(\beta_3-\beta_1)}.
\end{equation}
Note that to obtain Inequality \eqref{equ:90},
we need to use the fact that $\sin(\beta_3-\beta_1)>0$, which is due
to $\beta_3-\beta_1 \equiv b_3$ and $b_3\in (0,\pi)$.

As a summary, the above shows that $w_1=0$ and $w_i>0$ for $i=2,3$ if and only if both
Inequalities \eqref{equ:80} and \eqref{equ:90} hold. Recall that since
$\cos(x)>0$, $y$ is in $(\pi/2,3\pi/2)$.
Further, note that we can find an angle $y\in (\pi/2,3\pi/2)$
such that both \eqref{equ:80} and \eqref{equ:90} hold if
and only if the right side of Inequality \eqref{equ:80} is strictly larger than the right
side of Inequality \eqref{equ:90}. Hence, we obtain the following observation.

\begin{observation}
For the case $\cos(x)>0$, there exist $y$ and $\tau\geq 0$ such that  $w_1=0$ and $w_i>0$
for $i=2,3$ if and only if the right side of Inequality \eqref{equ:80}
is strictly larger than the right
side of Inequality \eqref{equ:90}, which is equivalent to the following
inequality:
\begin{equation}
\label{equ:100}
\delta\cdot \tan(x) < \delta_1-\delta_2,
\end{equation}
where $\delta$, $\delta_1$, and $\delta_2$ are defined as in the
statement of Lemma \ref{lem:110}.
\end{observation}
\begin{proof}
Recall that
$$\delta=\frac{\sin(\alpha_3-\alpha_1)}{\sin(\beta_3-\beta_1)}-
\frac{\sin(\alpha_2-\alpha_1)}{\sin(\beta_2-\beta_1)},$$
$$\delta_1=\frac{\cos(\beta_2-\beta_1)-\cos(\alpha_2-\alpha_1)}{\sin(\beta_2-\beta_1)},
\text{\ and \ }
\delta_2=\frac{\cos(\beta_3-\beta_1)-\cos(\alpha_3-\alpha_1)}{\sin(\beta_3-\beta_1)}.$$

That the right side of Inequality \eqref{equ:80} is larger than the right
side of Inequality \eqref{equ:90}  is equivalent to the following
\begin{equation*}
\frac{\sin(\alpha_3-\alpha_1)\cdot \tan(x)}{\sin(\beta_3-\beta_1)}
- \frac{\sin(\alpha_2-\alpha_1)\cdot \tan(x)}{\sin(\beta_2-\beta_1)}<
\frac{  \cos(\beta_2-\beta_1) -\cos(\alpha_2 - \alpha_1) } {\sin(\beta_2-\beta_1)}
-\frac{  \cos(\beta_3-\beta_1) -\cos(\alpha_3 - \alpha_1) }
{\sin(\beta_3-\beta_1)},
\end{equation*}
which is further equivalent to Inequality \eqref{equ:100}.
\end{proof}

Recall that $\cos(x)>0$. Thus, $x\in (-\pi/2,\pi/2)$.

If $\delta >0$, Inequality \eqref{equ:100} is equivalent to
$\tan(x)<\frac{\delta_1-\delta_2}{\delta}$. Consequently, since $x\in
(-\pi/2,\pi/2)$, we obtain that $x\in
(-\pi/2,\arctan(\frac{\delta_1-\delta_2}{\delta}))$.

If $\delta <0$, Inequality \eqref{equ:100} is equivalent to
$\tan(x)>\frac{\delta_1-\delta_2}{\delta}$, implying that $x\in
(\arctan(\frac{\delta_1-\delta_2}{\delta}),\pi/2)$.

If $\delta=0$, Inequality \eqref{equ:100} is equivalent to
$0 < \delta_1-\delta_2$. Hence, if $\delta_1>\delta_2$,
then any $x$ can make the inequality hold, implying that $x\in
(-\pi/2,\pi/2)$; otherwise, no $x$ can make it hold.

In summary, for the case $\cos(x)>0$,
there exist $y$ and $\tau\geq 0$ such that  $w_1=0$ and $w_i>0$
for $i=2,3$ if and only if: $x\in
(-\pi/2,\arctan(\frac{\delta_1-\delta_2}{\delta}))$ when $\delta>0$;
$x\in (\arctan(\frac{\delta_1-\delta_2}{\delta}),\pi/2)$ when $\delta<0$;
$x\in (-\pi/2,\pi/2)$ when $\delta = 0$ and $\delta_1>\delta_2$.

\subsubsection{The case $\cos(x)<0$.}
In this case, $x\in (\pi/2,3\pi/2)$. The analysis is similar to the
above case. We omit the details and only give the result below.

There exist $y$ and $\tau\geq 0$ such that  $w_1=0$ and $w_i>0$
for $i=2,3$ if and only if: $x\in
(\arctan(\frac{\delta_1-\delta_2}{\delta})+\pi,3\pi/2)$ when $\delta>0$;
$x\in (\pi/2,\arctan(\frac{\delta_1-\delta_2}{\delta})+\pi)$ when $\delta<0$;
$x\in (\pi/2,3\pi/2)$ when $\delta = 0$ and $\delta_1<\delta_2$.

\subsubsection{The case $\cos(x)=0$.}
In this case $x=\pm \pi/2$. The analysis for this case is different
from the above two cases. We will show the following result:
There exist $y$ and $\tau\geq 0$ such that  $w_1=0$ and $w_i>0$
for $i=2,3$ if and only if: $x = -\pi/2$ when $\delta>0$;
$x = \pi/2$ when $\delta<0$ (there is no value for $x$ when $\delta=0$).

By Lemma \ref{lem:130}, there exist
$y$ and $\tau\geq 0$ such that  $w_1=0$ and $w_i>0$
for $i=2,3$ if and only there exist $y$ and $\tau\geq 0$
such that  $w_i>0$ for $i=1,2,3$.
For convenience, in this case, we will find all directions $r_s$ such that
there exist $y$ and $\tau\geq 0$ with $w_i>0$ for $i=1,2,3$.

Since $\cos(x)=0$, we have $w_1=\tau\cdot \cos(y)$. Since we require $w_1>0$,
$\tau$ cannot be $0$ and thus
$\tau>0$ must hold. Therefore, we have $\cos(y)>0$  and $y\in (-\pi/2,\pi/2)$.
In addition, $\tau = \frac{w_1}{\cos(y)}$. By replacing $\tau$ with $\frac{w_1}{\cos(y)}$
and setting $\cos(x)$ to $0$ in Equation \eqref{equ:70} for $w_2$, we obtain the following
\begin{equation*}
\begin{split}
w_2 &= \cos(\alpha_2 - \alpha_1)\cdot \cos(x) - \sin(\alpha_2-\alpha_1)\cdot
\sin(x)\\
&\ \ \ \ + \frac{w_1}{\cos(y)}\cdot \big[\cos(\beta_2-\beta_1)\cdot \cos(y)
- \sin(\beta_2-\beta_1)\cdot \sin(y) \big]\\
& =  - \sin(\alpha_2-\alpha_1)\cdot
\sin(x) + \cos(\beta_2-\beta_1)\cdot w_1   - \sin(\beta_2-\beta_1)\cdot  \tan(y)
 \cdot w_1.
\end{split}
\end{equation*}

Hence, $w_2>0$ if and only if the following holds
\begin{equation*}
 \sin(\beta_2-\beta_1)\cdot  \tan(y)  \cdot w_1 < - \sin(\alpha_2-\alpha_1)\cdot
\sin(x) + \cos(\beta_2-\beta_1)\cdot w_1.
\end{equation*}

Recall that since
$\beta_1-\beta_2=b_1$ and $b_1\in (0,\pi)$, $\sin(\beta_2-\beta_1)<0$.
Since $w_1>0$, the above inequality is equivalent to the following
\begin{equation}\label{equ:110}
 \tan(y) > \frac{- \sin(\alpha_2-\alpha_1)\cdot \sin(x) +
 \cos(\beta_2-\beta_1)\cdot w_1}{\sin(\beta_2-\beta_1)\cdot w_1}
\end{equation}

For $w_3$, by the similar analysis (we also need to use the fact that
$\sin(\beta_3-\beta_1)>0$), $w_3>0$ if and only if the following holds
\begin{equation}\label{equ:120}
 \tan(y) < \frac{- \sin(\alpha_3-\alpha_1)\cdot \sin(x) +
 \cos(\beta_3-\beta_1)\cdot w_1}{\sin(\beta_3-\beta_1)\cdot w_1}
\end{equation}

Based on the above discussion, $w_i>0$ for $i=1,2,3$ if and only if we can find
$y\in (-\pi/2,\pi/2)$ such that both Inequalities \eqref{equ:110} and
\eqref{equ:120} hold. It is not difficult to see that we can find $y\in
(-\pi/2,\pi/2)$ such that both Inequalities \eqref{equ:110} and \eqref{equ:120}
hold if and only if the right side of Inequality \eqref{equ:110} is
strictly less than
the right side of Inequality \eqref{equ:120}, i.e.,
\begin{equation}\label{equ:130}
\frac{- \sin(\alpha_2-\alpha_1)\cdot \sin(x) +
\cos(\beta_2-\beta_1)\cdot w_1}{\sin(\beta_2-\beta_1)\cdot w_1} <
\frac{- \sin(\alpha_3-\alpha_1)\cdot \sin(x) +
\cos(\beta_3-\beta_1)\cdot w_1}{\sin(\beta_3-\beta_1)\cdot w_1}.
\end{equation}

Since $w_1>0$, we have the following observation.

\begin{observation}\label{obser:50}
$w_i>0$ for $i=1,2,3$ if and only if there exits $w_1>0$ such that
\begin{equation*}
\delta \cdot \sin(x) < w_1\cdot \Big[\cot(\beta_3-\beta_1)-\cot(\beta_2-\beta_1)\Big].
\end{equation*}
\end{observation}
\begin{proof}
The left side of Inequality \eqref{equ:130} is equal to
$\cot(\beta_2-\beta_1)-\frac{\sin(\alpha_2-\alpha_1)}{\sin(\beta_2-\beta_1)}\cdot
\frac{\sin(x)}{w_1}$ and the right side is equal to
$\cot(\beta_3-\beta_1)-\frac{\sin(\alpha_3-\alpha_1)}{\sin(\beta_3-\beta_1)}\cdot
\frac{\sin(x)}{w_1}$. Hence,
Inequality \eqref{equ:130} is equivalent to the following
\begin{equation*}
\frac{\sin(x)}{w_1}\cdot
\Big[\frac{\sin(\alpha_3-\alpha_1)}{\sin(\beta_3-\beta_1)}-\frac{\sin(\alpha_2-\alpha_1)}{\sin(\beta_2-\beta_1)}\Big]
< \cot(\beta_3-\beta_1) -\cot(\beta_2-\beta_1).
\end{equation*}

Recall that
$\delta=\frac{\sin(\alpha_3-\alpha_1)}{\sin(\beta_3-\beta_1)}-\frac{\sin(\alpha_2-\alpha_1)}{\sin(\beta_2-\beta_1)}$.
Since $w_1>0$, the observation is obtained.
\end{proof}

We also have the following observation.
\begin{lemma}\label{lem:140}
It holds that $\cot(\beta_3-\beta_1)-\cot(\beta_2-\beta_1) < 0$.
\end{lemma}
\begin{proof}
Note that $\beta_3-\beta_1\equiv b_3$ and $\beta_2-\beta_1\equiv b_3+b_2$.
Hence, to prove the lemma, it is
sufficient to prove $\cot(b_3)<\cot(b_3+b_2)$.

Recall that $b_i$ is in $(0,\pi)$ for each $i=1,2,3$.
Since $b_1+b_2+b_3=2\pi$, it holds that $b_3+b_2\in (\pi,2\pi)$.

Note that $\cot(b_3) = \cot(b_3 + \pi )$. Due to $b_3\in (0,\pi)$,
$b_3+\pi\in (\pi,2\pi)$.  Hence, both $b_3+b_2$ and $b_3 +\pi$ are in
$(\pi,2\pi)$. Since $b_2\in (0,\pi)$, $b_3+b_2<b_3+\pi$.
Because $\cot(\cdot)$ is a strictly decreasing function on
$(\pi,2\pi)$, we obtain $\cot(b_3+b_2) > \cot(b_3 + \pi) =
\cot(b_3)$, which proves the lemma.
\end{proof}

Recall that $x$ is either $\pi/2$  or $-\pi/2$.

In light of Lemma \ref{lem:140},
if $\delta>0$, then there exists $w_1>0$ such that the inequality in
Observation \ref{obser:50} is satisfied if and only if $\sin(x) < 0$, i.e.,
$x=-\pi/2$. Therefore, if $\delta>0$, by Observation \ref{obser:50}, $w_i>0$ for
$i=1,2,3$ if and only if $x=-\pi/2$.

Similarly, we can obtain that if $\delta<0$, then $w_i>0$ for
$i=1,2,3$ if and only if $x=\pi/2$.

If $\delta=0$, then for any $x$ and any $w_1>0$, the inequality in Observation
\ref{obser:50} can never be satisfied. Therefore, for any $x$, it is not possible
to have $w_i<0$ for all $i=1,2,3$.

\subsubsection{A summary of all three cases.}

Let $R(x)$ be the set of values for $x$ such that there exist $y$ and
$\tau\geq 0$ with $w_i>0$ for $i=1,2,3$. In other words, since $x=\alpha_1-\theta_t$, $R(s)=\{\alpha_1-\theta_t \ |\  \theta_t\in \piR(s,t)\}$. By our discussions for all
three cases $\cos(x)>0$, $\cos(x)<0$, and $\cos(x)=0$, we can obtain
the following:

If $\delta>0$, we have
\begin{equation*} 
R(x)= \begin{cases}
(-\pi/2, \arctan(\frac{\delta_1-\delta_2}{\delta})) & \text{for $\cos(x)>0$},\\
(\arctan(\frac{\delta_1-\delta_2}{\delta})+\pi, 3\pi/2)
=(\arctan(\frac{\delta_1-\delta_2}{\delta})-\pi, -\pi/2) &\text{for $\cos(x)<
0$},\\
\{-\pi/2\} & \text{for $\cos(x)=0$}.
\end{cases}
\end{equation*}


Therefore, if $\delta>0$, $R(x)$ is the union of the above three intervals,
which is $(\arctan(\frac{\delta_1-\delta_2}{\delta})-\pi,
\arctan(\frac{\delta_1-\delta_2}{\delta}))$.

Similarly, if $\delta<0$, we have
\begin{equation*} 
R(x) = \begin{cases}
(\arctan(\frac{\delta_1-\delta_2}{\delta}),\pi/2) & \text{for $\cos(x)>0$},\\
(\pi/2, \arctan(\frac{\delta_1-\delta_2}{\delta})+\pi) &\text{for $\cos(x)< 0$},\\
\{\pi/2\} & \text{for $\cos(x)=0$}.
\end{cases}
\end{equation*}

Therefore, if $\delta<0$, $R(x) = (\arctan(\frac{\delta_1-\delta_2}{\delta}),
\arctan(\frac{\delta_1-\delta_2}{\delta})+\pi)$.

Finally, if $\delta = 0$, depending on the values of $\delta_1$ and
$\delta_2$, there are three cases. If $\delta_1>\delta_2$,
$R(x)=(-\pi/2,\pi/2)$. If $\delta_1<\delta_2$,
$R(x)=(\pi/2,3\pi/2)$. If $\delta_1=\delta_2$, $R(x)=\emptyset$.

As a summary, we obtain the following
\begin{equation*}
R(x) = \begin{cases}
(\arctan(\frac{\delta_1-\delta_2}{\delta})-\pi,
\arctan(\frac{\delta_1-\delta_2}{\delta})) & \text{if $\delta>0$},\\
(\arctan(\frac{\delta_1-\delta_2}{\delta}),
\arctan(\frac{\delta_1-\delta_2}{\delta})+\pi) & \text{if $\delta<0$},\\
(-\pi/2,\pi/2) & \text{if $\delta=0$ and $\delta_1>\delta_2$}, \\
(\pi/2,3\pi/2) & \text{if $\delta=0$ and $\delta_1<\delta_2$}, \\
\emptyset & \text{if $\delta=0$ and $\delta_1=\delta_2$}, \\
\end{cases}
\end{equation*}

Since $R(s)=\{\alpha_1-\theta_t \ |\  \theta_t\in \piR(s,t)\}$, we obtain the range $\piR(s,t)$ for $\theta_s$ as follows.
\begin{equation*}
 \piR(s,t)= \begin{cases}
(\alpha_1-\arctan(\frac{\delta_1-\delta_2}{\delta}),\alpha_1
-\arctan(\frac{\delta_1-\delta_2}{\delta})+\pi)
& \text{if $\delta>0$},\\
(\alpha_1-\arctan(\frac{\delta_1-\delta_2}{\delta}) - \pi,\alpha_1
-\arctan(\frac{\delta_1-\delta_2}{\delta})) & \text{if $\delta<0$},\\
(\alpha_1-\pi/2,\alpha_1+\pi/2) & \text{if $\delta=0$ and $\delta_1>\delta_2$}, \\
(\alpha_1-3\pi/2,\alpha_1-\pi/2) & \text{if $\delta=0$ and $\delta_1<\delta_2$}, \\
\emptyset & \text{if $\delta=0$ and $\delta_1=\delta_2$}, \\
\end{cases}
\end{equation*}

To prove Lemma \ref{lem:110}, it remains to show that  $\delta=0$ and
$\delta_1=\delta_2$ if and only if $a_i=b_i$ for each $i=1,2,3$. To this end,
we first give the following lemma.

\begin{lemma}\label{lem:150}

If $\delta=0$ and $\delta_1=\delta_2$, then $\delta_1=\delta_2=0$.
\end{lemma}
\begin{proof}
Recall that
$$\delta=\frac{\sin(\alpha_3-\alpha_1)}{\sin(\beta_3-\beta_1)}-
\frac{\sin(\alpha_2-\alpha_1)}{\sin(\beta_2-\beta_1)},$$
$$\delta_1=\frac{\cos(\beta_2-\beta_1)-\cos(\alpha_2-\alpha_1)}{\sin(\beta_2-\beta_1)},
\text{\ and \ }
\delta_2=\frac{\cos(\beta_3-\beta_1)-\cos(\alpha_3-\alpha_1)}{\sin(\beta_3-\beta_1)}.$$

For the purpose of differentiation, we use $\delta(1,2,3)$,
$\delta_1(1,2,3)$, and $\delta_2(1,2,3)$ to represent the above $\delta$,
$\delta_1$, and $\delta_2$, respectively. We use $(1,2,3)$ because our previous
discussion considered $d'_{u_i,v_i}(s,t)$ in the order of $1,2,3$.

Suppose $\delta(1,2,3)=0$ and $\delta_1(1,2,3)=\delta_2(1,2,3)$. In the
following, we give an ``indirect'' approach to prove
$\delta_1(1,2,3)=\delta_2(1,2,3)=0$.

Our previous
discussion has proved that $\piR(s,t)=\emptyset$ if and only $\delta(1,2,3)=0$ and
$\delta_1(1,2,3)=\delta_2(1,2,3)$. In other words, there is no value
for $\theta_s$ such that there exist $\theta_t$ and $\tau\geq 0$ with $w_i>0$
for $i=1,2,3$ if and only if $\delta(1,2,3)=0$ and
$\delta_1(1,2,3)=\delta_2(1,2,3)$.

Our previous analysis considered the three functions $w_i$ in the order of
$1,2,3$. If we consider them in the order of $2,3,1$, by using the same
analysis, we can obtain that $\piR(s,t)=\emptyset$ if and only if  $\delta(2,3,1)=0$ and $\delta_1(2,3,1)=\delta_2(2,3,1)$, where
$$\delta(2,3,1)=\frac{\sin(\alpha_1-\alpha_2)}{\sin(\beta_1-\beta_2)}-
\frac{\sin(\alpha_3-\alpha_2)}{\sin(\beta_3-\beta_2)},$$
$$\delta_1(2,3,1)=\frac{\cos(\beta_3-\beta_2)-\cos(\alpha_3-\alpha_2)}{\sin(\beta_3-\beta_2)},
\text{\ and \ }
\delta_2(2,3,1)=\frac{\cos(\beta_1-\beta_2)-\cos(\alpha_1-\alpha_2)}{\sin(\beta_1-\beta_2)}.$$

Another way to think about this is that we can replace $1,2,3$ by $2,3,1$,
respectively, in $\delta(1,2,3)$,  $\delta_1(1,2,3)$, and $\delta_2(1,2,3)$.
The main reason why the same analysis as before still works is that if we follow
the order of $2,3,1$, the vertices of
$u_2,u_3,u_1$ are still in the counterclockwise order around $s$ and the vertices of
$v_2,v_3,v_1$ are still in the clockwise order around $t$.
We omit the detailed analysis.

Similarly, if we consider the three functions $w_i$ in the order of $3,1,2$,
by using the same analysis,
we can obtain that $\piR(s,t)=\emptyset$ if and only if  $\delta(3,1,2)=0$ and
$\delta_2(3,1,2)=\delta_3(3,1,2)$, where
$$\delta(3,1,2)=\frac{\sin(\alpha_2-\alpha_3)}{\sin(\beta_2-\beta_3)}-
\frac{\sin(\alpha_1-\alpha_3)}{\sin(\beta_1-\beta_3)},$$
$$\delta_1(3,1,2)=\frac{\cos(\beta_1-\beta_3)-\cos(\alpha_1-\alpha_3)}{\sin(\beta_1-\beta_3)},
\text{\ and \ }
\delta_2(3,1,2)=\frac{\cos(\beta_2-\beta_3)-\cos(\alpha_2-\alpha_3)}{\sin(\beta_2-\beta_3)}.$$

Recall that $\delta(1,2,3)=0$ and $\delta_1(1,2,3)=\delta_2(1,2,3)$. Thus,
$\piR(s,t)=\emptyset$. According to the above discussion, $\delta(2,3,1)=0$ and
$\delta_1(2,3,1)=\delta_2(2,3,1)$, and $\delta(3,1,2)=0$ and
$\delta_2(3,1,2)=\delta_3(3,1,2)$.
In the following, we show that
$\delta_1(1,2,3)=\delta_2(1,2,3)=\delta_1(2,3,1)=\delta_2(2,3,1)=\delta_2(3,1,2)
=\delta_3(3,1,2)=0$, which will prove the lemma.

First of all,  $\delta_2(1,2,3)=-\delta_1(3,1,2)$ holds
because $\cos(\beta_3-\beta_1)=\cos(\beta_1-\beta_3)$,
$\cos(\alpha_3-\alpha_1)=\cos(\alpha_1-\alpha_3)$, and
$\sin(\beta_3-\beta_1)=-\sin(\beta_1-\beta_3)$.

Similarly, we can obtain $\delta_1(2,3,1)=-\delta_2(3,1,2)$ and $\delta_1(1,2,3)=-\delta_2(2,3,1)$.

Due to $\delta_2(1,2,3)=-\delta_1(3,1,2)$ and
$\delta_1(2,3,1)=-\delta_2(3,1,2)$,
we derive
$\delta_1(1,2,3)=\delta_2(1,2,3)=-\delta_1(3,1,2)
=-\delta_2(3,1,2)=\delta_1(2,3,1)=\delta_2(2,3,1)$.
In particular, we have
$\delta_1(1,2,3)=\delta_2(2,3,1)$.


The above has obtained that both $\delta_1(1,2,3)=-\delta_2(2,3,1)$ and
$\delta_1(1,2,3)=\delta_2(2,3,1)$ hold, implying that
$\delta_1(1,2,3)=\delta_2(2,3,1)=0$. Consequently,
$\delta_1(1,2,3)=\delta_2(1,2,3)=\delta_1(2,3,1)=\delta_2(2,3,1)=\delta_2(3,1,2)
=\delta_3(3,1,2)=0$.

The lemma thus follows.
\end{proof}

\begin{lemma}\label{lem:160}
$\delta=0$ and $\delta_1=\delta_2$ if and only if $a_i=b_i$ for each $i=1,2,3$.
\end{lemma}
\begin{proof}
Recall that $a_1\equiv \alpha_2-\alpha_1$, $a_2\equiv \alpha_3-\alpha_2$,
$a_3\equiv\alpha_1-\alpha_3$, $b_1\equiv\beta_1-\beta_2$, $b_2\equiv\beta_2-\beta_3$,
and $b_3\equiv\beta_3-\beta_1$.
We first prove one direction of the lemma. Suppose $a_i=b_i$ for each $i=1,2,3$. Then,
\begin{equation*}
\begin{split}
\delta=\frac{\sin(\alpha_3-\alpha_1)}{\sin(\beta_3-\beta_1)}-
\frac{\sin(\alpha_2-\alpha_1)}{\sin(\beta_2-\beta_1)} =
\frac{-\sin(a_3)}{\sin(b_3)}- \frac{\sin(a_1)}{-\sin(b_1)} = 0.
\end{split}
\end{equation*}

Also, we have $\cos(\beta_2-\beta_1)=\cos(b_1)$ and
$\cos(\alpha_2-\alpha_1)=\cos(a_1)$. Since $a_1=b_1$, we obtain $\delta_1=0$.
Similarly, we have $\cos(\beta_3-\beta_1)=\cos(b_3)$ and
$\cos(\alpha_3-\alpha_1)=\cos(a_3)$. Since $a_3=b_3$, we obtain $\delta_2=0$.
This proves that $\delta=0$ and $\delta_1=\delta_2$.

Next we prove the other direction of the lemma. Suppose $\delta=0$ and
$\delta_1=\delta_2$. In the following we prove $a_i=b_i$ for $i=1,2,3$.
By Lemma \ref{lem:150}, we know that $\delta_1=\delta_2=0$.

Since $\delta_1=0$, it holds that
$\cos(\beta_2-\beta_1)=\cos(\alpha_2-\alpha_1)$, i.e., $\cos(b_1)=\cos(a_1)$.

Similarly, since $\delta_2=0$, it holds that
$\cos(\beta_3-\beta_1)=\cos(\alpha_3-\alpha_1)$, i.e., $\cos(b_3)=\cos(a_3)$.

On the other hand, $\delta=0$ implies that $\frac{\sin(a_3)}{\sin(b_3)}=
\frac{\sin(a_1)}{\sin(b_1)}$.

Note that $b_1+b_2+b_3=2\pi$. Recall that two or even all three of
$u_1,u_2,u_3$ may be the same. According to our definitions of the
three angles $a_1,a_2,a_3$, if $u_1=u_2=u_3$, then $a_1=a_2=a_3=0$;
otherwise, $a_1+a_2+a_3=2\pi$ although one of the three angles may
still be zero. In either case, it holds that $a_1+a_2+a_3\leq 2\pi$.

In the sequel, we first prove $a_1=b_1$ and $a_3=b_3$.

Note that $b_i\in (0,\pi)$ for each $i=1,2,3$, and $a_i\in [0,2\pi)$
for each $i=1,2,3$.

Assume to the contrary that $b_1\neq a_1$ and $b_3\neq a_3$.
Due to $\cos(b_1)=\cos(a_1)$, it must be that $a_1=2\pi - b_1$. Since
$b_1\in (0,\pi)$, we obtain $a_1\in (\pi,2\pi)$.
Similarly,
due to $\cos(b_3)=\cos(a_3)$, it must be that $a_3=2\pi - b_3$. Since
$b_3\in (0,\pi)$, we obtain $a_3\in (\pi,2\pi)$. This yields that
$a_1+a_3>2\pi$, contradicting with $a_1+a_2+a_3\leq 2\pi$.

Therefore, at least one of $b_1\neq a_1$ and $b_3\neq a_3$ must be true.
Without loss of generality, we assume $b_1=a_1$.

Since $\frac{\sin(a_3)}{\sin(b_3)}= \frac{\sin(a_1)}{\sin(b_1)}$ and
$b_1=a_1$, it is obviously true that $\sin(a_3)=\sin(b_3)$. Recall
that we already have $\cos(a_3)=\cos(b_3)$. Therefore, it must hold that
$a_3=b_3\mod 2\pi$. Since $a_i\in [0,2\pi)$ and $b_3\in (0,\pi)$, we
obtain that $a_3=b_3$.

This proves that $a_1=b_1$ and $a_3=b_3$. Since neither $b_1$ nor
$b_3$ is zero, we obtain $a_1\neq 0$ and $a_3\neq 0$, and thus
$a_1+a_2+a_3=2\pi$. This implies that $a_2=b_2$ since
$b_1+b_2+b_3=2\pi$. The lemma thus follows.
\end{proof}

\section{Computing the Candidate Points}
\label{sec:candidates}

In this section, with the help of the observations in Sections
\ref{sec:obser} and \ref{sec:range}, we compute a set $S$ of
candidate points such that all geodesic centers must be in $S$.

Let $s$ be any geodesic center. Recall that $F(s)$ is the set of all farthest points
of $s$. Depending on whether $s$ is in $\calV$, $E$,
or $\calI$, the size $|F(s)|$, whether
some points of $F(s)$ are in $\calV$, $E$, or $\calI$, whether $s$ has a degenerate farthest point,
there are a significant (but still constant) number of cases.
For each case, our algorithm uses an exhaustive-search approach to compute a set of
candidate points such that $s$ must be in the set.
In particular, there are four cases, called {\em dominating cases}, for which the
number of candidate points is $O(n^{11})$. But
the total number of the candidate points for all other
cases is only $O(n^{10})$.
Therefore, the set $S$ has a total of $O(n^{11})$ candidate points.
We will show that $S$ can be computed in
$O(n^{11}\log n)$ time.

To find the geodesic centers in $S$, a
straightforward algorithm works as follows. For each point $\hats\in S$, we can compute
$d_{\max}(\hats)$ in $O(n\log n)$ time by first computing the shortest path
map $\spm(\hats)$ of $\hats$ in $O(n\log n)$ time \cite{ref:HershbergerAn99} and
then obtaining the maximum geodesic distance from $\hats$ to all vertices of
$\spm(\hats)$.
Since all geodesic centers are in $S$, the points of $S$ with the smallest
$d_{\max}(\hats)$ are geodesic centers of $\calP$.

Since $|S|=O(n^{11})$, the above algorithm runs in $O(n^{12}\log n)$ time.
Let $S_d$ denote the set of the candidate points for the four dominating cases.
Clearly, the bottleneck is on finding the geodesic centers from $S_d$.
To improve the algorithm, when we compute the candidate points of
$S_d$, we will maintain the corresponding {\em path information}.
By using these path information and based on new observations,
we will present in Section \ref{sec:center}
an $O(n^{11}\log n)$ time ``pruning algorithm'' that can eliminate most of the points from $S_d$
such that none of the eliminated points is a geodesic center and the number of
remaining points in $S_d$ is only $O(n^{10})$. Consequently, we can use the above
straightforward algorithm to find all geodesic centers in $O(n^{11}\log n)$ time.

In the rest of this section, we focus on computing the set $S$.
Our algorithm for finding the geodesic centers from $S$ (in
particular, the pruning algorithm)
will be given in Section \ref{sec:center}.

In the following, we adopt the follow convention on notation: $s$
represents a true geodesic center and $\hats$ represents a corresponding candidate
point.

First of all, for the case $s\in \calV$, we consider all polygon vertices of $\calV$ as
candidate points. In the following, we only consider the case $s\in E$ or $s\in
\calI$. If $s$ has a degenerate farthest point, we refer to it as the {\em
degenerate case}; otherwise it is a {\em non-degenerate or general case}.
We will compute candidate points for the general case and the
degenerate case in Sections \ref{sec:general} and \ref{sec:degenerate},
respectively. The four dominating cases are all general cases. We begin with
computing the candidate points for a special case in Section \ref{sec:special}.

\subsection{A Special Case}
\label{sec:special}

Consider a non-degenerate farthest point $t$ of $s$ such that $t$ is
in $E$ or $\calI$. If $t$ is in $E$, then there are exactly two
shortest \st\ paths; we say that $t$ is a {\em special} farthest
point of $s$ if the $\pi$-range $\piR(s,t)$ of $s$ with respect to $t$ and the two
shortest paths is $\emptyset$ (i.e., the special case in Lemma \ref{lem:80}).
Similarly, if $t$ is in $\calI$, then there are
exactly three shortest \st\ paths; we say that $t$ is a {\em
special} farthest point of $s$ if the
$\pi$-range $\piR(s,t)$ of $s$ with respect to $t$ and the three
shortest paths is $\emptyset$ (i.e., the special case in Lemma
\ref{lem:110}).

If $s$ has a special farthest point, we refer to it as the {\em special case}.
Let $s$ be any geodesic center in the special case.
In this subsection, we will compute a set $S_1$ of $O(n^5)$
candidate points in $O(n^5\log n)$ time for the special
case, such that $s$ must be in $S_1$.

Let $t$ be a special farthest point of $s$. By the above definition, $t$ is
in $E$ or $\calI$. We discuss the two cases below.

\subsubsection{The case $t\in E$.}

Let $e$ be the polygon edge that contains $t$. Let $\pi_1(s,t)$ and $\pi_2(s,t)$
be the two shortest \st\ paths. For each $i=1,2$, let $u_i$ and $v_i$ be the
$s$-pivot and $t$-pivot of $\pi_i(s,t)$, respectively (i.e.,
$\pi_i(s,t)=\pi_{u_i,v_i}(s,t)$). Then, we have
$|su_1|+d(u_1,v_1)+|v_1t|=|su_2|+d(u_2,v_2)+|v_2t|$.
We define the angles
$\alpha_1,\alpha_2,\beta_1,\beta_2$ in the same way as those for Lemma \ref{lem:80} in Section
\ref{sec:rangeE}. Since $\piR(s,t)=\emptyset$, by Lemma \ref{lem:80},
 $\beta_1+\beta_2=\pi$ and $\alpha_2-\alpha_1=\pm\pi$.

Note that $\alpha_2-\alpha_1=\pm\pi$ implies that $s$ is on the line segment
$\overline{u_1u_2}$. $\beta_1+\beta_2=\pi$ implies that $l_t$ bisects the angle
$\angle v_1tv_2$, where $l_t$ is the line through $t$ and perpendicular to $e$.

In summary, $s$ and $t$ satisfy the following ``constraints'': (1) $t\in e$; (2)
$|su_1|+d(u_1,v_1)+|v_1t|=|su_2|+d(u_2,v_2)+|v_2t|$; (3) $s\in
\overline{u_1u_2}$; (4) $l_t$ bisects the angle $\angle v_1tv_2$.

If we consider the coordinates of $s$ and $t$ as four variables, the above
four (independent) constraints can determine $s$ and
$t$. An easy observation is that $s$ must be in a bisector edge of
two cells of $\spm(t)$ whose roots are $u_1$ and $u_2$, respectively (in fact
$s$ is the intersection of the bisector edge and $\overline{u_1u_2}$).
Correspondingly, we compute the candidate points in an exhaustive way as
follows.

We enumerate all possible combinations of a polygon edge as $e$ and two
polygon vertices as $v_1$ and $v_2$. For each combination, we find a
point $t$ on $e$ such that the above condition (4) can be
satisfied. 
Next we compute the shortest path map $\spm(t)$ of $t$.
For each bisector edge of $\spm(t)$, we consider it as $e$ and let $u_1$
and $u_2$ be the roots of the two cells of $\spm(t)$ incident to $e$;
we report the intersection $e\cap \overline{u_1u_2}$ (if any) as a
candidate point $\hats$.
In this way, we can compute at most $O(n)$ candidate points on
$\spm(t)$ since the combinatorial complexity of $\spm(t)$ is $O(n)$. Thus,
for each combination of $e$, $v_1$, and $v_2$, we can compute $O(n)$ candidate
points.  Since there are $O(n^3)$ combinations, we can compute
$O(n^4)$ candidate points and we add these points to $S_1$.
By our above discussions, the geodesic center $s$ must be in $S_1$.

Since computing a shortest path map takes $O(n\log n)$ time
\cite{ref:HershbergerAn99}, the
running time of the above algorithm is bounded by $O(n^4\log n)$.

\subsubsection{The case $t\in \calI$.}
In this case, there are exactly three shortest \st\ paths:
$\pi_i(s,t)=\pi_{u_i,v_i}(s,t)$ with $i=1,2,3$. Then we have
$|su_1|+d(u_1,v_1)+|v_1t|=|su_2|+d(u_2,v_2)+|v_2t|=|su_3|+d(u_3,v_3)+|v_3t|$.
We define the angles $a_i$, $b_i$, for $i=1,2,3$, in the same way as those for Lemma \ref{lem:110}
in Section \ref{sec:three}.  By Lemma \ref{lem:110}, $a_i=b_i$ for $i=1,2,3$.

If we consider the coordinates of $s$ and $t$ as four
variables, the above equation on the lengths of the three shortest
paths provide two constraints and the identities of the three pairs of angles
provide other two (independent) constraints, and thus the total four
(independent) constraints
can determine $s$ and $t$. Correspondingly, we compute
the candidate points for $s$ as follows.

We enumerate all possible combinations of three polygon vertices as
$v_1,v_2,v_3$. We compute the shortest path maps of $v_1$, $v_2$, and
$v_3$ in $O(n\log n)$ time.
Next we compute the overlay of the three shortest path maps. The overlay
is of size $O(n^2)$ and can compute in $O(n^2\log n)$ time \cite{ref:BaeTh13,ref:ChiangTw99}.
Then, for each cell of the overlay, we obtain the three
roots of the cell in the three shortest path maps and consider them as
$u_1,u_2,u_2$. Finally, we use
the above four constraints to determine a constant number of
pairs $(\hats,t)$ (we assume this can be done in constant time
since the angles $a_i$ and $b_i$ can be parameterized by the
coordinates of $\hats$ and $t$),
and we add each $\hats$ to $S_1$ as a candidate point.
In this way, for each combination of $v_1,v_2,v_3$, we can compute $O(n^2)$ candidate
points in $O(n^2\log n)$ time.
Since there are $O(n^3)$ combinations, we can compute $O(n^5)$
candidate points in $O(n^5\log n)$ time.

\subsection{The General Case}
\label{sec:general}

Recall that in the general case $s$ does not have a degenerate farthest point.
Depending on whether $s$ is in $\calI$ or in $E$, there are two main cases. We first
consider the case $s\in \calI$.

\subsubsection{The case $s\in \calI$.}
Depending on whether $|F(s)|$ is $1$, $2$, or at least $3$, there are further three cases.

\paragraph{The case $|F(s)|=1$.}
If $|F(s)|=1$, let $t$ be the only farthest point of $s$. Thus, $R(s)=R(s,t)$. By Lemma \ref{lem:10},
$R(s)=\emptyset$. Hence, $R(s,t)=\emptyset$. Since $s\in \calI$, by our
observations in Section \ref{sec:range} (in particular, Lemmas \ref{lem:60},
\ref{lem:70}, \ref{lem:80}, \ref{lem:90}, \ref{lem:110}, \ref{lem:120}), $t$ is
either in $E$ or $\calI$, and in either case $\piR(s,t)=\emptyset$ since $s\in \calI$ and $R_f(s)$ consists of all directions. Hence, $t$
must be a special farthest point of $s$, which implies that $s$ has
been computed  in the candidate
point set $S_1$.

\paragraph{The case $|F(s)|=2$.}
If $|F(s)|=2$, we assume that $F(s)$ does not have any special farthest points since
otherwise $s$ would have already been computed in $S_1$. Let
$F(s)=\{t_1,t_2\}$. Depending on whether each of $t_1$ and $t_2$ is in $\calV$, $E$, or
$\calI$, there are several cases. In the following, we use $(x,y,z)$ to
refer to the case where $x$, $y$, $z$ points of $F(s)$ are in $\calI$, $E$, and
$\calV$, respectively, with $x+y+z=2$. For example, $(1,1,0)$ refers
to the
case where one point of $F(x)$ is in $\calI$ and the other is in $E$.

\begin{enumerate}
\item
\noindent {\em Case $(2,0,0)$.}
We first consider the most general case where $t_i\in \calI$ for $i=1,2$. Other
cases are very similar. For each $i=1,2$, since $t_i\in \calI$,
$t_i$ has three shortest paths to $s$:
$\pi_{u_{ij},v_{ij}}(s,t_i)$ with $j=1,2,3$. Hence, we have the
following
\begin{equation*}
\begin{split}
& |{t_1v_{11}}|+d(v_{11},u_{11})+|{u_{11}s}|
=|{t_1v_{12}}|+d(v_{12},u_{12})+|{u_{12}s}|
= |{t_1v_{13}}|+d(v_{13},u_{13})+|{u_{13}s}|
\\
= & |{t_2v_{21}}|+d(v_{21},u_{21})+|{u_{21}s}|
=|{t_2v_{22}}|+d(v_{22},u_{22})+|{u_{22}s}|
=|{t_2v_{23}}|+d(v_{23},u_{23})+|{u_{23}s}|.
\end{split}
\end{equation*}

Further, since $s$ has only two farthest points and $s\in \calI$, it must hold
that $R(s)=\piR(s,t_1)\cap \piR(s,t_2)=\emptyset$ by Lemma
\ref{lem:10}. Recall that we have assumed that $F(s)$
does not have any special farthest points. Therefore, neither $\piR(s,t_1)$ nor
$\piR(s,t_2)$ is $\emptyset$. By Lemma \ref{lem:110}, each of $\piR(s,t_1)$ and
$\piR(s,t_2)$ is an open range of size $\pi$. For each $i=1,2$, the directions of
$\piR(s,t_i)$ are delimited by an open half-plane
whose bounding line contains $s$, and we
refer to the bounding line as the {\em bounding line} of the range
$\piR(s,t_i)$. Since $\piR(s,t_1)\cap \piR(s,t_2)=\emptyset$, we can obtain that
the bounding line of $\piR(s,t_1)$ is the same as that of $\piR(s,t_2)$.

If we consider the coordinates of $s$, $t_1$, and $t_2$ as six variables, the
above system of equations on the distances of the six shortest paths provide
five constraints
and that the bounding lines of the two $\pi$-ranges overlap provides another
constraint (note that the bounding lines of
the two $\pi$-range can be parameterized by the coordinates of $s$, $t_1$, and $t_2$).
Hence, the six  (independent) constraints can determine the triple
$(s,t_1,t_2)$. Correspondingly, our algorithm for computing the candidate
points works as follows.

We enumerate all possible combinations of six polygon vertices as
$v_{ij}$ for $i=1,2$ and $j=1,2,3$.
We compute overlay of the shortest path maps of the six vertices.
For each cell of the overlay, we obtain the six roots
in the six shortest path maps as
$u_{ij}$ for $i=1,2$ and $j=1,2,3$.
Using the six constraints as discussed above, we can obtain a constant number of
triples $(\hats,t_1,t_2)$, and each such $\hats$ is a candidate
point.

In this way, for each combination, we can compute $O(n^2)$ candidate points
in $O(n^2\log n)$ time. Since there are $O(n^6)$ combinations,
we can compute $O(n^8)$ candidate points in overall $O(n^8\log n)$ time.
As analyzed above, $s$ must be one of these candidate points.

\item
\noindent {\em Case $(1,1,0)$.}
In this case, one of $t_1$ and $t_2$ is in $\calI$ and the other is in $E$. Without loss of generality, we assume $t_1\in \calI$ and $t_2\in E$. Hence, $t_1$  has three shortest paths to $s$: $\pi_{u_{1j},v_{1j}}(s,t_1)$ for $j=1,2,3$, and $t_2$ has two shortest paths to $s$: $\pi_{u_{2j},v_{2j}}(s,t_2)$ for $j=1,2$. Hence, we have the following
\begin{equation*}
\begin{split}
& |{t_1v_{11}}|+d(v_{11},u_{11})+|{u_{11}s}|
=|{t_1v_{12}}|+d(v_{12},u_{12})+|{u_{12}s}|
= |{t_1v_{13}}|+d(v_{13},u_{13})+|{u_{13}s}|
\\
= & |{t_2v_{21}}|+d(v_{21},u_{21})+|{u_{21}s}|
=|{t_2v_{22}}|+d(v_{22},u_{22})+|{u_{22}s}|.
\end{split}
\end{equation*}

As in the previous case, $\piR(s,t_1)\cap \piR(s,t_2)=\emptyset$ and
each of the two $\pi$-ranges is nonempty; further,
and the two bounding lines of these two $\pi$-ranges must overlap.

Let $e$ be the polygon edge that contains $t$.

The above equations provide four constraints and the overlap of the two bounding
lines of the two $\pi$-ranges provides another constraint. In
addition, $t_2\in e$ gives the sixth constraint. As in the previous
case, the six
constraints can determine $s$, $t_1$, and $t_2$ if we consider their coordinates
as six variables. Correspondingly, we compute candidate points
as follows.

We enumerate all possible combinations of a polygon edge as $e$ and
five polygon vertices as $v_{11}, v_{12}, v_{13}, v_{21}, v_{22}$. We
compute the overlay of  the shortest path maps of the five vertices.
For each cell of the overlay, we
obtain the five roots in the five shortest path maps as $u_{11},
u_{12}, u_{13}, u_{21}, u_{22}$. Then we obtain the equations as above.
Along with the constraint that $t_2\in e$ and the constraint that the two
bounding lines of the $\pi$-ranges $\piR(\hats,t_1)$ and $\piR(\hats,t_2)$
overlap, we can obtain a constant number of triples
$(\hats,t_1,t_2)$, and each such $\hats$ is a candidate
point. In this way, for each combination, we can
compute $O(n^2)$ candidate points in $O(n^2\log n)$ time.
Since there are $O(n^8)$ combinations, we can compute $O(n^8)$
candidate points for $s$ in $O(n^8\log n)$ time.

\item
\noindent {\em Other Cases.}
The other cases are all very similar (i.e., $(1,0,1)$, $(0,2,0)$, $(0,1,1)$, $(0,0,2)$), and we can compute a total of
$O(n^8)$ candidate points in $O(n^8\log n)$ time. We omit the details.
\end{enumerate}

As a summary, for the case $|F(s)|=2$, we can compute at most $O(n^8)$ candidate
centers in $O(n^8\log n)$ time.

\paragraph{The case $|F(s)|\geq 3$.}
In this case, the geodesic center $s$ has at least three farthest points.
We assume $s$ does not have any special farthest points since
otherwise $s$ would have already been computed in $S_1$.
Hence, for each $t\in
F(s)$, the $\pi$-range $\piR(s,t)$ is not empty and thus is an open
range of size $\pi$. Together with Lemma \ref{lem:10},
this further implies that $s$ must have three farthest points $t_1$,
$t_2$, $t_3$ such that $\piR(s,t_1)\cap \piR(s,t_2)\cap \piR(s,t_3)=\emptyset$.
Depending on whether each of $t_i$ for $i=1,2,3$ is in $\calI$, $E$,
$\calV$, there are several cases.
Similarly, we use $(x,y,z)$ to refer to the case where $x$, $y$, and $z$ points of
$t_1,t_2,t_3$ are in $\calI,E$, and $\calV$, respectively, with $x+y+z=3$.

We begin our discussion with the most general case $(3,0,0)$, where
$t_1,t_2,t_3$ are all in $\calI$. This is one of the four dominating cases and
we will need to compute $O(n^{11})$ candidate points.
Further, our algorithm is slightly different from the previous
cases in the following sense: First, in addition to the candidate points, we
will also maintain the corresponding {\em path information}; second, when
computing the candidate points, we will have a ``validation procedure''.
Recall that $S_d$ is the set of candidate points for all four dominating cases.
The path information will be used in the next section to quickly prune most of
the points in $S_d$ and the validation procedure can be considered as a
``preliminary pruning step''. The details are given below.

For each $i=1,2,3$, since $t_i$ is in $\calI$, there are exactly three shortest paths
from $s$ to $t_i$:
$\pi_{u_{ij},v_{ij}}(s,t)$ with $j=1,2,3$. Hence, we have the following
\begin{equation*}
\begin{split}
& |{t_1v_{11}}|+d(v_{11},u_{11})+|{u_{11}s}|
=|{t_1v_{12}}|+d(v_{12},u_{12})+|{u_{12}s}|
= |{t_1v_{13}}|+d(v_{13},u_{13})+|{u_{13}s}|
\\
= & |{t_2v_{21}}|+d(v_{21},u_{21})+|{u_{21}s}|
=|{t_2v_{22}}|+d(v_{22},u_{22})+|{u_{22}s}|
=|{t_2v_{23}}|+d(v_{23},u_{23})+|{u_{23}s}|\\
= &|{t_3v_{31}}|+d(v_{31},u_{31})+|{u_{31}s}|
=|{t_3v_{32}}|+d(v_{32},u_{32})+|{u_{32}s}|=
|{t_3v_{33}}|+d(v_{33},u_{33})+|{u_{33}s}|.\\
\end{split}
\end{equation*}

If we consider the coordinates of $s, t_1, t_2$, and $t_3$ as eight variables,
the above equations on the lengths of nine paths provide eight
constraints, which are sufficient to determine all four points.
Correspondingly, we compute the candidate points as follows.

We enumerate all possible combinations of nine polygon vertices as
$v_{i1}, v_{i2},v_{i3}$, with $i=1,2,3$.
For each combination, we compute the overlay of the shortest
path maps of these nine vertices. The size of the overlay is
$O(n^2)$. For each cell $C$ of the overlay, we obtain nine roots of the
shortest path maps and consider them as
$u_{i1},u_{i2},u_{i3}$ for $i=1,2,3$.
We form the above system of eight equations and
solve it to obtain a constant number of quadruples of points
$(\hats,\hatt_1,\hatt_2,\hatt_3)$.
We also obtain a path length, denoted by $d(\hats)$, which is equal
to the value in the above equations, e.g.,
$d(\hats)=|\hats u_{11}|+d(u_{11},v_{11})+|v_{11}\hatt_1|$.
In addition, we perform a validation procedure on each such quadruple
$(\hats,\hatt_1,\hatt_2,\hatt_3)$ as follows.

First, we check whether $\hats$ is in $C$, which can be done in
$O(\log n)$ time by using a point location data structure
\cite{ref:EdelsbrunnerOp86,ref:KirkpatrickOp83} with $O(n^2)$ time and
space preprocessing on the overlay.
If yes, for each $t_i$ with $i=1,2,3$, we check whether
$d(\hats)$ is equal to $d(\hats,\hatt_i)$, which can be
computed in $O(\log n)$ time by using the two-point shortest path query data
structure given by Chiang and Mitchell \cite{ref:ChiangTw99}
with $O(n^{11})$ time and space preprocessing on $\calP$.
If yes, we check whether $v_{i1},v_{i2},v_{i3}$ satisfy the condition in Observation \ref{obser:10}(1), i.e.,
whether $\hatt_i$ is in the interior of the triangle $\triangle
v_{i1}v_{i2}v_{i3}$ for each $i=1,2,3$.
If yes, for each $\hatt_i$ with $i=1,2,3$, we check whether the order
of the vertices of $v_{i1},v_{i2},v_{i3}$ around $\hatt_i$ are consistent
with the order of the vertices of  $u_{i1},u_{i2},u_{i3}$ (we say that the two
orders are {\em consistent} if after reordering the indices,
$v_{i1},v_{i2},v_{i3}$ are clockwise around $\hatt_i$ while
$u_{i1},u_{i2},u_{i3}$ are counterclockwise around $s$; note that this consistency is needed for determining the $\pi$-range in
Lemma \ref{lem:110}).
If yes, for each $\hatt_i$ with $i=1,2,3$,
we compute the $\pi$-range $\piR(s,\hatt_i)$ determined by
Lemma \ref{lem:110}, and then check whether $\piR(\hats,\hatt_1)\cap
\piR(\hats,\hatt_2)\cap \piR(\hats,\hatt_3)$ is empty.
If yes, we say that the quadruple $(\hats,\hatt_1,\hatt_2,\hatt_3)$ {\em passes}
the validation procedure and we call $\hats$ a {\em valid} candidate point and
add $\hats$ to the set $S_d$.
In addition, we maintain the following {\em path information}:
$d(\hats)$, $\hatt_i$, $v_{ij}$, and $u_{ij}$, with $1\leq i\leq 3$
and $1\leq j\leq 3$. In fact, only $d(\hats)$ will be used later in
the algorithm and all other information are only for the reference
purpose in the analysis.

In this way, for each combination of nine polygon vertices, we can
compute $O(n^2)$ valid candidate for $S_d$ in $O(n^2\log n)$
time (not including the preprocessing time). Since there are $O(n^9)$
combinations, we can compute a total of $O(n^{11})$ valid candidate
points for $S_d$ in $O(n^{11}\log n)$ time.

Note that the geodesic center $s$ and the
quadruple $(s,t_1,t_2,t_3)$ discussed above must pass the validation procedure,
and thus the quadruple $(s,t_1,t_2,t_3)$ will be computed by our
exhaustive-search algorithm and
$s$ will be computed as a valid candidate point in $S_d$.

Based on our validation procedure, the following observation
summarizes the
properties of the valid candidate points. These properties will be
used to prove the correctness of our pruning algorithm in Section
\ref{sec:center}.

\begin{observation}\label{obser:60}
Suppose $(\hats,\hatt_1,\hatt_2,\hatt_3)$ is a quadruple that passes the
validation procedure, with $u_{ij}$ and $v_{ij}$, $i=1,2,3$ and
$j=1,2,3$ defined as above. Then the following hold.
\begin{enumerate}
\item
For each $i=1,2,3$, $\overline{\hats u_{ij}}\cup \pi(u_{ij},v_{ij})\cup
\overline{v_{ij}\hatt_i}$ is a shortest path from $\hats$ to $\hatt_i$ for
each $j=1,2,3$.
\item
For each $i=1,2,3$, 
$v_{i1},v_{i2},v_{i3}$ satisfy the condition of Observation
\ref{obser:10}(1), i.e.,
$\hatt_i$ is in the interior of the triangle $\triangle v_{i1}v_{i2}v_{i3}$.
\item
$d(\hats)=d(\hats,\hatt_i)$ for each $i=1,2,3$.
\item
$\piR(\hats,\hatt_1) \cap \piR(\hats,\hatt_2) \cap
\piR(\hats,\hatt_1) = \emptyset$.
\end{enumerate}
\end{observation}



The above computes the candidate points for the case $(3,0,0)$. In
the following, we compute candidate points for other cases. The
algorithms are similar.

\begin{enumerate}
\item
\noindent{\em Case $(2,1,0)$.} In this case, two of $t_1,t_2,t_3$ are
in $\calI$ and the third one is in $E$. Without loss of generality, we
assume $t_1$ and $t_2$ are in $\calI$ and $t_3$ is in $E$.
Let $e$ be the polygon edge containing $t_3$.
This is the second dominating case. We will also perform a
validation procedure and maintain the corresponding path information.

For each $i=1,2$, since $t_i$ is in $\calI$, there are three shortest
paths from $s$ to $t_i$: $\pi_{u_{ij},v_{ij}}(s,t)$ for $j=1,2,3$. Since $t_3\in E$,
there are two shortest paths from $s$ to $t_3$:
$\pi_{u_{3j},v_{3j}}(s,t)$ for $j=1,2$.
Hence, we have the following
\begin{equation*}
\begin{split}
& |{t_1v_{11}}|+d(v_{11},u_{11})+|{u_{11}s}|
=|{t_1v_{12}}|+d(v_{12},u_{12})+|{u_{12}s}|
= |{t_1v_{13}}|+d(v_{13},u_{13})+|{u_{13}s}|
\\
= & |{t_2v_{21}}|+d(v_{21},u_{21})+|{u_{21}s}|
=|{t_2v_{22}}|+d(v_{22},u_{22})+|{u_{22}s}|
=|{t_2v_{23}}|+d(v_{23},u_{23})+|{u_{23}s}|\\
= &|{t_3v_{31}}|+d(v_{31},u_{31})+|{u_{31}s}|
=|{t_3v_{32}}|+d(v_{32},u_{32})+|{u_{32}s}|.\\
\end{split}
\end{equation*}

If we consider the coordinates of $s,t_1,t_2,t_3$ as eight variables,
the above equations give seven constraints. With the additional
constraint that $t_3$ is on $e$, we can determine $(s,t_1,t_2,t_3)$.
Correspondingly, we compute the candidate points as follows.


We enumerate all combinations of eight polygon vertices (as
$v_{11},v_{12},v_{13},v_{21},v_{22},v_{23},v_{31},v_{32}$) and one polygon
edge (as $e$). For each combination, we can compute $O(n^2)$
quadruples $(\hats,\hatt_1,\hatt_2,\hatt_3)$.
For each such quadruple, we also perform the validation procedure
similarly as before with the following differences.
For $\hatt_3$, since there is no vertex $v_{33}$, we ignore all operations that involve $v_{33}$. Further, we check whether
$\{v_{11},v_{12}\}$ satisfy the condition in Observation
\ref{obser:10}(2) (instead of Observation \ref{obser:10}(1)).
We determine the range $\piR(s,t_3)$ by using Lemma \ref{lem:80} (instead of Lemma \ref{lem:110}).

In this way, we can compute at most $O(n^{11})$ valid
candidate points in $O(n^{11}\log n)$ time and add
them to $S_d$, and also, we maintain the corresponding path
information similarly as before. Observation \ref{obser:60} still holds
correspondingly with the following difference: in (1) there are only
two shortest paths for $\hatt_3$; in (2) for $\hatt_3$, $\{v_{11},v_{12}\}$ satisfy the condition of Observation \ref{obser:10}(2).

\item
\noindent{\em Cases $(1,2,0)$ and $(0,3,0)$.}
These are the other two dominating cases. We can compute $O(n^{11})$
valid candidate points in $O(n^{11}\log n)$ time.
The algorithms are very similar to the previous two cases.
For each case, we also need a validation procedure and keep the path information.
We omit the details.


We have discussed all four dominating cases.
Note that the candidate points for all four dominating cases are in
$\calI$.
The rest of the cases are not
dominating cases. For each of the remaining cases, the validation
procedure and the path information will not be needed any more.


\item
\noindent{\em Other cases.}
Our algorithms for all other cases (e.g. $(2,0,1)$, $(1,1,1)$, $(1,0,2)$, $(0,2,1)$, $(0,1,2)$, $(0,0,3)$) are
very similar as before (except that the validation
procedure and the path information are not needed). We can compute a total of $O(n^{10})$ candidate points in
$O(n^{10}\log n)$ time.
We omit the details.
\end{enumerate}

As a summary for the case $s\in \calI$, for the four dominating cases, we have
computed $O(n^{11})$ valid candidate points in $O(n^{11}\log n)$ time
with the corresponding path information.
For all other cases, we have computed $O(n^{10})$ candidate
points.

\subsubsection{The case $s\in E$.}
Our algorithm for computing candidate points for this case is slightly different.
Consider a geodesic center $s$ in $E$. Let $e_s$ denote the polygon edge that
contains $s$. Depending on whether $|F(s)|$ is $1$ or not, there are two cases.

\paragraph{The case $|F(s)|=1$.}
If $|F(s)|=1$, let $t$ be the only farthest point of $s$. We assume that $t$ is not a special farthest point of $s$ since otherwise $s$ would have already been computed in $S_1$ in Section \ref{sec:special}. Thus, by Lemma \ref{lem:80}, the $\pi$-range $\piR(s,t)$ is an open range of size $\pi$.
Since $t$ is the only farthest point of $s$, by Lemma \ref{lem:10},  $R(s)=R(s,t)=\emptyset$. By Lemma \ref{lem:90}, $R(s,t)=\emptyset$ implies that the bounding line of the range $\piR(s,t)$ contains $e_s$, and this provides a constraint for determining $s$ and $t$.  Depending on whether $t$ is in $\calI$, $E$, or
$\calV$, there are three cases.

\begin{enumerate}
\item
\noindent{\em The case $t\in \calI$.} We first consider the most general case in
which $t\in \calI$. There are three shortest \st\ paths: $\pi_{u_iv_i}(s,t)$ for $i=1,2,3$.
Thus, we have the following
\begin{equation*}
\begin{split}
& |{tv_{1}}|+d(v_{1},u_{1})+|{u_{1}s}|
=|{tv_{2}}|+d(v_{2},u_{2})+|{u_{2}s}|
= |{tv_{3}}|+d(v_{3},u_{3})+|{u_{3}s}|.
\end{split}
\end{equation*}

If we consider the coordinates of $s$ and $t$ as four variables, then
the above equations provide two constraints and $s\in e_s$ gives another
constraint. In addition, that the bounding line of the range
$\piR(s,t)$ contains $e_s$ provides the fourth constraint. Hence, the four constraints can determine $s$ and $t$. Correspondingly, we compute the
candidate points as follows.

We enumerate all possible combinations of three polygon vertices as
$v_1$, $v_2$, $v_3$. For each
combination, we compute the overlay of the shortest path maps of $v_1,v_2$, and
$v_3$. For each cell $C$ of the overlay,
if it contains a portion of a polygon edge, we consider the edge as
$e_s$ and obtain three roots of the shortest path maps as $u_1,u_2$, and
$u_3$. Using the four constraints discussed above,
we can obtain a constant number of tuples $(\hats,t)$ and each
such $\hats$ is a candidate point.
In this way, we can compute $O(n^5)$ candidate points in $O(n^5\log n)$ time.

\item
\noindent{\em The case $t\in E$.}
Let $e_t$ be the polygon edge that contains $t$.
There are two shortest \st\ paths, which gives one
constraint for determining $s$ and $t$.
$s\in e_s$ and $t\in e_t$ provide two
constraints. Finally, that the bounding line of the range
$\piR(s,t)$ contains $e_s$ gives the fourth constraint.
Hence, we still have four constraints to determine $s$ and $t$.

The algorithm is similar to the previous case. We can
compute $O(n^5)$ candidate points in $O(n^5\log n)$ time. We omit the details.

\item
\noindent{\em The case $t\in \calV$.}
In this case, $t$ is a polygon vertex, which provides two constraints (or $t$ is
fixed already). $s\in e_s$ is another constraint.
As before, that the bounding line of the range
$\piR(s,t)$ contains $e_s$ provides the fourth constraint.
We compute the candidate points as follows.

We enumerate all polygon vertices. For each vertex, we consider it as
$t$ and compute the shortest path map $\spm(t)$. For each cell $C$ of $\spm(t)$, if it contains a portion of polygon edge on its boundary, then we consider that edge as $e_s$. Next, we find a candidate point $\hats$ on that portion of $e_s$ in $C$ such that the $\pi$-range $\piR(\hats,t)$ (determined by Lemma \ref{lem:80}) has its bounding line contain $e_s$.
In this way, we can compute $O(n^2)$ candidate points in $O(n^2\log n)$ time.
\end{enumerate}

In summary, for the case $|F(s)|=1$, we compute $O(n^5)$
candidate points in $O(n^5\log n)$ time.

\paragraph{The case $|F(s)|\geq 2$.}
Let $t_1$ and $t_2$ be the two farthest points of $s$. We assume $s$ does not have any special farthest points since otherwise $s$ would have already been computed in $S_1$. Then, neither $\piR(s,t_1)$ nor $\piR(s,t_2)$ is empty. Depending on whether each of $t_1$ and $t_2$ is in $\calI$, $E$, or $\calV$, there are several cases.


We first consider the most general case where both $t_1$ and $t_2$ are in
$\calI$. Other cases are very similar.
For each $i=1,2$, since $t_i\in \calI$, there are three
shortest paths from $s$ to $t_i$:
$\pi_{u_{ij},v_{ij}}(s,t)$ with $j=1,2,3$.  Hence, we have
the following
\begin{equation*}
\begin{split}
& |{t_1v_{11}}|+d(v_{11},u_{11})+|{u_{11}s}|
=|{t_1v_{12}}|+d(v_{12},u_{12})+|{u_{12}s}|
= |{t_1v_{13}}|+d(v_{13},u_{13})+|{u_{13}s}|
\\
= & |{t_2v_{21}}|+d(v_{21},u_{21})+|{u_{21}s}|
=|{t_2v_{22}}|+d(v_{22},u_{22})+|{u_{22}s}|
=|{t_2v_{23}}|+d(v_{23},u_{23})+|{u_{23}s}|.\\
\end{split}
\end{equation*}

The above equations provide five constraints for $s$ and $t$. With the
constraint that $s\in e_s$, we can determine $s$, $t_1$, and $t_2$, if we consider their coordinates as six variables. Correspondingly, we compute the candidate points as follows.

We enumerate all possible combinations of six polygon vertices as
$v_{11},v_{12},v_{13},v_{21},v_{22},v_{23}$. For each combination, we compute the overlay of the shortest path
maps of the six vertices. For each cell $C$ of the overlay, if $C$ contains a
portion of a polygon edge, then we consider the edge as $e_s$ and using the
six constraints mentioned above to compute at most a constant number of triples
$(\hats,t_1,t_2)$. Hence, for each combination, we can compute at most $O(n^2)$
candidate points in $O(n^2\log n)$ time.
In this way, we can compute $O(n^8)$ candidate points in $O(n^8\log n)$ time.

Candidate points for other cases can be computed similarly.
We can compute $O(n^8)$ candidate
points in $O(n^8\log n)$ time. We omit the details.

As a summary for the case $s\in E$, we can compute a total of $O(n^{8})$
candidate points in $O(n^8\log n)$ time.

\subsection{The Degenerate Case}
\label{sec:degenerate}
In the degenerate case, the geodesic center $s$ has at least one degenerate
farthest point.
Depending on whether $s$ is in $\calI$ or $E$, there are two main cases.

\subsubsection{The Case $s\in \calI$.}
Depending on whether $|F(s)|$ is $1$ or not, there are further two
cases.

\noindent{\em The case $|F(s)|=1$.}
Let $t$ be the only point of $F(s)$. Then $t$ is a degenerate farthest point of
$s$. Depending on whether $t$ is in $\calV$, $E$, or $\calI$, there are three
cases.

\begin{enumerate}
\item
If $t\in \calV$, then since $t$ is degenerate, there are at least two shortest \st\ paths.
Further, since $t$ is polygon vertex and due to our general position assumption that no two
polygon vertices have more than one shortest path, it holds that $|U_s(t)|\geq 2$.

If $|U_s(t)|\geq 3$, then $s$ is a vertex of $\spm(t)$. Correspondingly, we can
compute candidate points as follows. For each polygon vertex, we consider it as
$t$ and compute its
shortest path map $\spm(t)$. Then, we report each vertex of $\spm(t)$ as a
candidate points. In this way, we can compute $O(n^2)$ candidate points
for $s$ in $O(n^2\log n)$ time.

If $|U_s(t)|=2$, let $U_s(t)=\{u_1,u_2\}$.
We first prove a claim: $s\in \overline{u_1u_2}$.

Indeed, suppose to the contrary that the claim is not true. Then,
$\overline{su_1}$ and $\overline{su_2}$ form an angle $\angle{u_1su_2}\in
(0,\pi)$.  Consider
the direction $r_s$ for moving $s$ along the bisector of $\angle v_1sv_2$ and
towards the interior of $\angle v_1sv_2$.
If we move $s$ along $r_s$ with unit speed, since $\angle{u_1su_2}$ is strictly less than $\pi$ and $U_s(t)=\{u_1,u_2\}$, each vertex $v\in
U_t(s)$ has a coupled $s$-pivot $u\in U_s(t)$ with $d'_{u,v}(s,t)<0$.
Since $s\in \calI$, $r_s$ is a free direction.
By Lemma \ref{lem:40}, $r_s$ is an admissible direction of $s$ with respect to $t$, i.e., $r_s\in R(s,t)$. Since $R(s)=R(s,t)$, we obtain that $R(s)$ is not empty, contradicting with Lemma \ref{lem:10}.
The claim thus follows.

$s\in \overline{u_1u_2}$ provides a constraint for
determining $s$ (if we consider $t\in \calV$ as fixed at a polygon vertex).
Further, there are two shortest \st\ paths $\pi_{u_iv_i}(s,t)$ with $i=1,2$.
Hence, $s$
must be in the bisector edge that is incident to the two cells whose roots are $u_1$ and $u_2$, respectively, in the shortest path map $\spm(t)$.
Correspondingly, we compute the candidate points as follows.

We enumerate all polygon vertices. For each polygon vertex, we consider it as
$t$ and compute $\spm(t)$. For each bisector edge of
$\spm(t)$, we obtain the roots of the two cells incident to the bisector edge as $u_1$ and $u_2$. Then, we compute a candidate
point $\hats$ on the bisector edge such that $\hats\in \overline{u_1u_2}$ (i.e., $\hats$ is the intersection of $\overline{u_1u_2}$ and the bisector edge).
In this way, we can compute $O(n^2)$ candidate points in $O(n^2\log n)$ time.

\item
If $t\in E$, then since $t$ is degenerate,
there are at least three shortest \st\ paths. Let
$e_t$ be the polygon edge that contains $t$.
Since $t$ is a farthest point of $s$, $t$ must have two vertices $v_1$
and $v_2$ that satisfy the condition in Observation \ref{obser:10}.
Hence, there are two shortest \st\ paths $\pi_{u_i,v_i}(s,t)$ for
$i=1,2$.

Since $s\in \calI$, we claim that the $\pi$-range $\piR(s,t)$ with respect to the
two shortest paths $\pi_{u_i,v_i}(s,t)$ for $i=1,2$ must be empty. Assume to the
contrary that this is not true. Then,
 by Lemma \ref{lem:50}(2), $R(s,t)$ is not empty (in fact,
 $\piR(s,t)\subseteq R(s,t)$). Since $t$ is the only
farthest point of $s$, we obtain that $R(s)=R(s,t)\neq \emptyset$, which
contradicts with Lemma \ref{lem:10}.

Since $\piR(s,t)=\emptyset$, one can check that $s$ has already been computed in
$S_1$ in Section \ref{sec:special} (for the case $t\in E$).

\item
If $t\in \calI$, since $t$ is degenerate, there are at most four shortest \st\ paths.
By Observation \ref{obser:10}, $|U_t(s)|\geq 3$.  Depending on whether
$|U_t(s)|=3$, there are two cases.

We first discuss the case where $|U_t(s)|=3$.
Let $v_1,v_2,v_3$ be the three vertices in $U_t(s)$. By Observation
\ref{obser:10}, $t$ is in the interior of the
triangle $\triangle v_1v_2v_3$. Then, there must exist three shortest
\st\ paths $\pi_{u_i,v_i}(s,t)$ such that no two paths cross each
other, 
and this implies that the three paths are canonical.
Let $\piR(s,t)$ denote the $\pi$-range of $s$ with respect to
$t$ and the above three shortest paths.

If $\piR(s,t)=\emptyset$, then one can check that $s$ has already been computed in
$S_1$ in Section \ref{sec:special}.  Otherwise, as the above second case,
since $s\in \calI$ and $F(s)=\{t\}$, by Lemma \ref{lem:50}(1)
any direction in $\piR(s,t)$ is in $R(s,t)$, implying that
$R(s)=R(s,t)\neq \emptyset$, which contradicts with
Lemma \ref{lem:10}.

Next we discuss the case where $|U_t(s)|>3$.

By Observation \ref{obser:10}, $t$ is in the interior of the convex hull of all
vertices of $U_t(s)$. If there exist three vertices $v_1,v_2,v_3$ of
$U_t(s)$ such that $t$ is in the interior of $\triangle v_1v_2v_3$, then we can
still use the same approach as the above case for $|U_t(s)|=3$. Otherwise,
there must exist four vertices $v_1,v_2,v_3,v_4\in U_t(s)$ such that $t$ is
the intersection of the two line segments $\overline{v_1v_2}$ and $\overline{v_3v_4}$.

On the other hand, since $s\in \calI$ and $t$ is the only farthest point of $s$,
we claim that $|U_s(t)|\geq 2$. Indeed, suppose to the contrary that
$|U_s(t)|=1$. Then, if we move $s$ towards the only vertex $u$ of $U_s(t)$ with unit
speed, then for each vertex $v\in U_t(s)$, $d'_{u,v}(s,t)<0$ holds.
Hence, by Lemma \ref{lem:20}(1), the above direction for $s$ is in $R(s,t)$,
implying that $R(s)=R(s,t)\neq \emptyset$, which contradicts with
Lemma \ref{lem:10}.
Therefore, the claim follows.

Depending on whether $|U_s(t)|\geq 3$, there are two cases.

\begin{enumerate}
\item
If $|U_s(t)|\geq 3$, then $s$ is a vertex of the shortest path map $\spm(t)$.
Correspondingly, we compute the candidate points as follows.

We enumerate all possible combinations of four polygon
vertices as $v_i$ for $1\leq i\leq 4$.
For each such combination, we compute the intersection of $\overline{v_1v_2}$
and $\overline{v_3v_4}$ and consider the intersection as $t$. Then
we compute $\spm(t)$ and return all vertices of $\spm(t)$ as candidate
points.
In this way, for each combination, we can compute at most $O(n)$ candidate
points in $O(n\log n)$ time. Since there are $O(n^4)$ combinations,
we can compute a total of $O(n^5)$
candidate points in $O(n^5\log n)$ time.

\item
If $|U_s(t)|= 2$, let $U_s(t)=\{u_1,u_2\}$. An easy observation is
that $s$ must be on the bisector edge of $u_1$ and $u_2$ in
$\spm(t)$. Further, by the same analysis as before,
we can show that $s$ must be on $\overline{u_1u_2}$.
Correspondingly, we can compute the candidate points as follows.

We enumerate all possible combinations of four polygon vertices as
$v_i$ for $1\leq i\leq 4$.
For each such combination, we compute the intersection of $\overline{v_1v_2}$
and $\overline{v_3v_4}$ and consider the intersection as $t$. Then, we compute
$\spm(t)$.
For each bisector edge of $\spm(t)$, we obtain the roots of the two
cells incident to the bisector edge as $u_1$ and
$u_2$, respectively. Then, we find the intersection of $\overline{u_1u_2}$ and the
bisector edge as a candidate point. Since there are $O(n)$ bisector
edges, we can find $O(n)$ candidate points.
In this way, we can compute $O(n^5)$ candidate points
in $O(n^5\log n)$ time.
\end{enumerate}
\end{enumerate}

This finishes the case for $|F(s)|=1$.

\noindent{\em The case $|F(s)|\geq 2$.}
Recall that $s$ has at least one degenerate farthest point.
Let $t_1$ and $t_2$ be any two points of $F(s)$ such that $t_1$ is degenerate.

We first consider the most general case where both $t_1$ and $t_2$ are
in $\calI$. For each $i=1,2$, $t_i$ has at least three shortest paths from $s$
and $|U_{t_i}(s)|\geq 3$. Since $t_1$ is degenerate, $t_1$ has at
least four shortest paths from $s$.
Depending on whether $|U_{t_1}(s)|=3$, there are two cases.

If $|U_{t_1}(s)|>3$, let $v_{1j}$ for $1\leq j\leq 4$ be any four
vertices of $U_{t_1}(s)$. Hence, there are four shortest paths from $s$ to
$t_1$: $\pi_{u_{1j},v_{1j}}(s,t_1)$ with $1\leq j\leq 4$. For
$t_2$, let $v_{2j}$ for $1\leq j\leq 3$ be any three vertices of
$U_{t_2}(s)$, and there are three shortest paths from $s$ to $t_2$:
$\pi_{u_{2j},v_{2j}}(s,t_2)$ with $1\leq j\leq 3$.
We can form a system of equations with the lengths of the above seven shortest
paths, which give six constraints to determine $(s,t_1,t_2)$.
Correspondingly, we can compute $O(n^9)$ candidate points in $O(n^9\log
n)$ time. The algorithm is similar to the previous algorithms and we omit the
details.

If $|U_{t_1}(s)|=3$, let $v_{1j}$ for $1\leq j\leq 3$ be the three vertices of
$U_{t_1}(s)$. Since there are at least four shortest paths from $s$ to $t_1$ and
due to our general position assumption that no two polygon vertices have
more than one shortest path, there exists a vertex, say $v_{11}$, in  $U_{t_1}(s)$
such that there are at least two shortest paths from $s$ to $v_{11}$. Hence, $s$
must be on a bisector edge of the shortest path map $\spm(v_{11})$.

Let $\pi_{u_{1j},v_{1j}}(s,t_1)$ for $j=1,2,3$ be three shortest paths from $s$
to $t_1$. Let $\pi_{u_{2j},v_{2j}}(s,t_2)$ for $j=1,2,3$ be three shortest paths
from $s$ to $t_2$. We can form a system of equations with the lengths of the six
paths. These equations can give five constraints, and along with that $s$ is on
a bisector edge of $\spm(t_1)$, we can determine $(s,t_1,t_2)$ (equivalently,
one may also consider that the seven shortest paths give six
constraints). Correspondingly, we compute the candidate points as follows.

We enumerate all possible combinations of six polygon vertices as $v_{ij}$ for
$i=1,2$ and $j=1,2,3$. For each combination, we compute the overlay of
the shortest path maps of these vertices.
For each cell of the overlay, if its boundary has a portion of a bisector edge
of $\spm(v_{11})$, then based on the above system of equations on the
lengths of the six paths, we determine a constant number of triples
$(\hats,t_1,t_2)$ and each such $\hats$ is a candidate point.
Since the size of the overlay
is $O(n^2)$, we can determine $O(n^2)$ candidate points  for each
combination. Since there are $O(n^6)$ combinations,
we can compute $O(n^8)$ candidate points
in $O(n^{8}\log n)$ time.

The above computes $O(n^9)$ candidate points for the case where both
$t_1$ and $t_2$ are in $\calI$.

Other cases are very similar. For example, consider the case where $t_1\in
\calI$ and $t_2$ is on a polygonal edge $e\in E$. Comparing with the previous case,
since $t_2$ is in $E$, there will be one less constraint on the equations of
shortest path lengths, but $t_2\in e$ gives one more constraint.
All other cases are similar.
We can compute at most $O(n^9)$ candidate points in
$O(n^{9}\log n)$  for all other cases. The details are omitted.

This finishes the case for $|F(s)|\geq 2$ and thus the case for $s\in
\calI$.

\subsubsection{The Case $s\in E$.}
Let $e_s$ be the polygon edge that contains $s$.
Let $t$ be a degenerate farthest point of $s$.
Depending on whether $t$ is in $\calV$, $E$, or $\calI$, there are three cases.

\begin{enumerate}
\item
If $t\in \calV$, then there are at least two shortest \st\ paths. Further, since
$t$ is a polygon vertex, due to our general position assumption that
any two polygon vertices have no more than one shortest path, $U_s(t)$
has at least two vertices.
This implies that $s$ is in a bisector edge of $\spm(t)$.
Further, as $s$ is on $e_s$, $s$ is at the intersection of a bisector edge of
$\spm(t)$ and $e_s$. Correspondingly, we compute the candidate points as follows.

We consider all polygon vertices. For each vertex, we consider it as $t$ and
compute $\spm(t)$. For each bisector edge, if it
intersects a polygon edge, we compute the intersection as a candidate
point.  In this way, we can compute $O(n^2)$ candidate points in $O(n^2\log n)$ time.

\item
If $t\in E$, let $e_t$ be the polygon edge that contains $t$.
Since $t\in E$ and $t$ is degenerate, there are at least three shortest paths
from $s$ to $t$.
Since $t\in e_t$, $U_t(s)$ must have two vertices $v_1$
and $v_2$ that satisfy the condition of Observation \ref{obser:10}(2).
Hence, $|U_t(s)|\geq 2$. Depending on whether $|U_t(s)|\geq 3$, there are two
cases.

\begin{enumerate}
\item
If $|U_t(s)|\geq 3$, let $v_i$ with $i=1,2,3$ be any three vertices of $U_t(s)$.
Then, there are three shortest \st\ paths $\pi_{u_i,v_i}(s,t)$ for
$i=1,2,3$. We have the following
\begin{equation*}
\begin{split}
& |{tv_{1}}|+d(v_{1},u_{1})+|{u_{1}s}|
=|{tv_{2}}|+d(v_{2},u_{2})+|{u_{2}s}|
=|{tv_{3}}|+d(v_{3},u_{3})+|{u_{3}s}|.
\end{split}
\end{equation*}

The above equations provide two constraints,
and along with the constraints that $s\in e_s$ and $t\in e_t$, the
four constraints can determine $s$ and $t$.
Correspondingly, we compute the candidate points as follows.

We enumerate all possible combinations of three vertices as $v_i$ for $i=1,2,3$
and a polygon edge as $e_t$.
For each combination, we compute the overlay of the shortest path maps
of $v_1$, $v_2$, and $v_3$. For each cell $C$ of the overlay,
if it has an edge that is a portion of a
polygon edge, then we consider the polygon edge as $e_s$ and obtain the three roots of
the three shortest path maps as $u_i$ for $i=1,2,3$.
Next, using the four constraints mentioned above, we can determine a
constant number of tuples $(\hats,t)$ and each such $\hats$ is a
candidate point. Since the size of the overlay is $O(n^2)$, for each combination we
can obtain $O(n^2)$ candidate points in $O(n^2\log n)$ time.
Since there are $O(n^4)$ combinations, we can
compute $O(n^6)$ candidate points in $O(n^6\log n)$ time.

\item
If $|U_t(s)|=2$, let $v_1$ and $v_2$ be the two vertices of $U_t(s)$. Since
there are at least three shortest \st\ paths and due to our general position
assumption that no two polygon vertices have more than one shortest path, one of
$v_1$ and $v_2$, say $v_1$, must have at least two shortest paths to $s$ such
that $|U_s(v_1)|\geq 2$. Hence, $s$ is on a bisector edge of $\spm(v_1)$. Since
$s\in e_s$, $s$ is the intersection of $e_s$ and a bisector edge of $\spm(v_1)$.
Correspondingly, we can compute the candidate points as follows.

We enumerate all polygon vertices. For each vertex, we consider it as $v_1$ and
compute its shortest path map $\spm(v_1)$. Then, for each
intersection between a bisector edge of $\spm(v_1)$ and a polygon
edge, we consider it as a candidate point.
Since the size of $\spm(v_1)$ is $O(n)$, there are $O(n)$ such candidate
points.
In this way, we can compute at most $O(n^2)$ candidate points in $O(n^2\log n)$
time.
\end{enumerate}

\item
If $t\in \calI$, then there are at least four shortest \st\ paths, and $|U_t(s)|\geq 3$ by Observation \ref{obser:10}. Depending on whether $|U_t(s)|\geq 4$, there two cases.

\begin{enumerate}
\item
If $|U_t(s)|\geq 4$, then let $v_i$ with $1\leq i\leq 4$ be any four vertices of $U_t(s)$. Then, there are four shortest \st\ paths $\pi_{u_i,v_i}(s,t)$ for $1\leq i\leq 4$. We have the following
\begin{equation*}
\begin{split}
& |{tv_{1}}|+d(v_{1},u_{1})+|{u_{1}s}|
=|{tv_{2}}|+d(v_{2},u_{2})+|{u_{2}s}|\\
&=|{tv_{3}}|+d(v_{3},u_{3})+|{u_{3}s}|
=|{tv_{4}}|+d(v_{4},u_{4})+|{u_{4}s}|.
\end{split}
\end{equation*}

The above equations provide three constraints,
and along with the constraints that $s\in e_s$, the
four constraints can determine $s$ and $t$.
Correspondingly, we can compute $O(n^6)$ candidate points in $O(n^6\log n)$ time. The algorithm is similar as before and we omit the details.

\item
If $|U_t(s)|=3$, let $v_i$ with $i=1,2,3$ be the three vertices of $U_t(s)$. Since there are at least four shortest \st\ paths and due to our general position assumption that no two polygon vertices have more than one shortest path, one of $v_i$ with $i=1,2,3$, say $v_1$, must have at least two shortest paths to $s$ such that $|U_s(v_1)|\geq 2$. Hence, $s$ is on a bisector edge of $\spm(v_1)$. Since $s\in e_s$, $s$ is the intersection of $e_s$ and a bisector edge of $\spm(v_1)$.
Correspondingly, by using the same approach as the above case 2(b), we can compute $O(n^2)$ candidate points in $O(n^2\log n)$ time.
\end{enumerate}
\end{enumerate}

This finishes the case for $s\in E$.

This also finishes our algorithms for computing candidate points for the
degenerate case. In summary, we can compute a total of $O(n^9)$ candidate
points in $O(n^9\log n)$ time such that $s$ is one of these candidate
points.

\section{Computing the Geodesic Centers}
\label{sec:center}

In this section, we find all geodesic centers from the candidate point
set $S$. Let $S'$ denote the set of candidate points for all cases
other than the four dominating cases. Recall that $S_d$ is the
set of candidate points for the four dominating cases.
Hence, $S=S'\cup S_d$.
As discussed in Section \ref{sec:candidates}, $|S'|=O(n^{10})$ and we can
find all geodesic centers in $S'$ in $O(n^{11}\log n)$ time by
computing their shortest path maps. In the
following, we focus on finding all geodesic centers in $S_d$.

We first remove all points from $S_d$ that are also in $S'$,
which can be done in $O(n^{11}\log n)$ time (e.g., by first sorting
these points by their coordinates).
Then, for any point $s\in S_d$, if $s$ is a geodesic center,
$s$ does not have any degenerate farthest point since
otherwise $s$ was also in $S'$ and thus would have already been
removed from $S_d$.

Recall that each point $s$ of $S_d$ is a valid candidate point and we have
maintained its path information for $s$ (in particular,
the value $d(s)$).

We first perform the following {\em duplication-cleanup} procedure: for each point
$s\in S_d$, if there are many copies of $s$, we only keep the
one with the largest value $d(s)$ (if more than one copy has the
largest value, we keep an arbitrary one).
This procedure can be done in $O(n^{11}\log n)$ time (e.g., by
first sorting all points of $S_d$ by their coordinates).
According to our algorithm for computing the candidate points of $S_d$,
we have the following observation.

\begin{observation}\label{obser:70}
After the duplication-cleanup procedure, for any point $s\in
S_d$, if $s$ is a geodesic center, then
$d_{\max}(s)=d(s)$.
\end{observation}
\begin{proof}
For differentiation, we use $S_d$ to denote the original set of $S_d$ before the
duplication-cleanup procedure and use $S_d'$ to denote the set after
the procedure. Note that $S_d$ and $S_d'$ contain the same physical
points and the difference is that for each point of $S_d'$, $S_d$ may
contain multiple copies of the point, and each copy is associated with different path
information. Recall that none of the point of $S_d$ is in $S'$.

Consider any point $s\in S_d'$ and suppose $s$ is a geodesic
center.  According to our algorithm for computing candidate
points of $S_d$, there must be a copy of $s$ in $S_d$ for which we
have maintained its path information that includes a farthest point $t$ of
$s$ and $d(s)=d(s,t)$ by Observation
\ref{obser:60}. Since $t$ is a farthest point of $s$,
$d(s,t)=d_{\max}(s)$.
Therefore, there exists a copy of $s$ in $S_d$
with $d(s)=d_{\max}(s)$. Let $s'$ denote any other copy of
$s$ in $S_d$. Again, by Observation \ref{obser:60}, $d(s')$ is
equal to the shortest path distance from $s'$ to a point $t'$.
Hence, $d(s')\leq d_{\max}(s')=d_{\max}(s)=d(s)$.

According to our duplication-cleanup procedure,
if $d(s')< d(s)$, then the copy $s'$ will be removed
from $S_d'$. Otherwise, either copy may be removed, but in either
case, the $d(\cdot)$ value of the remaining copy is always equal to $d_{\max}(s)$.
Hence, the observation follows.
\end{proof}



Recall that all points of $S_d$ are in $\calI$.

In the following, we give a {\em pruning algorithm} that can eliminate most of the
points from $S_d$ such that none of these eliminated points is a geodesic center
and the number of remaining points of $S_d$ is $O(n^{10})$.
Our pruning algorithm relies on the property that each candidate point $s$ of
$S_d$ is valid. Specifically, if $s$ is computed for the dominating
case $(3,0,0)$, then $s$ is associated with
the following path information $d(s)$, $t_i$, $v_{ij}$, and $u_{ij}$ for $1\leq
i\leq 3$ and $1\leq j\leq 3$, such that Observation \ref{obser:60} holds (i.e., $s$ is $\hats$ and each $t_i$ is $\hatt_i$). For
other three dominating cases (e.g., $(2,1,0)$, $(1,2,0)$, and $(0,3,0)$), there are similar
properties.
By using these properties, we have the following lemma.

\begin{lemma}\label{lem:170}
Let $s$ be any point in $S_d$. If $s$ is in the interior of a
cell or an edge of $\spmed$, then for any other point $s'$ in the interior of the
same cell or edge of $\spmed$, it holds that $d_{\max}(s')>d(s)$.
\end{lemma}
\begin{proof}
The proof uses similar techniques as that for Lemma \ref{lem:30}. We only prove
the case for $s$ being a candidate point for Case $(3,0,0)$, and other cases
are similar.

Recall that $s$ is associated with
the following path information $d(s)$, $t_i$, $v_{ij}$, and
$u_{ij}$ for $1\leq i\leq 3$ and $1\leq j\leq 3$, such that
Observation \ref{obser:60} holds.

If $s'$ be any other point in the same cell or edge of $\spmed$ that contains
$s$. Hence, $\spm(s)$ and $\spm(s')$ are topologically equivalent.
Consider any $t_i$ with $1\leq i\leq 3$.
By Observation \ref{lem:60}, for each $j=1,2,3$, $\pi_{u_{ij},v_{ij}}(s,t)$
is a shortest paths from $s$ to $t_i$. Since $\spm(s)$ and $\spm(s')$
are topologically equivalent, $\spm(s')$ has one and only one vertex $t_i'$
corresponding to $t_i$ and for each $j=1,2,3$,
$\pi_{u_{ij},v_{ij}}(s',t_i')$ is a shortest path from $s'$ to
$t_i'$.

Let $r$ be the direction from $s$ to $s'$. By Observation \ref{obser:60}, $\piR(s,t_1)\cap \piR(s,t_2)\cap \piR(s,t_3) = \emptyset$.
Therefore, $r$ is not in $\piR(s,t_i)$ for some $i$ with $1\leq i\leq 3$. Without loss of generality, we assume $r$ is not in $\piR(s,t_1)$.

Suppose we move $s$ along the direction $r$ with unit
speed and move $t_1$ to $t_1'$ with speed $|t_1t_1'|/|s s'|$. Hence, when $s$ arrives at $s'$, $t_1$ arrives at $t_1'$ simultaneously.
Since $r$ is not in $\piR(s,t_1)$, there must be a path $\pi_{u_{1j},v_{1j}}(s,t_1)$ for some $j$ with $1\leq j\leq 3$ such that
$d'_{u_{1j},v_{1j}}(s,t_1)<0$ does not hold. Without loss of generality, we
assume $j=1$. In other words, either $d'_{u_{11},v_{11}}(s,t_1)>0$ or $d'_{u_{11},v_{11}}(s,t_1)=0$.

\begin{enumerate}
\item
If $d'_{u_{11},v_{11}}(s,t_1)>0$, then since $d''_{u_{11},v_{11}}(s,t_1)\geq 0$ always holds, we have
$d_{u_{11},v_{11}}(s,t_1)<d_{u_{11},v_{11}}(s',t_1')$. Since
$d(s)=d_{u_{11},v_{11}}(s,t_1)$ and $d_{\max}(s')\geq d(s',t_1')=
d_{u_{11},v_{11}}(s',t_1')$, we obtain that $d_{\max}(s')>d(s)$, which proves the lemma.
\item
If $d'_{u_{11},v_{11}}(s,t_1)=0$, recall that $d''_{u_{11},v_{11}}(s,t_1)\geq 0$ always holds. If $d''_{u_{11},v_{11}}(s,t_1)>0$, then we again obtain
$d_{u_{11},v_{11}}(s,t_1)<d_{u_{11},v_{11}}(s',t_1')$, and consequently,  $d_{\max}(s')>d(s)$.

Otherwise, $d''_{u_{11},v_{11}}(s,t_1)=0$. We claim that this case
cannot happen. Indeed, suppose to the contrary that
$d''_{u_{11},v_{11}}(s,t_1)=0$. Then, we have
$d_{u_{11},v_{11}}(s,t_1)=d_{u_{11},v_{11}}(s',t_1')$.
Since $\pi_{u_{1j},v_{1j}}(s',t_1')$ for $j=2,3$ is also a shortest path from
$s$ to $t_1'$, we have
$d_{u_{11},v_{11}}(s',t_1')=d_{u_{12},v_{12}}(s',t_1')=d_{u_{13},v_{13}}(s',t_1')$,
and thus,
$d_{u_{12},v_{12}}(s,t_1)=d_{u_{12},v_{12}}(s',t_1')$ and
$d_{u_{13},v_{13}}(s,t_1)=d_{u_{13},v_{13}}(s',t_1')$.
This is only possible if $d'_{u_{1j},v_{1j}}(s,t_1)=0$  and
$d''_{u_{1j},v_{1j}}(s,t_1)=0$ for both $j=2,3$.
However, by Observation \ref{obser:60}, $t_1$ is in the interior of the triangle
$\triangle v_{11}v_{12}v_{13}$; as discussed earlier (i.e., the discussions for Fig.~\ref{fig:specialpivots}), this implies that
it is not possible to have $d'_{u_{1j},v_{1j}}(s,t_1)=0$  and $d''_{u_{1j},v_{1j}}(s,t_1)=0$ for all $j=1,2,3$,
incurring contradiction.
\end{enumerate}

This completes the proof of the lemma.
\end{proof}

\begin{lemma}\label{lem:180}
For any two points $s_1$ and $s_2$ of $S_d$ that are in the interior of the same cell or the same edge of $\spmed$, if $d(s_1)<d(s_2)$, then $s_1$ cannot be a geodesic
center, and if $d(s_1)=d(s_2)$, then neither $s_1$ nor $s_2$ is a geodesic
center.
\end{lemma}
\begin{proof}
Consider any two points $s_1$ and $s_2$ of $S_d$ that are in the interior of the same cell or the same edge of $\spmed$.
By Lemma \ref{lem:170}, $d_{\max}(s_1)> d(s_2)$ and $d_{\max}(s_2)>
d(s_1)$.

If $d(s_1)<d(s_2)$, then we obtain $d_{\max}(s_1)>d(s_2)>d(s_1)$.
Thus, $s_1$ cannot be a geodesic center since otherwise $d(s_1)$ would be equal
to $d_{\max}(s_1)$ by Observation \ref{obser:70}.

Similarly, if $d(s_1)=d(s_2)$, then we have both $d_{\max}(s_1)> d(s_2) =
d(s_1)$ and $d_{\max}(s_2)> d(s_1)=d(s_2)$.
Thus, neither $s_1$ nor $s_2$ is a geodesic
center.
\end{proof}

Based on Lemma \ref{lem:180}, our pruning algorithm for
$S_d$ works as follows. For each point $s$ of $S_d$,
we determine the cell, edge, or vertex of $\spmed$ that contains $s$ in its interior, which can be done in $O(\log n)$ time by using a point location data
structure \cite{ref:EdelsbrunnerOp86,ref:KirkpatrickOp83} with $O(n^{10})$ time and space preprocessing on $\spmed$. For each edge or
cell, let $S_d'$ be the set of points of $S_d$ that are contained in its interior.
We find the point $s$ of $S_d'$ with the largest value $d(s)$. If there are more
than one such point in $S_d'$, we remove all points of $S_d'$ from $S_d$; otherwise,
remove all points of $S'_d$ except $s$ from $S_d$. By Lemma \ref{lem:180}, none of the points of $S_d$ that are removed above is a geodesic center. After the above pruning algorithm, $S_d$ contains at most one point in the interior of each cell, edge, or vertex of $\spmed$. Hence, $|S_d|=O(|\spmed|)$. Since
$|\spmed|=O(n^{10})$ \cite{ref:ChiangTw99}, we obtain $|S_d|=O(n^{10})$.

\begin{theorem}\label{theo:20}
All geodesic centers of $\calP$ can be computed in $O(n^{11}\log n)$ time.
\end{theorem}


\section*{Acknowledgment}
The author wishes to thank Yan Sun for the discussions on proving the $\pi$-range property.

\bibliographystyle{plain}
\bibliography{reference}

\begin{thebibliography}{10}

\bibitem{ref:AhnA15}
H.-K. Ahn, L.~Barba, P.~Bose, J.-L.~De Carufel, M.~Korman, and E.~Oh.
\newblock A linear-time algorithm for the geodesic center of a simple polygon.
\newblock In {\em Proc. of the 31st Annual Symposium on Computational Geometry
  (SoCG)}, pages 209--223, 2015.

\bibitem{ref:AsanoCo85}
T.~Asano and G.~Toussaint.
\newblock Computing the geodesic center of a simple polygon.
\newblock Technical Report SOCS-85.32, McGill University, Montreal, Canada,
  1985.

\bibitem{ref:BaeTh13}
S.W. Bae, M.~Korman, and Y.~Okamoto.
\newblock The geodesic diameter of polygonal domains.
\newblock {\em Discrete and Computational Geometry}, 50:306--329, 2013.

\bibitem{ref:BaeCo14CCCG}
S.W. Bae, M.~Korman, and Y.~Okamoto.
\newblock Computing the geodesic centers of a polygonal domain.
\newblock In {\em Proc. of the 26th Canadian Conference on Computational
  Geometry (CCCG)}, 2014.
\newblock Journal version published online in {\em Computational Geometry:
  Theory and Applications}, 2015.

\bibitem{ref:BaeCo15}
S.W. Bae, M.~Korman, Y.~Okamoto, and H.~Wang.
\newblock Computing the {$L_1$} geodesic diameter and center of a simple
  polygon in linear time.
\newblock {\em Computational Geometry: Theory and Applications}, 48:495--505,
  2015.

\bibitem{ref:ChazelleA82}
B.~Chazelle.
\newblock A theorem on polygon cutting with applications.
\newblock In {\em Proc. of the 23rd Annual Symposium on Foundations of Computer
  Science}, pages 339--349, 1982.

\bibitem{ref:ChiangTw99}
Y.-J. Chiang and J.S.B. Mitchell.
\newblock Two-point {Euclidean} shortest path queries in the plane.
\newblock In {\em Proc. of the Annual ACM-SIAM Symposium on Discrete
  Algorithms}, pages 215--224, 1999.

\bibitem{ref:DjidjevAn92}
H.N. Djidjev, A.~Lingas, and J.-R. Sack.
\newblock An {$O(n\log n)$} algorithm for computing the link center of a simple
  polygon.
\newblock {\em Discrete and Computational Geometry}, 8:131--152, 1992.

\bibitem{ref:EdelsbrunnerOp86}
H.~Edelsbrunner, L.~Guibas, and J.~Stolfi.
\newblock Optimal point location in a monotone subdivision.
\newblock {\em SIAM Journal on Computing}, 15(2):317--340, 1986.

\bibitem{ref:HalperinNe94}
D.~Halperin and M.~Sharir.
\newblock New bounds for lower envelopes in three dimensions, with applications
  to visibility in terrains.
\newblock {\em Discrete and Computational Geometry}, 12:313--326, 1994.

\bibitem{ref:HershbergerMa97}
J.~Hershberger and S.~Suri.
\newblock Matrix searching with the shortest-path metric.
\newblock {\em SIAM Journal on Computing}, 26(6):1612--1634, 1997.

\bibitem{ref:HershbergerAn99}
J.~Hershberger and S.~Suri.
\newblock An optimal algorithm for {Euclidean} shortest paths in the plane.
\newblock {\em SIAM Journal on Computing}, 28(6):2215--2256, 1999.

\bibitem{ref:KeAn89}
Y.~Ke.
\newblock An efficient algorithm for link-distance problems.
\newblock In {\em Proc. of the 5th Annual Symposium on Computational Geometry},
  pages 69--78, 1989.

\bibitem{ref:KirkpatrickOp83}
D.~Kirkpatrick.
\newblock Optimal search in planar subdivisions.
\newblock {\em SIAM Journal on Computing}, 12(1):28--35, 1983.

\bibitem{ref:MitchellGe00}
J.S.B. Mitchell.
\newblock {\em {\em Geometric shortest paths and network optimization, in}
  Handbook of Computational Geometry, {\em J.-R Sack and J. Urrutia (eds.)}},
  pages 633--702.
\newblock Elsevier, Amsterdam, the Netherlands, 2000.

\bibitem{ref:NilssonCo91}
B.J. Nilsson and S.~Schuierer.
\newblock Computing the rectilinear link diameter of a polygon.
\newblock In {\em Proc. of the International Workshop on Computational Geometry
  - Methods, Algorithms and Applications}, pages 203--215, 1991.

\bibitem{ref:NilssonAn96}
B.J. Nilsson and S.~Schuierer.
\newblock An optimal algorithm for the rectilinear link center of a rectilinear
  polygon.
\newblock {\em Computational Geometry}, 6:169--194, 1996.

\bibitem{ref:PollackCo89}
R.~Pollack, M.~Sharir, and G.~Rote.
\newblock Computing the geodesic center of a simple polygon.
\newblock {\em Discrete and Computational Geometry}, 4(1):611--626, 1989.

\bibitem{ref:SchuiererAn96}
S.~Schuierer.
\newblock An optimal data structure for shortest rectilinear path queries in a
  simple rectilinear polygon.
\newblock {\em International Journal of Compututational Geometry and
  Applications}, 6:205--226, 1996.

\bibitem{ref:SuriCo89}
S.~Suri.
\newblock Computing geodesic furthest neighbors in simple polygons.
\newblock {\em Journal of Computer and System Sciences}, 39:220--235, 1989.

\end{thebibliography}

%


\end{document}